\pdfoutput=1 % use direct-to-pdf rendering so we can include pdf figures.
\PassOptionsToPackage{backref=page,bookmarksopen=true}{hyperref}

\documentclass[11pt,letterpaper, final]{article}

\usepackage[margin=1in,nohead]{geometry}

\usepackage[utf8]{inputenc}
\usepackage{lmodern}        % prettier font, fixes \textasciitilde
\usepackage[T1]{fontenc}
\usepackage{microtype}
\usepackage{stmaryrd}

\usepackage{authblk}

\usepackage{amsmath}
\usepackage{amsthm}
\usepackage{amsfonts}
\usepackage{amssymb}
\usepackage{mathtools}
\usepackage{thmtools}
\usepackage{thm-restate}
\usepackage{xspace}
\usepackage{textpos}

\usepackage{algorithm}
\usepackage[noEnd,commentColor=green!60!black]{algpseudocodex} % newer version of algpseudocode
\algnewcommand\LeftComment[1]{\State\(\triangleright\)~\textit{\textcolor{green!60!black}{#1}}}

\date{}

\usepackage{float}

\usepackage{tikz}

\usepackage{caption}
\usepackage{subcaption}

\usepackage{enumitem}

\setlist[itemize]{topsep=0pt,itemsep=0pt}%,partopsep=0pt,parsep=0pt}

\usepackage{draftwatermark}\SetWatermarkLightness{0.95}\SetWatermarkScale{1}
\usepackage[%
  textwidth=2.125cm,%
  obeyFinal%
]{todonotes}

\usepackage[hypertexnames=false]{hyperref} % last, except for packages that use it...
\usepackage[capitalize]{cleveref}
\usepackage[%
]{showkeys} % for debugging, after hyperref to show labels for bibliography.

\crefname{claim}{Claim}{Claims}
\Crefname{claim}{Claim}{Claims}

\DeclareFontShape{T1}{lmr}{m}{scit}{<->ssub * lmr/m/scsl}{}
\DeclareFontShape{T1}{lmr}{bx}{sc}{<->ssub * lmr/bx/n}{}

\declaretheorem[style=plain,name={Theorem}]{theorem}
\declaretheorem[style=plain,name={Claim},sibling=theorem]{claim}
\declaretheorem[style=plain,name={Lemma},sibling=theorem]{lemma}

\DeclarePairedDelimiter{\abs}{\lvert}{\rvert}

\DeclarePairedDelimiter{\ceil}{\lceil}{\rceil}
\DeclarePairedDelimiter{\floor}{\lfloor}{\rfloor}

\DeclarePairedDelimiter{\set}{\lbrace}{\rbrace}
\newcommand{\suchthat}{\mathrel{}\mathclose{}\ifnum\currentgrouptype=16\middle\fi\vert\mathopen{}\mathrel{}}

\DeclarePairedDelimiter{\paren}{\lparen}{\rparen}

\DeclareMathOperator{\Oh}{\mathcal{O}}
\DeclareMathOperator{\BigO}{\mathcal{O}}
\DeclareMathOperator{\BigOTilde}{\tilde{\BigO}}

\newcommand{\Weight}{\omega}

\newcommand{\bitor}{\mathbin{\mathrm{or}}}
\newcommand{\lmax}{\ell_{\max}}
\newcommand{\bnd}{\partial}
\newcommand{\LL}{\mathcal{L}}
\newcommand{\cover}{\mathrm{cover}}
\newcommand{\argcover}{\mathrm{argcover}}
\newcommand{\firstedge}{\mathrm{firstedge}}
\newcommand{\Insert}{\Call{Insert}}
\newcommand{\Delete}{\Call{Delete}}
\newcommand{\AreBiconnected}{\Call{AreBiconnected}}
\newcommand{\NextCutVertex}{\Call{NextCutVertex}}
\newcommand{\UncoverPath}{\Call{UncoverPath}}
\newcommand{\UniformUncover}{\Call{UniformUncover}}
\newcommand{\Select}{\Call{Select}}
\newcommand{\UpdateMark}{\Call{UpdateMark}}
\newcommand{\FindSize}{\Call{FindSize}}
\newcommand{\Cut}{\Call{Cut}}
\newcommand{\coverfrom}{\cover^{\mathrm{from}}}
\newcommand{\covertop}{\cover^{\mathrm{top}}}
\newcommand{\size}{\mathrm{size}}
\newcommand{\meet}{\mathrm{meet}}
\newcommand{\dist}{\mathrm{dist}}
\newcommand{\interface}{\mathrm{interface}}
\newcommand{\cluster}{\mathrm{cluster}}

\newcommand{\smallestpoint}{\mathrm{smallestpoint}}

\newcommand{\Link}{\Call{Link}}
\newcommand{\Expose}{\Call{Expose}}
\newcommand{\TransientExpose}{\Call{TransientExpose}}
\newcommand{\TransientUnexpose}{\Call{TransientUnexpose}}
\newcommand{\Create}{\Call{Create}}
\newcommand{\Destroy}{\Call{Destroy}}
\newcommand{\Merge}{\Call{Merge}}
\newcommand{\Split}{\Call{Split}}
\newcommand{\Connected}{\Call{Connected}}
\newcommand{\CoverLevel}{\Call{CoverLevel}}
\newcommand{\MinCoveredPair}{\Call{MinCoveredPair}}
\newcommand{\Cover}{\Call{Cover}}
\newcommand{\LocalUncover}{\Call{LocalUncover}}
\newcommand{\Mark}{\Call{Mark}}
\newcommand{\Unmark}{\Call{Unmark}}
\newcommand{\FindFirstReach}{\Call{FindFirstReach}}
\newcommand{\FindStrongReach}{\Call{FindStrongReach}}
\newcommand{\PromoteEdge}{\Call{PromoteEdge}}
\newcommand{\FindNextEvent}{\Call{FindNextEvent}}
\newcommand{\Swap}{\Call{Swap}}

\newcommand{\Unzip}{\Call{Unzip}}
\newcommand{\Zip}{\Call{Zip}}
\newcommand{\SetWeight}{\Call{SetWeight}}
\newcommand{\Level}{\Call{Level}}
\newcommand{\UpdateCounters}{\Call{UpdateCounters}}
\newcommand{\SumCounters}{\Call{SumCounters}}
\newcommand{\UpdateMarks}{\Call{UpdateMarks}}
\newcommand{\OrMarks}{\Call{OrMarks}}
\newcommand{\FindMarked}{\Call{FindMarked}}
\newcommand{\FindStrongMarked}{\Call{FindStrongMarked}}
\newcommand{\LongZip}{\Call{LongZip}}
\newcommand{\LongUnzip}{\Call{LongUnzip}}
\newcommand{\SelectedLevel}{\Call{SelectedLevel}}

\newcommand{\FindInPath}{\Call{FindInPath}}

\newcommand{\RefineClusterEdgeContainer}{\Call{RefineClusterEdgeContainer}}
\newcommand{\RefinePathClusterContainer}{\Call{RefinePathClusterContainer}}

\newcommand{\CasePointToPointPoint}{\hyperref[fig:toptree-cases]{(1)}\xspace}
\newcommand{\CasePointToPathPoint}{\hyperref[fig:toptree-cases]{(2)}\xspace}
\newcommand{\CasePointToPointPath}{\CasePointToPathPoint}
\newcommand{\CasePath}{\hyperref[fig:toptree-cases]{(3)}\xspace}

\newcommand{\OnCreate}{\Call{OnCreate}}
\newcommand{\OnMerge}{\Call{OnMerge}}
\newcommand{\Clean}{\Call{Clean}}

\newcommand{\vectorsplice}[3]{[{#1} \,:\, {#2} \,:\, {#3}]}
\newcommand{\matrixsplice}[3]{\llbracket {#1} \,:\, {#2} \,:\, {#3} \rrbracket}
\newcommand{\matrixsum}{\mathrm{sum}}
\newcommand{\uppersum}{\mathrm{uppersum}}
\newcommand{\addvector}{\mathrm{addvector}}

\newcommand{\clustercnt}{\mathrm{clustercnt}}
\newcommand{\clustercntvec}{\mathbf{clustercnt}}
\newcommand{\totalcnt}{\mathrm{totalcnt}}
\newcommand{\totalcntvec}{\mathbf{totalcnt}}
\newcommand{\diagcnt}{\mathrm{diagcnt}}
\newcommand{\diagcntvec}{\mathbf{diagcnt}}

\newcommand{\bmarks}{\mathrm{bmarks}}
\newcommand{\bmarksvec}{\mathbf{bmarks}}
\newcommand{\ismarked}{\mathrm{ismarked}}
\newcommand{\ismarkedvec}{\mathbf{ismarked}}
\newcommand{\totalmarks}{\mathrm{totalmarks}}
\newcommand{\totalmarksvec}{\mathbf{totalmarks}}

\newcommand{\diagmarksvec}{\mathbf{diagmarks}}

\newcommand{\ncost}{\mathrm{ncost}}
\newcommand{\Tc}{\mathcal{T}}
\newcommand{\Uc}{\mathcal{U}}

\newcommand{\Xsel}{X_{\mathrm{sel}}}

\algnewcommand{\Assert}{\textbf{assert }}
\algnewcommand{\Break}{\textbf{break}}
\algnewcommand{\Downto}{\textbf{ downto }}
\algnewcommand{\False}{\textbf{false}}
\algnewcommand{\True}{\textbf{true}}
\algnewcommand{\Not}{\textbf{not }}

\newcommand{\wn}[2][]{\todo[color=yellow!50,#1]{{\textbf{W:} #2}}}

\newcommand{\ms}[2][]{\todo[color=red!60,#1]{{\textbf{M:} #2}}}
\newcommand{\wninline}[2][]{\todo[inline, color=yellow!50,#1]{{\textbf{W:} #2}}}
\newcommand{\jhinline}[2][]{\todo[inline, color=blue!25,#1]{{\textbf{J:} #2}}}

\newcommand{\msinline}[2][]{\todo[inline, color=red!60,#1]{{\textbf{M:} #2}}}

\newcommand{\nologlogs}{{2}}

\newcommand{\mailto}[1]{\texttt{<\href{mailto:#1}{#1}>}}
\author[1]{Jacob Holm\thanks{Supported by the VILLUM Foundation grant 54451 ``Basic Algorithms Research Copenhagen (BARC)''.}}
\affil[1]{BARC, University of Copenhagen \mailto{jaho@di.ku.dk}}
\author[2]{Wojciech Nadara\thanks{ Supported by Independent Research Fund Denmark grant 2020-2023 (9131-00044B) “Dynamic Network Analysis” (while being employed in Denmark) and European Union’s Horizon 2020 research and
		innovation programme, grant agreement No. 948057 — BOBR (while being employed in Poland)}}
\affil[2]{University of Warsaw and Technical University of Denmark \mailto{w.nadara@mimuw.edu.pl} }
\author[3]{Eva Rotenberg\thanks{This work was supported by Independent Research Fund Denmark grant 2020-2023 (9131-00044B) ``Dynamic Network Analysis''. The 3rd author also thanks grant VIL37507 ``Efficient Recomputations for Changeful Problems''.}}
\affil[3]{Technical University of Denmark \mailto{erot@dtu.dk}}
\author[4]{Marek Sokołowski\thanks{This work was created during this author's work at the Institute of Informatics of the University of Warsaw.
  The work of this author on this manuscript is a part of a project that has received funding from the European Research Council (ERC), grant agreement No.\ 948057 -- BOBR.}}
\affil[4]{Max Planck Institute for Informatics \mailto{msokolow@mpi-inf.mpg.de}}

\title{Fully dynamic biconnectivity in $\BigOTilde(\log^2n)$ time}

\begin{document}
\maketitle
\thispagestyle{empty}

\begin{textblock}{20}(-1.83, 8.1)
	\includegraphics[width=40px]{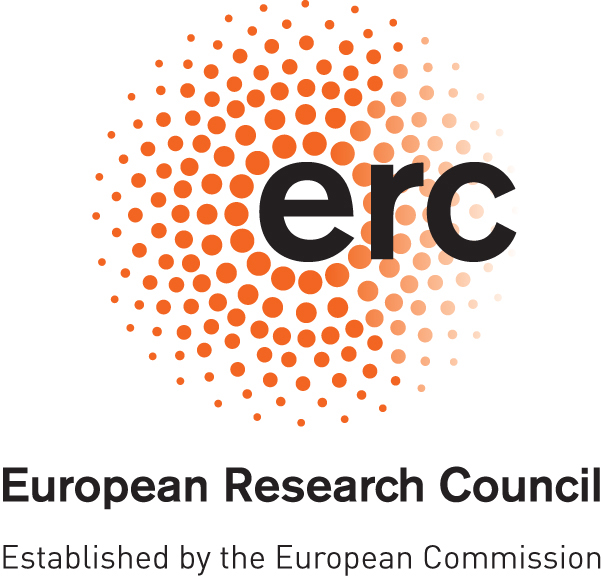}%
\end{textblock}
\begin{textblock}{20}(-2.05, 8.5)
	\includegraphics[width=60px]{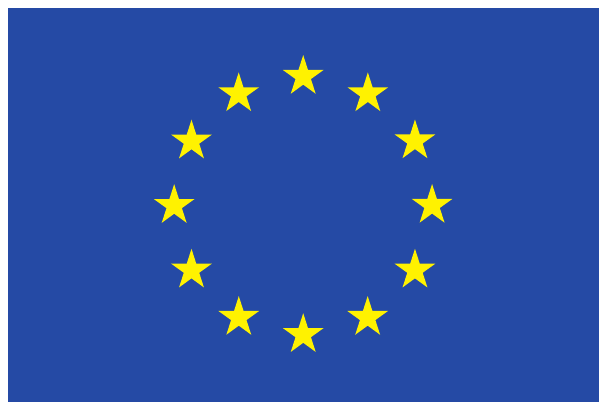}%
\end{textblock}

\begin{abstract}
We present a deterministic fully-dynamic data structure for maintaining information about the cut-vertices in a graph; i.e. the vertices whose removal would disconnect the graph. Our data structure supports insertion and deletion of edges, as well as queries to whether a pair of connected vertices are either biconnected, or can be separated by a cutvertex, and in the latter case we support access to separating cutvertices. All update operations are supported in amortized $\BigO(\log ^2 n \log^{\nologlogs}\log  n)$ time, and queries take worst-case $\BigO(\log n \log ^{\nologlogs}\log n)$ time. Note that these time bounds match the current best for deterministic dynamic connectivity up to $\log\log n$ factors. 

The previous best algorithm for biconnectivity had an update time of $\BigO(\log ^4 n \log \log n)$ by Thorup [STOC'00], based on the $\BigO(\log ^5 n)$ algorithm by Holm, de Lichtenberg, and Thorup [STOC'98].

We obtain our improved running time by a series of reductions from the original problem into well-defined data structure problems. While we do indeed apply the well-known techniques for improving running time of two-edge connectivity [STOC'00, SODA'18], surprisingly, these techniques alone do not lead to an update time of $\BigOTilde(\log^3 n)$, let alone the $\BigOTilde(\log^2 n)$ we give as a final result. 

Our contributions include a formally defined \emph{transient expose} operation, which can be thought of as a cheaper read-only expose operation on a top tree. For each vertex in the graph, we maintain a data structure over its neighbors, and in this data structure we apply biasing (twice) to save an $\BigOTilde(\log n)$ factor (twice, so two $\BigOTilde(\log n)$ factors). One of these biasing techniques is a new, simple biased disjoint sets data structure, which may be of independent interest. 
Moreover, in this neighborhood data structure, we facilitate that the vertex can select two VIP neighbors that get special treatment, corresponding to its potentially two neighbors on an exposed path, improving an otherwise $\log n$-time operation down to constant time. It is this combination of VIP neighbors with the transient expose operation that saves an $\BigOTilde(\log n)$-factor from another bottle neck. 

Combining these technical contributions with the well-known techniques for two-edge connectivity [STOC'00, SODA'18], we obtain the desired update times of $\BigO(\log^2 n \log^2 \log n)$. The near-linear query time follows directly from the usage of transient expose.

%Our final new technical contribution is a new biased disjoint sets data structure, which may be of independent interest. Together with our other ideas, this biased disjoint sets data structure is necessary in order to improve the final $\log n$ factor, i.e. from near log cubed to near log squared.

%OUr  contrib: transient expose (cheaper read-only expose in a top tree), formalising neighborhood trees, og bygger en ordentlig (biased!) datastructure. Selecting two "VIP" neighbors of a vertex makes another log become a constant. 

%by splitting the problem of biconnectivity into smaller, more well-defined parts, and applying techniques for two-edge connectivity [STOC'00, SODA'18] to each part. Surprisingly, these techniques alone do not lead to an update time of $\BigOTilde(\log^3 n)$, let alone the $\BigOTilde(\log^2 n)$ we give as a final result.

\end{abstract}

\newpage

\setcounter{page}{1}
\section{Introduction}

Graphs are an important discrete mathematical model for networks, and are useful for analysing networks and answering questions such as connectivity. In dynamic graph algorithms, we seek to efficiently update our analysis of a changing graph, in order to promptly answer questions about its properties. Historically, one of the first problems to be studied for dynamic graphs was connectivity~\cite{Frederickson85}, which has since received ample attention~\cite{Eppstein97,Henzinger:1999,HeTh97,HolmDeLichtenbergThorup,Thorup00,Eppstein:2003,patrascu06,Lacki15,ChuzhoyGLNPS20,HuangHKPT23}. A problem that is very related to whether a graph is connected is whether --- and where --- it is close to being disconnected, in the sense that the removal of one edge or one vertex would disconnect the graph. These problems, of $2$-edge connectivity~\cite{WestbrookT92,Frederickson97,Eppstein99,HolmDeLichtenbergThorup,Henzinger97,Holm18a,HolmR21}, %todo: cite incremental stuff
and $2$-vertex connectivity, also known as biconnectivity~\cite{WestbrookT92,Rauch92,Rauch94,Henzinger95,HenzingerK95,Eppstein99,HolmDeLichtenbergThorup,Henzinger00}, %todo: more cite
also received attention, even in restricted graph classes~\cite{PoutreW98,HolmR20,HolmHR23,HolmR24}.

Biconnectivity can be described as a property of a vertex pair: two vertices are biconnected if there exist two vertex disjoint paths connecting them. Equivalently, by Menger's Theorem~\cite{menger}, there is no cut-edge or cut-vertex separating them. Another description of biconnectivity is as an equivalence relation on edges, where a pair of different edges are related if they are both on some cycle. Intuitively, the cutvertices partition the edge set into biconnected components under that equivalence relation. 
\wn{low prio: unify cut-vertex vs cutvertex}

The related notion of two-edge connectivity has a similar description: a pair of vertices are two-edge connected if there are two \emph{edge}-disjoint paths connecting them, or, by Menger's Theorem, there does not exist a cut-\emph{edge} separating them. And two-edge connectivity is an equivalence relation on the \emph{vertices}. 

In a dynamic graph, some of the main challenges with maintaining biconnectivity and two-edge connectivity compared to connectivity, can be understood via these equivalence relations on, respectively,  edges and vertices. For just connectivity, deleting an edge of the graph can lead to one equivalence class (i.e. one connected component) splitting up into two equivalence classes (i.e. two connected components). For two-edge and biconnectivity, this is no longer the case. Deleting just one edge may split one equivalence class into up to linearly many new ones. For two-edge connectivity, a data structure has to (explicitly or implicitly) notice or register new cut-edges that appear because of the deletion. For biconnectivity, the data structure has to be similarly aware both of new cut-edges and of new cutvertices, that now suddenly separate the endpoints of the deleted edge.
%
%
%todo: similarly, $k$-vertex connectivity has only recently seen linear time algorithms, where $k$-edge connectivity algorithms have existed longer. yadayada. not sure. -er.

Due to these increasingly hard challenges, there has historically been a hierarchy between connectivity, two-edge connectivity, and biconnectivity, in that order, all of which have been studied using similar methods through the history of dynamic graph algorithms. Frequently, while the technical tools for dynamic connectivity are useful towards dynamic two-edge connectivity and biconnectivity, the solutions of the latter are increasingly more technically involved, with more complicated solutions or slower running times. A chief example of this is the seminal deterministic dynamic connectivity paper~\cite{HolmDeLichtenbergThorup}, where the running times are $\log ^2 n$, $\log ^4 n$, and $\log ^5 n $, respectively. 

With this paper, we show that while there is very much still a hierarchy in the complexity of the data structures needed to maintain those different notions of connectivity, their best deterministic algorithms now all have operation times that are in the order of $\BigOTilde(\log ^2 n)$, differing only in the number of $\log\log n$ factors. Here, we use $\BigOTilde(f(x))$ to hide $\log(f(x))$ factors.

Our main contribution is the following theorem:

\begin{theorem}
  \label{thm:main}
  There exists a~deterministic data structure maintaining a~fully dynamic $n$-vertex undirected graph in the word RAM model with $\Omega(\log n)$ word size.
  The data structure can handle the following updates and queries:
  \begin{itemize}
    \item $\Insert{v, w}$ insert an~edge with endpoints $v$ and $w$, 
    %in $\BigO((\log n)^2 (\log \log n)^{999})$ amortized time,
	\item $\Delete{e}$ delete the edge $e$,
    \item $\AreBiconnected{v, w}$: given two vertices $v$, $w$, determine whether $v$ and $w$ are biconnected,
    \item $\NextCutVertex{v,w}$: given two vertices $v$ and $w$ that are in the same connected component, if there exists at least one cutvertex whose removal separates them, return the such cutvertex closest to $v$.  (Note that any path from $v$ to $w$ will have the cutvertices separating v from w in the same linear order, so ``first'' is well-defined.) If no such cutvertex exists, return $w$, regardless of whether $w$ is cutvertex.
    % and not biconnected, return the cutvertex closest to $v$ whose removal disconnects them. 
  \end{itemize}

  The amortized time consumption of $\Insert{}$ and $\Delete{}$ is   $\BigO\left(\log^2 n \log ^\nologlogs \log n\right)$, and the queries take $\BigO\left( \log n \log ^\nologlogs \log n\right)$ worst-case time. 

  The space consumption of the data structure is $\BigO\left(m + n  \log ^\nologlogs \log n\right)$.
\end{theorem}

The $\NextCutVertex{}$ operation above is thought of as a useful iterator: one of the applications of fully dynamic biconnectivity is its usage in static graph algorithms: Often, one can translate a specific type of induction proof into an efficient algorithm. Namely, the strategy, in which one deletes an arbitrary edge, notices whether biconnectivity is violated, handles each occurring biconnected component recursively in the affirmative case (motivating the $\NextCutVertex{}$ function), and combines them into a final solution that also takes the deleted edge into account. Note that for such use cases, amortized running times like ours are no less attractive than worst-case ones.

\paragraph{History of fully dynamic biconnectivity} The first sub-linear update time algorithms were presented in FOCS 1992. One by Henzinger~\cite{Rauch92,Henzinger95fully}, presenting a deterministic, amortized update time of $\BigO(m^{2/3})$ for the problem\footnote{We use $n$ and $m$ to denote, respectively, the number of vertices in the graph, and the number of edges at the time of the update.}. Another by Eppstein, Galil, Italiano, and Nissenzweig~\cite{EppsteinGIN92,Eppstein97}, whose general sparsification framework can be used to give an update time for biconnectivity that is linear in the number of vertices rather than in the number of edges. Later, these were improved to deterministic worst-case $\BigOTilde(m^{1/2})$, also by Henzinger~\cite{Rauch92,Henzinger00}. It was soon observed that those two flavours of techniques can be combined to obtain a running time of $\BigOTilde(n^{1/2})$.

%\erinline{Was Henzinger first? Wasn't the first sparsification also in 92? Their update time is, wait for it, $\BigO(n)$. Lol.}

%\erinline{Todo: we think it has been noted somewhere that the times by Henzinger can be improved using sparsification as a pretty plug-and-play like thing.}

%For the special class of planar graphs, these update times were improved using the sparsification framework of Eppstein, Galil, Italiano, Nissenzeig, and Spencer, first presented at FoCS 1992 in the context of dynamic connectivity~\cite{EppsteinGIN92,Eppstein97}, and later used to give $\BigO(\sqrt{n})$ deterministic worst-case update times for biconnectivity~\cite{Eppstein99}.

The first polylogarithmic algorithm for biconnectivity was presented in FOCS 1995 by Henzinger and King~\cite{HenzingerK95}, obtaining a Las Vegas style randomised algorithm with an amortized expected update time of $\BigO(\log ^4 n)$ and a query time of $\BigO(\log ^2 n)$.  

The first \emph{deterministic} polylogarithmic fully-dynamic algorithm for biconnectivity was presented in STOC'98~\cite{HolmLT98,HolmDeLichtenbergThorup}. The original conference version of that paper (\cite{HolmLT98}) erroneously claimed (in the abstract and theorem statement) that the amortized time per operation was $\BigO(\log ^4 n)$, and in the journal version~\cite{HolmDeLichtenbergThorup} this is updated to the $\BigO(\log^5 n)$ that is on-par with the actual data structures and proofs from the paper. When later Thorup~\cite{Thorup00} at STOC 2000 presented a method for shaving an $\BigO(\log n / \log\log n)$ factor off from the running times of \cite{HolmDeLichtenbergThorup}'s algorithms for connectivity, 2-edge connectivity, and biconnectivity, the miscount of log-factors had carried over, and \cite{Thorup00} paper states that the new technique leads to improved update times of $\BigO(\log ^3 n \log\log n)$ --- this is not the case; it leads to improved update times of $\BigO(\log^4 n \log\log n)$. 

Later, in SODA 2018, Holm, Rotenberg, and Thorup present further improved algorithms for the related problem of two-edge connectivity~\cite{Holm18a}, presenting a data structure with an update time of $\BigO(\log ^2 n \log^2 \log n)$, that is, shaving yet another $\BigO(\log n / \log\log n)$ factor. The authors of \cite{Holm18a} hypothesise that similar techniques yield an immediate improvement for biconnectivity. However, as we will soon clarify in Section~\ref{sec:tree-structure}, those techniques  rather 
%improve the running time from $\BigOTilde(\log ^4 n + \log ^4 n)$ to $\BigOTilde(\log ^3 n + \log ^4 n)$. That is, they 
improve the running time for only one of the bottlenecks for the update algorithm. Even obtaining $\BigOTilde(\log ^3 n)$ amortized update time for biconnectivity requires non-trivial work. And obtaining $\BigOTilde(\log ^2 n)$ additionally requires new technical innovations and ideas, which we present in later sections. 

\subsection{Technical Overview} \label{ssec:overview}
Our data structure, at its core, shares design principles with the original $\BigO(\log^5 n)$ update time biconnectivity data structure~\cite{HolmDeLichtenbergThorup}: we dynamically maintain a~spanning forest $F$ of the graph $G$ at hand.
For every vertex $v$ of the graph, we maintain an~instance of \emph{neighborhood data structure} $N_v$ that basically maintains the restriction of the biconnectivity relation in $G$ to the set of neighbors of $v$ in $F$.
On top of that, we maintain \emph{top trees}~\cite{TopTreesOriginal} that dynamically provide a~hierarchical decomposition of $F$ into well-structured subtrees (\emph{clusters}), enabling us to deduce the global biconnectivity relation from the information stored in the neighborhood data structures $N_v$.
Each update of $G$ brings $\BigO(\log n)$ changes to $F$ on average, and each such change in turn adjusts $\BigO(\log n)$ clusters of $F$.
In the original biconnectivity data structure, each cluster adjustment requires an~$\BigO(\log^3 n)$-time update of a~neighborhood data structure, followed by an~$\BigO(\log^2 n)$-time maintenance of bookkeeping information stored together with the cluster, resulting in amortized $\BigO(\log^5 n)$ time per graph update.

We structure our work as a~series of reductions to progressively simpler data structure problems.
In~\Cref{sec:graph-structure}, we reduce the original problem of biconnectivity in $G$ to a~tree problem of the maintenance of \emph{tree cover levels} in $F$.
While this step is mostly based on the ideas from~\cite{HolmDeLichtenbergThorup}, the original work had a~lot of moving parts that made reasoning about the correctness and efficiency quite troublesome (hence the incorrect time complexity analysis in the conference version of their paper).
Here, we propose a~new view of this reduction, which makes the required arguments much more transparent; we consider this to be (one of) significant technical contributions of this work.

Next, in~\Cref{sec:tree-structure}, our crucial contribution is to show how to use top trees to efficiently reduce the maintenance of tree cover levels in $F$ to the implementation of \emph{biased} neighborhood data structures, where we assign \emph{weights} to each element of the neighborhood data structure, aiming to process queries related to \emph{heavier} elements much more efficiently.
This way, we ensure that queries to the neighborhood data structures performed by the tree cover level data structure have --- on average --- \emph{constant cost}.
We also show how to use the methods of Thorup~\cite{Thorup00} and Holm, Rotenberg and Thorup~\cite{Holm18a} to optimize the $\BigO(\log^2 n)$ time required to bookkeep information stored for each cluster of the top tree to $\BigO(\log^2 \log n)$.

Then, in~\Cref{sec:nbd}, we reduce the implementation of neighborhood data structure to the problem of \emph{biased disjoint sets}, which we solve in~\Cref{sec:bds}.
Eventually, this allows us to process each update of $N_v$ of \emph{cost} $C$ in amortized $\BigO(C \cdot \log^2 \log n)$ time.

We now follow with a~more in-depth explanation of the techniques used in our work.

\paragraph{Reducing dynamic biconnectivity to the dynamic tree cover level data structure}
%\wninline{TODO: Advertise that we have a nicer abstraction etc than the old paper}
%We preface that part by saying that it mostly based on corresponding ideas from \cite{HolmDeLichtenbergThorup} that solved this problem in $\Oh(\log^5 n)$ time per an update, however that version of the algorithm had quite a lot of notions and moving parts that are intertwined in complex ways. We consider it a significant contribution to find different ways of using these ideas that allowed us to express the whole process using much cleaner abstractions. Without it, the path to the improvement to the final result of $\BigOTilde(\log^2 n)$ time per an update would have been much less clear to follow.
%\wninline{I tried to advertise these nicer abstractions, but I think that advertising things is not my strongest side, so improvements are welcomed. Also, it may be good to summarize how we did it, but I don't have the old paper in my RAM good enough to give a meaningful comparison. I am sure uncovering by successively decreasing the level one by one rather than fully uncovering and then recovering was the main idea though.}
The original amortized polylogarithmic data structure for biconnected components~\cite{HolmDeLichtenbergThorup} maintains a~spanning forest $F$ of the graph $G$.
Each non-tree edge of the graph is associated a~level $i \in \{0, 1, \ldots, \lmax\}$, where $\lmax \in \BigO(\log n)$.
This assignment induces a~sequence of graphs $G = G_0 \supseteq G_1 \supseteq \ldots \supseteq G_{\lmax} \supseteq F$, where $G_i$ is the subgraph of $G$ whose edge set contains precisely $E(F)$ and the non-tree edges of level at least $i$. We are going to continue using that idea.

According to the definition of biconnectedness, if two vertices $u$ and $v$ are connected by a bridge, they are not biconnected, however for our needs it will be more convenient to treat them as such. Hence we introduce the notion of \emph{pseudo-biconnectivity}, that is, we say that two vertices $u$ and $v$ are \emph{pseudo-biconnected} if and only if they are either biconnected or connected by a bridge. Checking if two vertices are connected by a bridge is not a challenge, hence these two notions are somewhat equivalent for our end goals. We remark that what \cite{HolmDeLichtenbergThorup} somewhat misleadingly called as biconnectivity relation is exactly the prescribed pseudo-biconnectivity relation.

%Consider then two adjacent tree edges $vx, vy$ (that is, edges of $F$).
%We say that the pair $vx, vy$ is \emph{biconnected at level~$i$} if there exists a~path from $x$ to $y$ in $G_i$ that excludes $v$ and all non-tree edges of level lower than $i$ (equivalently, $x$ and $y$ are biconnected in~$G_i$).
%Then, we define the \emph{cover level} of the pair $vx, vy$ as the largest~$i$ such that $vx, vy$ is biconnected at level $i$; if $vx, vy$ are not biconnected at level $0$, we define the cover level of $vx, vy$ to be $-1$.
Consider then two edges $e_1, e_2$.
We say that the pair $e_1, e_2$ is \emph{biconnected at level~$i$} if there exists a biconnected component in $G_i$ containing both $e_1$ and $e_2$.
Then, we define the \emph{cover level} of the pair $e_1, e_2$ as the largest~$i$ such that $e_1, e_2$ is biconnected at level $i$; if $e_1, e_2$ is not biconnected at level $0$, we define the cover level of $e_1, e_2$ to be $-1$.
We are going to focus mostly on cover level of pairs $vx, vy$, which are adjacent tree edges (that is, edges of $F$). For such pairs, the condition of belonging to the same biconnected component of $G_i$ can be equivalently expressed as the existence of a path from $x$ to $y$ in $G_i$ that excludes $v$ and all non-tree edges of level lower than $i$. \wn{I extended the cover level definition from adjacent tree pairs to all pairs}

It is proved in~\cite{HolmDeLichtenbergThorup} that for every vertex $v$, the relation of biconnectedness at any fixed level~$i$ between the pairs of edges incident to $v$ is an~equivalence relation.
Hence the cover levels of edges incident to $v$ can be represented as a~sequence of equivalence relations $\mathcal{L}_0^v, \mathcal{L}_1^v, \ldots, \mathcal{L}_{\lmax}^v$.
The sequence is \emph{descending}, i.e., it satisfies the following property: for any pair of levels $i < j$ and an~equivalence class $X \in \mathcal{L}_j^v$, there is $Y \in \mathcal{L}_i^v$ with $X \subseteq Y$.
For convenience, for any edge $e$ incident to $v$, let $\mathcal{L}_i^v(e)$ be the equivalence class of $\mathcal{L}_i^v$ containing $e$. It is also proved in~\cite{HolmDeLichtenbergThorup} that the information of cover levels of adjacent tree edges is sufficient to determine whether two non-adjacent vertices $p, q$ are pseudo-biconnected. Namely, the highest value of $i$ such that $p$ and $q$ are pseudo-biconnected in $G_i$ is equal to the smallest value of the cover level of two adjacent edges on the unique simple path $p \ldots q$ in $F$. Here, we assume that if $p$ and $q$ are not pseudo-biconnected even in $G_0$, then the said highest value of $i$ equals $-1$ (for convenience, assume that the cover level of a~one- or two-vertex path is $\lmax$). %Here, and throughout the paper, $p\ldots q$ denotes the unique path between $p$ and $q$ in the spanning tree. 
We point out that adjacent vertices are always pseudo-biconnected.%\er{added.}

%\wninline{Come back to the statement above after finalizing biconn definition. Comment what happens if $pq \in E(F)$ no matter the decision.}

This observation motivates abstracting away the notion of pseudo-biconnectedness by introducing the following dynamic tree problem.
Fix an~integer $\lmax$.
Our goal is to maintain a~dynamic forest $F$ whose every vertex $v$ is annotated with a~descending sequence of equivalence relations $\mathcal{L}_0^v, \mathcal{L}_1^v, \ldots, \mathcal{L}^v_{\lmax}$ over the set of tree edges incident to~$v$, where these equivalence relations denote the cover levels of adjacent tree edges. We aim to maintain these relations and answer queries regarding minimum cover level on given paths, subject to any changes stemming from external updates to the graph and internal updates that we will issue. We remark that for the sake of answering whether two vertices are pseudo-biconnected, it is sufficient to simply know whether the cover level is $-1$ or at least $0$, but introducing the levels is the crucial idea behind making this data structure efficient. 

The structure of the forest can be altered and queried via the following operations:

\begin{itemize}
  \item $\Link{v, w}$: Add an~edge $(v, w)$ to the dynamic tree. The new edge is at cover level $-1$ with all adjacent tree edges.

  \item $\Cut{v, w}$: Remove the edge $(v, w)$ from the tree and all equivalence relations $\mathcal{L}^v_i$, $\mathcal{L}^w_i$.
  
  \item $\Connected{v, w}$: Return whether $v$ and $w$ are connected by a~tree path.
\end{itemize}

In the setting of maintaining biconnectivity of a~graph $G$, $\Link{}$ and $\Cut{}$ correspond, respectively, to adding a~bridge to the graph and removing an~edge from the spanning forest $F$ of the graph.

We also implement additional types of updates, using which we will introduce or remove non-tree edges in $G$, and increase or decrease the levels of these edges.
Consider first adding a~non-tree edge $pq$ to the graph; this insertion causes every pair of consecutive edges $e_1, e_2$ on the path $p \ldots q$ at cover level $-1$ to increase its cover level to $0$.
%\wninline{Shall we introduce the $p \ldots q$ notation more explicitly? I think it's not standard}
%\erinline{If we want? I think it's been used before in the related work.}

Similarly, whenever the level of $pq$ is increased from, say, $i-1$ to $i$, we increase the cover level of each pair $e_1, e_2 \subseteq p \ldots q$ from $i - 1$ to $i$. (Note that since the level of $pq$ was $i-1$ before the update, the cover level of each such pair $e_1, e_2$ was already equal to or more than $i - 1$.)
We will abstract both types of updates using the following operation $\Cover{}$:

\begin{itemize}
  \item $\Cover{p, q, i}$: Suppose that the cover level of $p \ldots q$ is at least $i - 1$.
  For every pair of consecutive edges $e_1, e_2$ on $p \ldots q$ at cover level exactly $i - 1$, increase the cover level of $e_1, e_2$ to $i$ as follows: let $x$ be the common endpoint of $e_1$ and $e_2$.
  Then in $\mathcal{L}^x_i$ replace $\mathcal{L}^x_i(e_1)$ and $\mathcal{L}^x_i(e_2)$ with $\mathcal{L}^x_i(e_1) \cup \mathcal{L}^x_i(e_2)$.
\end{itemize}

We also aim to specify a~reverse operation, bringing the cover level of each pair of edges $e_1, e_2$ on $p \ldots q$ from $i$ to $i - 1$ under the assumption that the cover level of each such pair is at least $i$ before the update.
This, however, is quite problematic: even if the cover level of $e_1, e_2$ is exactly $i$ (which means that $\mathcal{L}^x_i(e_1) = \mathcal{L}^x_i(e_2)$ and $\mathcal{L}^x_{i+1}(e_1) \cap \mathcal{L}^x_{i+1}(e_2) = \emptyset$), it is not clear at all how to split the equivalence class $\mathcal{L}^x_i(e_1)$ in $\mathcal{L}^x_i$ into two parts.

%\wninline{The wording of the following sentence is a bit confusing to me. It reads as if however we split it, as long as this invariant is satisfied, we will be fine. Which I don't think is the case?}
%\msinline{Fixed.}

Naturally, a~split must preserve the following invariant: after the update, we should have $\mathcal{L}^x_i(e_1) \supseteq \mathcal{L}^x_{i+1}(e_1)$ and $\mathcal{L}^x_i(e_2) \supseteq \mathcal{L}^x_{i+1}(e_2)$.
So if it happens that $\mathcal{L}^x_i(e_1) = \mathcal{L}^x_{i+1}(e_1) \cup \mathcal{L}^x_{i+1}(e_2)$, the update actually is determined uniquely: we must split $\mathcal{L}^x_i(e_1)$ into $\mathcal{L}^x_{i+1}(e_1)$ and $\mathcal{L}^x_{i+1}(e_2)$.
We will call the pair $e_1, e_2$ satisfying this condition \emph{uniform}.
With this notion in place, we define the following operation that is inverse of $\Cover{}$ (whenever it is legal to call it):
 
\begin{itemize}
  \item $\UniformUncover{p, q, i}$: Suppose that the cover level of $p \ldots q$ is at least $i$, and that every pair of consecutive edges $e_1, e_2$ on $p \ldots q$ at cover level exactly $i$ is \emph{uniform}.
  Then for each such pair $e_1, e_2$ decrease the cover level of $e_1, e_2$ to $i - 1$ as follows: let $x$ be the common endpoint of $e_1$ and $e_2$. In $\mathcal{L}^x_i$, replace $\mathcal{L}^x_i(e_1)$ with $\mathcal{L}^x_{i+1}(e_1)$ and $\mathcal{L}^x_{i+1}(e_2)$. For convenience, we assume that $p \ldots q$ may consist of zero or one edges, in which cases it does nothing.
\end{itemize}

For technical reasons (needed to ensure the efficiency of the data structure), we may also specify that the dynamic tree cover level data structure is \emph{restricted} in the following sense: at any point of time, one path $a \ldots b$ in the forest may be \emph{exposed} through a~call $\Expose{a, b}$.
Then for $\UniformUncover{p, q, \cdot}$, we additionally require that $p \ldots q$ be a~subpath of $a \ldots b$.

We also define a~variant of the operation that permits uncovering non-uniform pairs of edges $e_1, e_2$ by specifying that the equivalence class $\mathcal{L}^x_i(e_1)$ after the split is as small as possible (i.e., equal to $\mathcal{L}^x_{i+1}(e_1)$).
This will come at the expense of limiting the operation to \emph{local} uncovers, affecting only a~single specified pair of edges:

%\wninline{I think that the descriptions above and below are not consistent, i.e. the above one says that the class of $e_2$ will be as small as possible, while the one below that it would be the case for $e_1$. I think that in my section I used the version with $e_1$ being as small as possible.}
%\msinline{Fixed above.}

\begin{itemize}
  \item $\LocalUncover{e_1, e_2, i}$: Let $e_1, e_2$ be adjacent edges with the common endpoint $x$ and at cover level exactly $i$.
  Decrease the cover level of $e_1, e_2$ to $i - 1$ by replacing $\mathcal{L}^x_i(e_1)$ in $\mathcal{L}^x$ with $\mathcal{L}^x_{i+1}(e_1)$ and $\mathcal{L}^x_i(e_1) \setminus \mathcal{L}^x_{i+1}(e_1)$.
\end{itemize}

The data structure allows to query the cover levels in the tree via the following interface:

\begin{itemize}
  \item $\CoverLevel{p, q}$: Return the cover level of the path $p \ldots q$.

  \item $\MinCoveredPair{p, q}$: Assuming that $p$ and $q$ are not connected by a~tree edge, return the first pair of edges $e_1, e_2$ on $p \ldots q$ at cover level $\CoverLevel{p, q}$.
\end{itemize}

We note that for any $i$, each biconnected component of $G_i$ induces a connected subtree of $F$. Adding an edge $vw$ to $G_i$ causes all biconnected components of $G_i$ sharing an edge with $v \ldots w$ to merge into one biconnected component. Viewing a deletion of a non-tree edge as an inverse of the addition, we conclude that deleting a non-tree edge $vw$ may cause a biconnected component of $G_i$ to split into a string of biconnected components that are linearly ordered along $v \ldots w$. We note that an internal vertex $u$ of $v \ldots w$ before the addition of $vw$ may touch only two biconnected components that will be a part of the biconnected component containing $v$ and $w$. Inversely, if we remove $vw$ from $G_i$, then after the removal there will be at most two biconnected components of $G_i - vw$ touching $u$ that were the part of the biconnected component of $G_i$ containing $v$ and $w$ before the removal. In other words, after a removal of an edge, the equivalence class of $\mathcal{L}_i^u$ containing the two edges of $v \ldots w$ incident to $t$ may split into two parts, each containing one such edge.

%\wninline{If somebody with drawing talent (Eva?) has some spare time, then maybe adding a picture for the above paragraph can be useful?}
%\erinline{sounds like a nice-to-have. I will see if I can get to it at some point.}

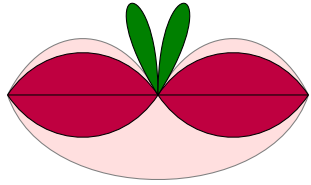
\begin{figure}[h]
\begin{center}
\begin{tikzpicture}
%	\draw (-2,0) -- (2,0);
%	\filldraw [gray] (0,0) circle (2pt);
	\draw[gray, fill=pink!50!white] (-2,0) .. controls (-1.5,1) and (-.5,1) .. (0,0);
	\draw[gray,fill=pink!50!white] (0,0) .. controls (.5,1) and (1.5,1) .. (2,0);
	\draw[gray,fill=pink!50!white] (-2,0) .. controls (-1.5,-1.5) and (1.5,-1.5) .. (2,0);

	\draw[fill=purple] (-2,0) .. controls (-1.5,.75) and (-.5,.75) .. (0,0);
	\draw[fill=purple] (2,0) .. controls (1.5,.75) and (.5,.75) .. (0,0);
	\draw[fill=purple] (-2,0) .. controls (-1.5,-.75) and (-.5,-.75) .. (0,0);
	\draw[fill=purple] (2,0) .. controls (1.5,-.75) and (.5,-.75) .. (0,0);
	\draw[fill=green!50!black] (0,0) .. controls (-.95,1.5) and (0, 1.75) .. (0,0);
	\draw[fill=green!50!black] (0,0) .. controls (.95,1.5) and (0, 1.75) .. (0,0);
	
	\draw (-2,0) to (0,0);
	\draw (2,0) to (0,0);

\end{tikzpicture}
\end{center}
\vspace{-1em} 
\caption{ Let $xy$ and $yz$ be consecutive along the tree path $v \ldots w$. Then, deleting the non-tree edge $vw$ may cause the biconnected component containing $xy$ and $yz$ (pink) to split up into new biconnected components (purple). Other biconnected components near $y$ are unaffected (green).}
\end{figure}

The only operation among $\Insert{}, \Delete{}, \AreBiconnected{}$ and $\NextCutVertex{}$ that is nontrivial to model using the tree cover level data structure is $\Delete{v, w}$ and we are going to sketch how it is handled now. For that, we mostly follow ideas from~\cite{HolmDeLichtenbergThorup}, but we adjust them to our needs, yielding a~cleaner abstraction at the same time.

As aforementioned, it is \emph{a priori} not obvious how to remove the influence of a non-tree edge $vw$ on the cover levels (the case of removing a tree edge can be reduced to the case of removing a non-tree edge). 
Let us denote its level as $i$.
The original data structure~\cite{HolmDeLichtenbergThorup} simply removes the edge from the graph, causing the equivalence relations $\mathcal{L}_j^u$ for internal vertices $u$ of $v \ldots w$ and $j \in [0, i]$ to be temporarily invalidated and multiple internal invariants to be broken; then, they propose an~involved scheme that progressively \emph{recovers} these relations and invariants.
Instead, we propose a~much more manageable view of the process: we will gradually \emph{decrease} the level of $vw$, maintaining the invariants controlling the equivalence relations at all points of time, until the level of $vw$ drops to $-1$, at which point the edge can be safely removed from the graph.

Now, we only need to understand what happens when the level of $vw$ is decreased from $i$ to $i-1$.
If the common vertex $u$ of a~pair of adjacent edges $ux, uy$ is not an internal vertex of $v \ldots w$, then the cover level of the pair $uv, uy$ is not influenced by the removal of $vw$. It is also not influenced if the cover level of $uv, uy$ is not equal to $i$. But for the equivalence classes $\mathcal{L}_i^u$ for vertices $u$ that are internal vertices of $v \ldots w$, it is not immediately clear \emph{if} they will be split, and if yes, it is not clear \emph{how} they will be split.
In order to get to know that, we will inspect various other level-$i$ edges that may affect the shape of these equivalence classes.
The number of these edges may even be linear, so in order to make the procedure efficient, we are going to amortize our work by promoting such edges from level $i$ to $i+1$. We will ensure that the biggest achievable level will be $\Oh(\log n)$, so the total number of promotions throughout the whole algorithm will be $\Oh(m \log n)$, effectively bounding the amount of work that needs to be done. In order to bound the highest achievable level, we maintain the invariant that the biconnected components on level $i$ have at most $\ceil{\frac{n}{2^i}}$ vertices.
 
However, this invariant complicates matters: sometimes, it may be not possible to promote an edge, because promoting it would cause that invariant to be broken. In order to deal with that, we carefully craft an~order of browsing the edges of interest ``from left to right'', and determine the biconnected components resulting from the required splits in the order from $v$ to $w$ (since they are linearly order along $v \ldots w$).
When we encounter an edge that we cannot promote without breaking the invariant, we stop the procedure and repeat the symmetric process ``from right to left''. Hence, there will be only two edges that we will handle that will not be accounted for in the amortization argument.
But this leaves us with another issue: the premature stopping of the process may leave some edges of interest unprocessed. The final argument is to note that the two edges that caused both search phases to be stopped prematurely has to be contained within the same unique big resulting biconnected component, hence the searches from both ends ``met'' in the same component that does not need to be split any further, ensuring that all the required splits were actually performed. This is the case because there can be only one resulting biconnected component larger than half of the original biconnected component, and promoting edges within smaller resulting biconnected components will always be legal.  
% chcemy powiedziec ze 
% 

Both the data structure of~\cite{HolmDeLichtenbergThorup} and our sketched algorithm crucially rely on the efficient \emph{counting} of the vertices that are reachable from any edge of a~given path $p \ldots q$ via a~path of a~given cover level $i$ in order to check if it is legal to promote an edge.
%Since we adapt this data structure to our needs, we also define a~variant of this operation.
We define that a~vertex $y$ is \emph{$i$-reachable} from a~tree path $P$ (at an~edge $e$) if there exists a~tree path $P'$ of cover level at least $i$ starting from $e$, ending at $y$ and intersecting $P$ precisely at $e$; we call such $P'$ the \emph{$i$-reachability witness}. We may additionally specify that $y$ is $i$-reachable from $P$ at $e$ \emph{through $w$} if $w$ is the closer of the two ends of $e$ to $y$.
Abusing the notation slightly, we will also sometimes say that $y$ is $i$-reachable from an~ordered edge $\vec{vw}$ or from an unordered edge $vw$ through $w$ if there exists a~tree path $P$ of cover level at least $i$ with two first vertices $v$, $w$ and the final vertex $y$.
We aim to implement the following method:

\begin{itemize}
  \item $\FindSize{p, q, i}$: Return the number of $i$-reachable vertices from $p \ldots q$.
\end{itemize}

This method is crucial, because $\FindSize{p, q, i}$ determines the size of the biconnected component of $G_i$ containing $p$ and $q$ after hypothetically adding the edge $pq$ to it. Hence, calling it tells us if it is legal to promote the edge $pq$ from level $i-1$ to $i$.

Next, we want to be able to \emph{mark} some vertices of the tree.
That is, every vertex $u$ can contain a~set of user marks; more precisely, for every $i \in \{0, 1, \ldots, \lmax\}$, $u$ can be either $i$-marked or $i$-unmarked by the user of the data structure; initially, each element is $i$-unmarked for all $i$.
Roughly speaking, $u$ will be $i$-marked whenever there exists a~non-tree edge of level $i$ incident to $u$.
The marks can be altered using the following procedures:

\begin{itemize}
  \item $\Mark{u, i}$: Make $u$ $i$-marked.
  
  \item $\Unmark{u, i}$: Make $u$ $i$-unmarked.
\end{itemize}

Finally, we want to \emph{search} for $i$-marked vertices that are $i$-reachable from a~given tree path $P$:

\begin{itemize}
  \item $\FindFirstReach{p, q, i}$: Return $(ab, c, y)$, where $ab$ is an edge of $p \ldots q$, $a$ is closer to $p$ than $b$ and $y$ is an $i$-marked vertex that is $i$-reachable from $p \ldots q$ at $ab$ through $c \in \{a, b\}$. Among all such tuples, choose the one where $\dist_F(a, p)$ is minimum, and in case of ties, minimize $\dist_F(c, p)$. Return $(\perp, \perp, \perp)$ if there is no tuple satisfying the conditions.
\end{itemize}

This searching method is meant to help us identify level-$i$ edges that influence cover levels of some pairs of edges along $p \ldots q$. For an $i$-marked vertex $y$ that is $i$-reachable from $p \ldots q$, we have that it has a neighbor $z$ such that $yz$ is a level-$i$ non-tree edge. Such an edge is interesting to look at since it guarantees that cover levels of adjacent pairs of edges on $p \ldots q \cap y \ldots z$ stays at least~$i$ after the removal of $pq$ from $G_i$. The condition of being $i$-reachable ensures that $yz$ belongs to the same biconnected component of $G_i$ as $pq$ before the removal of $pq$. %, as otherwise there would no point at looking at $yz$.
Minimizing $\dist_F(a, p)$ and $\dist_F(c, p)$ will help us identify them in the desired order ``from left to right''. 

For technical reasons, we will sometimes additionally require the first pair of edges of the reachability witness to be at cover level strictly larger than $i$.
Formally, we will say that a~vertex $y$ is \emph{strongly $i$-reachable}
%\ms{\tiny I guess this definition requires a fix}
from a~tree path $P$ at $e=vw$ through $w$ (or from $\vec{vw}$, or from $vw$ through $w$) if it is $i$-reachable from $P$ at $e$ through $w$ and, moreover, the reachability witness $P'$ consists of at least two edges, where the first two edges of $P'$ are at cover level at least $i + 1$.
Then we implement the following query:

\begin{itemize}
 \item $\FindStrongReach{p, q, e, b, i}$: Return an $i$-marked vertex $y$ that is strongly $i$-reachable from $p \ldots q$ at $e$ through $b$. Return $\perp$ in case there is no $y$ satisfying the conditions.
\end{itemize}

For a particular vertex $u$ whose class $\mathcal{L}_i^u$ is supposed to split into two parts, this method will help us distinguish between edges $yz$ that are supposed to be a part of the resulting biconnected component on the left of $u$, from these that are supposed to be a part of the resulting biconnected component on the right of $u$.

The following lemma proves that the sketched dynamic biconnectivity data structure can be efficiently reduced to the (restricted) dynamic tree cover level data structure:

\begin{lemma}
  \label{lem:first-reduction}
  The dynamic biconnectivity data structure can be implemented using a~restricted dynamic tree cover level data structure with $\lmax \in \BigO(\log n)$ so that, when initialized with an~edgeless $n$-vertex graph, any sequence of $m$ $\textsc{Insert}$ or $\textsc{Delete}$ updates can be modeled using:
  \begin{itemize}
    \item $\BigO(m)$ calls to $\Link{}$, $\Cut{}$ and $\Expose{}$;
    \item $\BigO(m \log n)$ calls to $\Cover{}$, $\UniformUncover{}$, $\LocalUncover{}$, $\FindSize{}$, $\Mark{}$, $\Unmark{}$, $\FindFirstReach{}$ and $\FindStrongReach{}$.
  \end{itemize}
  Moreover, $\textsc{AreBiconnected}$ and $\NextCutVertex{}$ queries can be implemented with a~constant number of calls to $\Connected{}$, $\CoverLevel{}$, $\MinCoveredPair{}$ and $\FindSize{}$.
  The time complexity required for this reduction is $\Oh(m \cdot \log^2 n)$.
\end{lemma}

The proof of Lemma~\ref{lem:first-reduction} is given in Section~\ref{sec:graph-structure}.
Afterwards, to prove Theorem~\ref{thm:main}, we will need to provide a~restricted dynamic tree cover level data structure implementation for $\lmax \in \BigO(\log n)$, where $\Link{}$, $\Cut{}$, $\Expose{}$ are performed in amortized time $\BigOTilde(\log^2 n)$; $\Connected{}$, $\CoverLevel{}$ and $\MinCoveredPair{}$ are implemented in worst-case time $\BigO(\log n)$; and the remaining operations are done in amortized time $\BigOTilde(\log n)$.

\paragraph{Reducing the tree cover level data structure to the neighborhood data structure}
Next we sketch an~implementation of the tree cover level data structure.
At the very high level, the structure holds: (i) an~$n$-vertex dynamic tree $F$ implemented via \emph{top trees} of Alstrup, Holm, de Lichtenberg, and Thorup~\cite{TopTreesOriginal}, and (ii) for every $v \in V(F)$, an~aforementioned descending sequence of equivalence relations $\mathcal{L}^v_0, \ldots, \mathcal{L}^v_{\lmax}$ ($\lmax \in \BigO(\log n)$) over the set of edges incident to $v$, each stored in a~separate instance $N_v$ of a~\emph{neighborhood data structure} that we will introduce in a~moment.

Recall that a~top tree dynamically maintains a~recursive edge-partitioning of $F$ into progressively smaller well-structured subtrees of $F$, called \emph{clusters}.
Here, a~cluster is a~connected subgraph $C$ of $F$ containing at most two vertices (called \emph{boundary vertices} of $C$) incident to the edges outside of $C$; in particular, the entire tree $F$ is a~cluster itself, and so is every single edge of $F$.
A~top tree is then a rooted tree of constant branching and height $\BigO(\log n)$, where: (1) each node is identified with a~cluster of $F$, (2) the root is identified with $F$, (3) each leaf is identified with an~edge of $F$, (4) the children of a~non-leaf cluster $C$ form an~edge-partitioning of $C$.
We will assume that a~non-leaf cluster splits into child clusters in a~very well-structured way: see \Cref{fig:toptree-cases} for all possible ways in which a~cluster $C$ may be edge-partitioned into two or three smaller clusters.

Each node in the top tree then stores summary information on the cluster $C$ identified with the node. In particular, if $C$ has two boundary vertices $v, w$, then the node corresponding to $C$ contains: (a) the cover level of the \emph{cluster path} $\pi(C) \coloneqq v \ldots w$, (b) for every $i \in [0, \lmax]$, the count of vertices of $C$ that are $i$-reachable from $\pi(C)$, and (c) for every $i \in [0, \lmax]$, the information on the $i$-reachability from $\pi(C)$ of $i$-marked vertices in $C$, allowing us to recover some $i$-marked vertex that is $i$-reachable from $\pi(C)$.
Here, the counts are stored in a~space-optimized array of \emph{counters} of length $\lmax + 1$, which we call \emph{counter vectors}.
In the description below, we will denote by $T(\lmax)$ the time needed to perform basic operations on counter vectors, such as coordinate-wise addition of entries of the array; using a~technique of Thorup~\cite{Thorup00}, we will show in \Cref{sec:approx-counting} that $T(\lmax) \in \BigO(\log \lmax)$ in the word RAM model.

We remark that each cluster $C$ needs to hold some additional information to ensure \emph{compositionality}: namely, that the information stored in $C$ can be determined only from information stored in the direct children of $C$ and the neighborhood data structures for the boundary vertices of the children of $C$.
Determining exactly what information should be stored and how it should be composed is a~formidable technical challenge that we skip for the purposes of this overview.
For now, it is enough to know that the solution is inspired by the original $\BigO(\log^5 n)$ work on biconnectivity~\cite{HolmDeLichtenbergThorup}, while also borrowing some technical tricks from the $\BigOTilde(\log^2 n)$ dynamic $2$-edge-connectivity data structure~\cite{Holm18a} (e.g., \emph{approximate counting} and a~clever use of \emph{prefix sums} for counter vectors).

On the other hand, a~\emph{neighborhood data structure} provides an~abstraction for a~descending sequence of equivalence relations $\mathcal{L}_0, \mathcal{L}_1, \ldots, \mathcal{L}_{\lmax}$ over a~ground set $X$.
We formally represent each $\mathcal{L}_i$ as a~partition of $X$, though we will sometimes write $x \sim_i y$ to mean that $x, y$ are in the same part of $\mathcal{L}_i$.
For convenience, we assume that $\mathcal{L}_{-1} = \{X\}$, i.e., all elements of $X$ are in the same part of $\mathcal{L}_{-1}$.
Similarly, $\mathcal{L}_{\lmax + 1} = \{\{x\} \,\mid\, x \in X\}$, so all elements of $X$ are in separate parts of $\mathcal{L}_{\lmax + 1}$.
The \emph{level} of a~pair $x, y$ is the maximum integer $i \geq -1$ such that $x \sim_i y$.
Note that this directly mirrors the previously introduced notion of cover levels: we will instantiate a~neighborhood data structure $N_v$ over the set of edges of $F$ incident to $v$, where any pair of edges is precisely at level given by their cover level in $F$.

The most basic variant of the neighborhood data structure supports the following types of updates and queries:
\begin{itemize}
  \item $\Insert{x}$: Add an~item $x$ to $X$.
  	For every level $i \in \{0, 1, \ldots, \lmax\}$, add $\{x\}$ to $\mathcal{L}_i$.
  \item $\Delete{x}$: Remove $x$ from $X$ and from each equivalence relation $\mathcal{L}_i$.
  \item $\Level{x, y}$: Return the level of the pair $x, y$.
  \item $\Zip{x, y, i}$: Given that $x \sim_{i-1} y$ and $x \not\sim_i y$, unify the parts containing $x$ and $y$ in $\mathcal{L}_i$, i.e., replace $\mathcal{L}_i(x)$ and $\mathcal{L}_i(y)$ with $\mathcal{L}_i(x) \cup \mathcal{L}_i(y)$ in $\mathcal{L}_i$.
  \item $\Unzip{x, y, i}$: Given that $x \sim_i y$ and $x \not\sim_{i+1} y$, separate $x$ from $y$ in $\mathcal{L}_i$ by replacing $\mathcal{L}_i(x)$ with $\mathcal{L}_{i+1}(x)$ and $\mathcal{L}_i(x) \setminus \mathcal{L}_{i+1}(x)$.
\end{itemize}

Naturally, we will call $\Insert{}$ and $\Delete{}$ in $N_v$ when adding or removing edges incident to $v \in V(F)$, and $\Zip{}$ and $\Unzip{}$ when altering the cover levels of pairs of edges incident to $v$.
The $\Level{}$ query shall be used as follows: suppose we have a~cluster $C$ that splits into clusters $A, B, P$ according to case~\CasePath in \Cref{fig:toptree-cases}; in particular, assume that the cluster path $\pi(C)$ is a~concatenation of cluster paths $\pi(A)$, $\pi(B)$.
Let also $v$ be the common endpoint of $\pi(A)$ and $\pi(B)$, and let $e_A, e_B$ be the edges of $\pi(A)$ and $\pi(B)$, respectively, each incident to $v$.
Then the cover level of $\pi(C)$ is the minimum of the following values: the cover level of $\pi(A)$, the cover level of $\pi(B)$, and the cover level of the pair $(e_A, e_B)$, determined by testing $\Level{e_A, e_B}$ in $N_v$.

In order to support the maintenance of counts of $i$-reachable vertices in the clusters of a~top tree, we augment the basic neighborhood data structure with a~\emph{counting extension}: each element $x \in X$ is assigned a~counter vector ${\bf c}^x = (c^x_0, c^x_1, \ldots, c^x_{\lmax})$ containing integers not exceeding $n$, initially populated with zeroes.
The counters can be modified and accessed as follows:
\begin{itemize}
  \item $\UpdateCounters{x, {\bf c}^x}$: Replace the counter vector ${\bf c}^x$ of $x$.
  \item $\SumCounters{x}$: Return the counter vector ${\bf s}^x = (s^x_0, s^x_1, \ldots, s^x_{\lmax})$ where $s^x_i = \sum \{c^y_i \,\mid\, x \sim_i y\}$.
\end{itemize}

Observe that for $v \in V(F)$, the query ${\bf s}^{\vec{vw}} \coloneqq \SumCounters{\vec{vw}}$ in the neighborhood data structure $N_v$ has the following semantics: suppose that $c^{\vec{vw}}_i = 0$ and, for all the remaining edges $\vec{vu}$ incident to $v$, $c^{\vec{vu}}_i$ equals the number of vertices $i$-reachable from $\vec{vu}$ in $F$.
Then $s^{\vec{vw}}_i$ is precisely the number of vertices that are $i$-reachable from $\vec{wv}$, excluding $v$ itself.
In other words, it is the number of vertices $y$ such that: (i) $y$ is in the subtree of $F$ rooted at an~edge $\vec{vu}$, (ii) $y$ is $i$-reachable from $\vec{vu}$, and (iii) $vu \sim_i vw$.
Thus the counting extension of a~neighborhood data structure is used to count the number of vertices in $F$ that are $i$-reachable from a~given edge of $F$.

By the same token, the \emph{marking extension} will facilitate the compositionality of the information on $i$-reachable $i$-marked vertices: each element $x \in X$ is given a~\emph{mark vector} ${\bf b}^x = (b^x_0, b^x_1, \ldots, b^x_{\lmax})$ containing boolean marks, initially false.
We say that $x$ is $i$-marked if $b^x_i$ is true, and $i$-unmarked otherwise.
These marks will be modified and accessed through the following queries:

\begin{itemize}
  \item $\UpdateMarks{x, {\bf b}^x}$: Replace the mark vector ${\bf b}^x$ of $x$.
  \item $\OrMarks{x}$: Return the bit vector ${\bf a}^x$ such that $a^x_i$ is true whenever there exists an~$i$-marked element $y$ such that $x \sim_i y$.
  \item $\FindMarked{x, i}$: Return an~$i$-marked element $y$ such that $x \sim_i y$, or $\bot$ if no such element exists.
\end{itemize}

Note that $\OrMarks{}$ is completely analogous to $\SumCounters{}$ in the counting extension of the neighborhood data structure (ultimately allowing us to answer queries of the form ``does $F$ contain an $i$-marked vertex $i$-reachable from a~path $p \ldots q$'', for all levels $i$ simultaneously), while $\FindMarked{}$ allows us to recover the identifier of any such vertex for some concrete level $i$.

While the neighborhood data structure with counting and marking extensions will already be enough to implement a~tree cover level data structure with amortized poly-logarithmic update and query guarantees, it is not yet enough to reach the desired time complexity guarantees ($\BigOTilde(\log^2 n)$ for $\Link{}$, $\Cut{}$, $\Expose{}$, and $\BigOTilde(\log n)$ for the remaining operations).
Thus we give another two extensions to the neighborhood data structure that will eventually enable us to reach our time complexity target.
First, we have the \emph{selection extension}, where we may decide to \emph{select} a~two-element subset $\Xsel$ of \emph{VIP neighbors} of $X$.
Then we additionally support the following operations:

\begin{itemize}
  \item $\Select{X_{\mathrm{sel}}}$: Redefine $X_{\mathrm{sel}} \subseteq X$ as the set of selected items. 
  \item $\SelectedLevel$: Suppose that $X_{\mathrm{sel}} = \{x, y\}$.
  Return the level of $x, y$.
  \item $\LongZip{i_1, i_2}$: Suppose that $i_1 < i_2$, $X_{\mathrm{sel}} = \{x, y\}$ and the level of $x, y$ is $i_1$.
  Unify the parts containing $x$ and $y$ in $\mathcal{L}_{i_1 + 1}, \ldots, \mathcal{L}_{i_2}$ by replacing, for every $i \in [i_1 + 1,\, i_2]$, $\mathcal{L}_i(x)$ and $\mathcal{L}_i(y)$ with $\mathcal{L}_i(x) \cup \mathcal{L}_i(y)$.
  (This is equivalent to calling $\Zip{x, y, i}$ for each $i = i_1 + 1, \ldots, i_2$ in succession.)
%  \wninline{Do we not require that $\mathcal{L}_{i_1+1}(x) = \ldots = \mathcal{L}_{i_2}(x)$? If we do not, then it could be the case that $\LongZip{i_1, i_2}$ is legal, but the following $\LongUnzip{i_2, i_1}$ is not.}
%  \msinline{You are correct, and this is intentional. Should we stress it?}
%  \wninline{That's very surprising to me. How come?}
  \item $\LongUnzip{i_2, i_1}$: Suppose that $i_1 < i_2$, $X_{\mathrm{sel}} = \{x, y\}$ and the level of $x, y$ is $i_2$.
  Under the assertion that $\mathcal{L}_{i_1 + 1}(x) = \ldots = \mathcal{L}_{i_2}(x) = \mathcal{L}_{i_2 + 1}(x) \cup \mathcal{L}_{i_2 + 1}(y)$, separate $x$ from $y$ in $\mathcal{L}_{i_1 + 1}, \ldots, \mathcal{L}_{i_2}$, by replacing, for every $i \in [i_1 + 1, i_2]$, $\mathcal{L}_i(x)$ with $\mathcal{L}_{i_2 + 1}(x)$ and $\mathcal{L}_{i_2 + 1}(y)$.
  (This is equivalent to calling $\Unzip{x, y, i}$ for each $i = i_2, i_2 - 1, \ldots, i_1 + 1$ in succession, and equivalent to calling $\Unzip{y, x, i}$ for each $i = i_2, i_2 - 1, \ldots, i_1 + 1$.)
%  \wninline{Shouldn't the equality $\mathcal{L}_{i_1}(x) = \ldots = \mathcal{L}_{i_2}(x) = \mathcal{L}_{i_2 + 1}(x) \cup \mathcal{L}_{i_2 + 1}(y)$ start from $i_1+1$ rather than $i_1$?}
%  \msinline{Fixed.}
\end{itemize}

The intuition behind the selection extension is that some basic operations on the neighborhood data structure $N_v$ (namely, $\Level{}$, $\Zip{}$, $\Unzip{}$) will be performed much more frequently on a~specific pair of edges incident to $v$ than on other pairs of edges.
For instance, when a~cluster $C$ is constructed as an~edge sum of child clusters $A, B, P$ according to case~\CasePath in \Cref{fig:toptree-cases}, $v$ is the common vertex of $A, B, P$ and $e_A, e_B$ are the edges of $\pi(A)$, $\pi(B)$, respectively, incident to $v$, then we will regularly call $\Level{e_A, e_B}$, $\Zip{e_A, e_B, \cdot}$, and $\Level{e_A, e_B, \cdot}$ in $N_v$.
In this case, we choose to perform a~(computationally expensive) call $\Select{\{e_A, e_B\}}$, which in turn will enable us to determine the cover level of the pair $e_A, e_B$ much more efficiently (via $\SelectedLevel{}$ instead of $\Level{}$), and perform batch (\emph{long}) updates of the cover level of the pair $e_A, e_B$: increase or decrease the cover level of the pair by several levels in a~single step.

% \msinline{\footnotesize I commented out one paragraph below. I don't care too much about it, maybe some of you do?}

%Why do we need to perform long updates of the cover level of the selected edges?
%It turns out that the top tree cover level data structure performs the updates $\Cover{}$ and $\UniformUncover{}$ lazily: when $C$ is a~cluster with nonempty $\pi(C)$, and we are supposed to perform an~update (a~cover or a~uniform uncover) on the entire cluster path $\pi(C)$, we will only \emph{pretend} to perform it, and actually apply it to $C$ only when we need to access any descendant of $C$ in the top tree.
%This may, however, cause a~sequence of pending updates in $C$ to accumulate.
%We will show that every sequence of pending updates is equivalent to a~sequence of \emph{covers} at levels $i_1 + 1, \ldots i_2$ for some $i_1 \leq i_2$, followed by a~sequence of \emph{uncovers} at levels $i_2, \ldots, i'_1 + 1$ for some $i'_1 \leq i_2$.
%So assuming the setting of case~\CasePath in \Cref{fig:toptree-cases}, actually performing these updates in $C$ boils down to: (1) propagating the pending updates to child clusters $A$ and $B$, (2) adjusting some internal bookkeeping information, and (3) updating the cover level of the pair $e_A, e_B$ in $N_v$ to $i'_1$ if it was previously equal to some $i \in [i_1, i_2]$.
%Observe then that (3) is equivalent to calling $\LongZip{i, i_2}$, followed by $\LongUnzip{i_2, i'_1}$.

We remark that $\LongUnzip{}$ is a~partial inverse of $\LongZip{}$ in the following sense: whenever $i_1 < i_2$ and $\LongUnzip{i_2, i_1}$ is legal to perform in $N_v$, the sequence of updates $\LongUnzip{i_2, i_1};\allowbreak \LongZip{i_1, i_2}$ is a~no-op.
However, it is not a~\emph{full} inverse: in some cases, it may be illegal to bring the cover level of the selected pair of items from $i_2$ down to $i_1$.
The astute reader is encouraged to find how this restriction of $\LongUnzip{}$ is analogous to the uniformity condition of $\UniformUncover{}$.
And observe that $\LongZip{i_1, i_2}$ is \emph{always legal}, similarly to how the interface of $\Cover{}$ in the interface of the tree cover level data structure does not place any uniformity conditions.

%\msinline{Maybe write why LongZip and LongUnzip aren't simple inverses of each other}

The final, crucial extension is \emph{biasing}: each item in $X$ can be assigned a~positive integer weight $w(x) \leq \BigO(n)$, set to $1$ by default.
The weights can be modified through the following update:
\begin{itemize}
  \item $\SetWeight{x, w(x)}$: Replace the weight of $x$ with $w(x)$.
\end{itemize}
As it is the case with many biased data structures, biasing allows us to perform queries on \emph{heavy} elements of $X$ very efficiently.
In $N_v$, the weight of an~edge $vw$ will be usually roughly equal to the size of the component $F[\vec{vw}]$ of $F - vw$ containing $w$.
This way, we ensure quick accesses and updates of the elements of $N_v$ related to the edges $vw$ incident to $v$ for which the subtree $F[\vec{vw}]$ is large.
This efficiency is formalized through the notion of a~\emph{normalized cost} of an~operation in a~neighborhood data structure.
Namely, letting $w(X)$ be the total weight of all elements of $X$:
\begin{itemize}
  \item queries $\Insert{}$, $\Delete{}$, $\Select{}$ and $\SetWeight{}$ have normalized cost $\log n$;
  \item queries $\Zip{x, y, \cdot}$, $\Unzip{x, y, \cdot}$, $\Level{x, y, \cdot}$ have normalized cost $1 + \log \frac{w(X)}{w(x)} + \log \frac{w(X)}{w(y)}$;
  \item queries $\UpdateCounters{x, \cdot}$, $\SumCounters{x}$, $\UpdateMarks{x, \cdot}$ and $\OrMarks{x}$ have normalized cost $1 + \log \frac{w(X)}{w(x)}$;
  \item query $\FindMarked{x, \cdot}$ has normalized cost $1 + \log \frac{w(X)}{w(x)}$ if it returns $\bot$, and $1 + \log \frac{w(X)}{w(x)} + \log \frac{w(X)}{w(y)}$ if it returns an~element $y$;
  \item queries $\LongZip{}$, $\LongUnzip{}$ and $\SelectedLevel{}$ have normalized cost $1$.
\end{itemize}

We will implement the cover level data structure so that each operation affects only nodes of the top tree present on a~constant number of root-to-leaf paths (in the top tree); and for each node examined, we issue a~constant number of calls to the neighborhood data structures.
On our way to the proof of the efficiency of our data structure, we will prove a~powerful structural result about top trees --- the Vertical Path Telescoping Lemma (\Cref{lem:top-tree-telescope}) --- which will allow us to argue that, under certain conditions, the \emph{total} normalized cost of the calls to the neighborhood data structures performed when examining a~root-to-leaf path in the top tree is bounded by $\BigO(\log n)$, even though we sometimes need to perform as many as $\BigO(\log n)$ calls to these structures in total.
This, in turn, will allow us to show that a~huge array of operations in the cover level data structure can be implemented in $\BigOTilde(\log n)$ time, plus $\BigO(\log n)$ calls to neighborhood data structures of total normalized cost $\BigO(\log n)$.

The use of selection and biasing extensions poses, however, an~unexpected challenge.
The usual framework of performing updates and queries in top trees is to \emph{expose} the set of vertices associated with the query~\cite{TopTreesOriginal}; so for example, in order to determine the cover level of a~path $p \ldots q$, we would first call $\Expose{p, q}$.
This rebuilds the top tree slightly by altering a~total of $\BigO(\log n)$ clusters, and ultimately causes the answer to the query to be conveniently placed as part of information associated with the root cluster.
However, we cannot afford to use this technique here directly: rebuilding the top tree on each query turns out to be too computationally expensive due to the normalized cost of $\Select{}$ and $\SetWeight{}$.
Therefore, we design a~technique of \emph{transient expose} in \Cref{ssec:tree-structure-overview}, which essentially constructs a temporary ``read-only view'' of the top tree with a~given set of selected vertices.
This technique offers the best of two worlds: it is both computationally cheap, allowing us to perform it relatively frequently, and it considerably simplifies the implementation of updates and queries such as $\CoverLevel{}$.
With this final technical tool at hand, we can finally give a~statement of an~efficient reduction from the tree cover level data structure to the neighborhood data structure:

\begin{restatable}{lemma}{TopTreeReductionLemma}
  \label{lem:top-tree-reduction}
  Let $\lmax \in \BigO(\log n)$ and $\hat{T} \coloneqq T(\lmax) \cdot \log \lmax$.
  There exists a~restricted dynamic tree cover level data structure with $\lmax$ levels that processes each operation of the form:
  \begin{itemize}
    \item $\Link{}$, $\Cut{}$, $\Expose{}$ in worst-case $\BigO(\log n \cdot \hat{T})$ time, plus queries to the neighborhood data structures of total normalized cost $\BigO(\log^2 n)$;
    \item $\Connected{}$ in worst-case $\BigO(\log n)$ time;
    \item $\Cover{}$, $\UniformUncover{}$, $\LocalUncover{}$, $\CoverLevel{}$, $\MinCoveredPair{}$, $\FindSize{}$, $\Mark{}$, $\Unmark{}$, $\FindFirstReach{}$, $\FindStrongReach{}$ in worst-case $\BigO(\log n \cdot \hat{T})$ time, plus queries to the neighborhood data structures of total normalized cost $\BigO(\log n)$.
  \end{itemize}
\end{restatable}

Since $\hat{T} \in \BigO(\log^2 \lmax) = \BigO(\log^2 \log n)$ in the word RAM model, it now remains to give an~implementation of a~neighborhood data structure that runs each query of normalized cost $C$ in amortized time $\BigO(C \cdot \mathrm{poly} \log \log n)$.

\paragraph{Reducing the neighborhood data structure to the biased disjoint sets problem}

In the neighborhood data structure, we want to efficiently keep track of the leveled hierarchy of biconnected components of all the edges incident to some vertex. Conceptually, we want to do so via a tree of height $\lmax$, in which we can let nearest common ancestor queries return a certificate of the highest level at which a specified pair of neighboring edges are biconnected. Thus, in this neighborhood tree for the vertex $v$, we want each leaf to correspond to an edge $uv$ incident to $v$. 

Implementing this idea as-is would incur too many $\log n$-factors. One of the crucial ideas to avoid this, is to use a biased variant of heavy path decompositions~\cite{SleatorT81}, as introduced in \cite{BentST80,BentST85}, and as also utilised in \cite{HolmR20}. Then, the edge $uv$ would be a weighted leaf in $v$'s neighborhood tree, whose weight corresponds to the subtree rooted in $u$, which would help allowing the Vertical Path Telescoping Lemma~\ref{lem:top-tree-telescope}.

However, it can happen that the vertex $v$ lies on an exposed path. In this case, we need to take special care of its two incident exposed path edges, ensuring that information about them and their hierarchy of biconnected components is at hand. This requires careful bookkeeping: we have two leaves in the tree that needs to be considered as being in a `superposition' of being not-biconnected on any level at all, to being biconnected all the way down, corresponding to two entire paths of nodes that may, sometimes, depending on an indicator, have to be considered as two separate versions of the same node, sometimes not. 

In particular, what careful bookkeeping we have in mind must be able to accommodate cover level changes of the path, and be susceptible to a change in which two edges are the exposed edges around a vertex. In other words, the data structure needs to be able to `zip' and `unzip' the biconnectivity of the exposed neighbors fast. We also need to efficiently increase the cover level of a pair of neighbors. When doing so, we may have to union and split the sets of light children of a node in a heavy path decomposition efficiently. 
Here, it is imperative that we use a non-trivial data structure for joining and splitting disjoint sets, in order to make these operations on sets of light children without incurring an additional log factor.

Many of techniques draw on inspiration from \cite[Section 3]{HolmR20} ``Biased Dynamic Trees''. Some main differences are the following. 
Firstly, we use a slightly modified definition of heavy edges, to accommodate the exposed edges incident to a vertex; this incurs only an additive constant to the light depths of leaves. 
Secondly, the collection of light children are organised using our new biased disjoint set data structure. 
Finally, we can use a naive balanced binary search tree over the heavy paths, since those have length at most $\lmax \in \BigO(\log n)$. 

%\jhinline{Cite~\cite{HolmR20}[Lemma~16].  We can not use~\cite{HolmR20}[Lemma~20] directly because our weights are not $k$-positive.}
%

\begin{lemma}\label{lem:nbhtree}
  \label{lem:neighborhood-ds}
  Let $\lmax \in \BigO(\log n), W:=w(X)\in\BigO(n)$ and $\hat{L} \coloneqq T(\lmax) \cdot (\log \lmax + \log\log W)$.
  There exists a~neighborhood data
  structure with $\lmax$ levels, supporting marking, counting, selection and biasing extensions
  where:
  \begin{itemize}

%  \item $T(x) \in \BigO(x)$ if the neighborhood data structure is implemented in the combinatorial model, and $T(x) \in \BigO(\log x)$ when implemented in the Word RAM model;

  \item $\Select{}$, $\Insert{}$ and $\Delete{}$ are
  performed in amortized time $\BigO(\log n\cdot \hat{L})$;

  \item $\Unzip{x,y,\cdot}$ and $\Zip{x, y, \cdot}$ are performed in
  amortized time $\BigO((1+\log \frac{W}{w(x)} + \log \frac{W}{w(y)}) \cdot \hat{L})$;

  \item $\Level{x, y, \cdot}$ is performed in
  worst-case time $\BigO((1+\log \frac{W}{w(x)} + \log \frac{W}{w(y)}) \cdot \hat{L})$;

  \item $\SetWeight{x,\cdot}$, $\UpdateCounters{x, \cdot}$ and
  $\UpdateMarks{x, \cdot}$
  are performed in amortized time
  $\BigO((1+\log \frac{W}{w(x)}) \cdot\hat{L})$;

  \item $\SumCounters{x, \cdot}$ and $\OrMarks{x, \cdot}$
  are performed in worst-case time
  $\BigO((1+\log \frac{W}{w(x)}) \cdot\hat{L})$;

  \item $\FindMarked{x, \cdot}$ returns $\bot$ in worst-case time $\BigO((1+\log \frac{W}{w(x)}) \cdot \hat{L})$, or an~element $y$ in amortized time $\BigO((1+\log \frac{W}{w(x)}  + \log \frac{W}{w(y)})\cdot \hat{L})$.

  \item $\LongZip{}$ and $\LongUnzip{}$ are performed in worst-case
  time $\BigO(1)$;

  \item $\SelectedLevel{}$ is performed in worst-case time
  $\BigO(1)$.

  \end{itemize}
\end{lemma}

Note that in the statement above, a~data structure without the biasing extension can be emulated by assigning unit weights to all items in $X$.
Then each method $\Zip{}$, $\Level{}$, $\UpdateCounters{}$, $\SumCounters{}$, $\UpdateMarks{}$, $\FindMarked{}$ will take $\BigO(\log n \cdot \hat{L})$ time.

%\msinline{\footnotesize Not sure if we should talk more about counters here anymore (e.g., about the fact that we should ensure the polylog ``depth of computation'' for approximate counters). --- MS}

\newcommand{\MakeSet}{\Call{MakeSet}}
\newcommand{\Coalesce}{\Call{Coalesce}}
\newcommand{\Union}{\Call{Union}}
\newcommand{\Find}{\Call{Find}}
\newcommand{\Replace}{\Call{Replace}}
\newcommand{\RootUnion}{\Call{RootUnion}}
\newcommand{\SingletonUnion}{\Call{SingletonUnion}}

\paragraph{Biased Disjoint Sets}

At the lowest level of the sequence of our reductions lies the biased disjoint sets data structure. 
Given a set $X$ of \emph{items} with positive integer weights, a \emph{perfectly biased binary tree} for $X$ is a binary tree with $X$ as leaves and where the depth of each leaf $x\in X$ is $\BigO\paren*{1 + \log \frac{w(X)}{w(x)}}$.  If instead the depth of each leaf $x\in X$ is $\BigO\paren*{\log \frac{w(X)}{w(x)} + \log\log w(X)}$ we say the tree is \emph{almost biased}.

Our goal is to maintain a dynamic collection of almost biased binary trees whose leaf sets are disjoint, under the following operations:
\begin{itemize}\sloppypar
	
	\item \MakeSet{\textbf{free item} $x$, \textbf{weight} $w$} $\to$ \textbf{new
		root}:
	
	Create a new tree representing the set $X=\set{x}$ with weight
	$w(x)=w$ and return the new root.
	
	Afterwards, $x$ is a \textbf{singleton} and no longer a \textbf{free item}.
	
        \item \RootUnion{\textbf{root} $X$, \textbf{root} $Y$} $\to$ \textbf{new root}:

        Assumes $X\neq Y$. Construct the set $Z = X\cup Y$ and return the
        root of the tree representing it.

        Afterwards the sets $X$ and $Y$ no longer exist.

	\item \Find{\textbf{item} $x$} $\to$ \textbf{existing root}:
	
	Return the root of the tree representing $X\ni x$.
	
	\item \Delete{\textbf{item} $x$} $\to$ \textbf{new root} or $\bot$:
	
	Delete $x$ from the set $X$ containing it and return the root of the
	tree representing the (possibly empty) new set $X\setminus\set{x}$.
	
	Afterwards the set $X$ no longer exists, and $x$ is a \textbf{free item}.
	
  \item \Coalesce{\textbf{item} $x$, \textbf{item} $y$, \textbf{free
			item} $z$} $\to$ \textbf{new root}:
	
	Let $x\in X$ and $y\in Y$ be distinct (but possibly $X=Y$), and let
	$z$ be a new item which will be given the weight
	$w(z):=w(x)+w(y)$. Construct the set $Z = (X\cup
	Y\cup\set{z})\setminus \set{x,y}$ and return the root of the tree
	representing it.
	%% \jhinline{Note that $\frac{w(A)+w(B)}{w(a)+w(b)} \leq
	%% 	\max\set*{\frac{w(A)}{w(a)}, \frac{w(B)}{w(b)}}$. This simplifies the analysis when we use it in our application.}
	
	Afterwards the sets $X$ and $Y$ no longer exist, $x$ and $y$ are
	\textbf{free items} and $z$ is no longer a~\textbf{free item}.
	
	%% \jhinline{In this latest version, \textproc{Coalesce} can be easily
	%% 	done in the correct time using \textproc{Find}, \textproc{Delete},
	%% 	\textproc{MakeSet}, and \textproc{RootUnion}.  In earlier
	%% 	versions, it needed access to internals to be fast enough.}
	
	\item \Union{\textbf{item} $x$, \textbf{item} $y$} $\to$ \textbf{new
		or existing root}:
	
	Let $x\in X$ and $y\in Y$ be distinct (but possibly
	$X=Y$). Construct or find the set $X\cup Y$ and return the root of
	the tree representing it.
	
	Afterwards the sets $X$ and $Y$ no longer exist (unless
	$X=Y$).  %\jhinline{We don't actually use \Union{} in the
	% reduction, and it is trivial to construct from \Find{}
	% and \RootUnion{}. I only mentioned it because it is part of
	% the standard interface for disjoint set data structures.}

	%% \item \Replace{\textbf{item} $x$, \textbf{singleton} $x'$} $\to$ \textbf{new root}:
	
	%% Where $x\in X$ and $x'$ is a singleton item, construct the set
	%% $X'=(X\setminus\set{x})\cup\set{x'}$ and return the root of the tree
	%% representing it.  \jhinline{This can be used to
	%% 	change the weight of an item, as well as change which item with a
	%% 	given weight is stored. Note that if $w(x')\geq w(x)$ the running
	%% 	time reduces to $\BigO\paren*{ \log\frac{w(X)}{w(x)} + \log\log w(X') }$. }
	
	%% Afterwards the set $X$ no longer exists, and $x$ is a \textbf{singleton}.
	
	%% \jhinline{In this version, \textproc{Replace} is simply a
	%% 	\textproc{Delete} and a \textproc{RootUnion}. Earlier versions
	%% 	needed more tricks.}
	
\end{itemize}

The idea in our structure, inspired by Binomial Heaps and Fibonacci Heaps, is to maintain for each set $X$ a partition of $X$ into at most $2\log_2 w(X)$ subsets $X_1,\ldots,X_t$, represented by a~perfectly biased binary tree for each $X_i$, and then using a simple (e.g. weight-balanced) tree with the roots of those trees as leaves.  As long as we can ensure $t\in \BigO(\log w(X))$ this \emph{upper tree} can be guaranteed to have height $\BigO(\log\log w(X))$ and for $x\in X_i$ the depth of each node in the perfect \emph{lower tree} for $X_i$ is (by definition) $\BigO\paren*{1 + \log\frac{w(X_i)}{w(x)}}$. Thus the depth of any $x\in X$ in the complete tree is $\BigO\paren*{\log\frac{w(X)}{w(x)}+\log\log w(X)}$ as desired.

In \Cref{sec:bds} we give a detailed description of that data structure and finally prove the following Theorem:

\begin{theorem} \label{thm:bds}
	There exists a biased disjoint sets data structure where:
	\begin{itemize}
			\item $\MakeSet{}$ is performed in the worst-case $\BigO(1)$ time;
			
			\item $\RootUnion{}$ is performed in $\BigO\paren*{ \log\log(w(X)+w(Y)) }$ amortized time;
			
			\item $\Find{}$ is performed in $\BigO\paren*{ \log\frac{w(X)}{w(x)} + \log\log w(X)}$ worst-case time;
			
			%\erinline{$\BigO(\log \frac{w(X)}{\min\{w(x),
					%		w(y)\}} \log \lmax)$ can be rewritten as $\BigO(\log \frac{w(X)}{w(x)} \log \lmax + \log \frac{w(X)}{w(y)} \log \lmax)$}
			
			\item $\Delete{}$ is performed in $\BigO\paren*{ \log\frac{w(X)}{w(x)} + \log\log w(X)}$ amortized time;
			
			\item $\Coalesce{}$ is performed in $\BigO\paren*{ \log\frac{w(X)}{w(x)} +
				\log\frac{w(Y)}{w(y)} + \log\log(w(X)+w(Y))}$ amortized time;

			\item $\Union{}$ is performed in $\BigO\paren*{
				\log\frac{w(X)}{w(x)} + \log\frac{w(Y)}{w(y)}  + \log\log(w(X)+w(Y))}$ amortized time.
			
			%% \item $\Replace{}$ is performed in $\BigO\paren*{ \log\frac{w(X)}{w(x)} + \log\log w(X)
			%% 	+\log\log w(X')}$ amortized time.
			
	\end{itemize}
\end{theorem}

\section{Preliminaries}
In this work we work with undirected simple graphs without self-loops.
For a~graph $G$, we denote by $V(G)$ the set of vertices of $G$ and by $E(G)$ --- the set of edges.
A~\emph{forest} is a~graph without any cycles, and a~\emph{tree} is a~connected forest.
For two vertices $u, v$ in the same connected component of a~forest, we denote by $u \ldots v$ the unique simple path connecting $u$ and $v$.
If $u \neq v$, we define $s^u(v)$ as the vertex adjacent to $v$ on $u \ldots v$.
We will use the fact that standard dynamic tree data structures can determine $s^u(v)$ in time $\BigO(\log n)$, whilst supporting edge insertions and removals within the same complexity bounds~\cite{SleatorT81}.
For three vertices $u, v, w$, we define $\meet(u, v, w)$ --- the \emph{projection} of $w$ onto $u \ldots v$ --- as the unique vertex connected by simple paths to each vertex $u, v, w$.
We will write $P' \subseteq P$ to mean that $P'$ is a~subpath of $P$, e.g., if $e_1, e_2$ are two adjacent edges of a~graph, then $e_1e_2 \subseteq P$ means that the subpath containing two edges $e_1, e_2$ is a~subpath of $P$.
The distance between $u$ and $v$ in a~graph $G$ is denoted $\dist_G(u, v)$; we will drop the subscript when convenient.

While the considered graphs are undirected, it is sometimes more convenient to work with \emph{oriented} edges: for an~oriented edge $e = \vec{uv}$, we say that $u$ is the \emph{tail} of $e$ and $v$ is the \emph{head} of $e$.
In a~tree $T$, we define a~subtree $T[e]$ \emph{rooted at} $e = \vec{uv}$ as the subtree induced by the vertices of $T$ that are closer to $v$ than $u$ in $T$.

\subsection{Counter and bit vectors}
\label{ssec:prelims-counters}

In this work we will extensively use the concept of counters: nonnegative integers that can be added together and compared, but not subtracted from one another.
Various sets of counters will be used by us to monitor the satisfaction of various invariants related to the sizes of biconnected components maintained by our data structure.

If $\ell \in \mathbb{N}$, we can then consider a~\emph{counter vector} ${\bf c} = (c_0, c_1, \ldots, c_\ell)$ comprising $\ell + 1$ counters.
For convenience, let ${\bf 0}$ denote the all-zero counter.
We implement the following kinds of operations on counter vectors:
\begin{itemize}
  \item initialization with constant vectors;
  \item extracting and updating single elements of vectors;
  \item element-wise addition ${\bf c} + {\bf d}$, defined naturally;
  \item splicing: given counter vectors ${\bf c}$, ${\bf d}$ and $k \in \{0, \ldots, \ell+1\}$, define $\vectorsplice{\bf c}{k}{\bf d}$ as the counter vector comprising the length-$k$ prefix of ${\bf c}$ and the length-$(\ell-k+1)$ suffix of ${\bf d}$.
\end{itemize}
We do not support counter subtraction; this will become crucial later on.

Let $T(\ell)$ denote the maximum time it takes to perform a~single vector operation.
Assuming counter vectors are implemented using length-$(\ell+1)$ arrays of integers, we have $T(\ell) = \BigO(\ell)$; however, later in \Cref{sec:approx-counting} we will introduce a~notion of \emph{approximate counting} which will allow us to roll out counter vectors with $T(\ell) = \BigO(\log \ell)$.

Finally we consider \emph{counter matrices} ${\bf M} = ({\bf c}_0, {\bf c}_1, \ldots, {\bf c}_\ell)$ comprising $\ell + 1$ counter vectors, called rows, each of length $\ell + 1$.
We denote $M_{i,j} = (c_i)_j$.
We allow the following types of operations on counter matrices:
\begin{itemize}
  \item initialization with a~constant-zero matrix ${\bf 0}$;
  \item extracting and updating entries of matrices;
  \item addition of a~vector ${\bf v}$ to a~single row in ${\bf M}$: when ${\bf A} = \addvector({\bf M}, {\bf v}, r)$, we have $A_{i,j} = M_{i,j} + v_j \cdot [i = r]$; %\wn{$v_j$ rather than $v_i$, right?}
  \item splicing matrices: given two matrices ${\bf M}$, ${\bf N}$ and $k \in \{0, \ldots, \ell + 1\}$, define $\matrixsplice{\bf M}{k}{\bf N}$ as the counter matrix comprising $k$ first rows of ${\bf M}$ and $\ell-k+1$ last rows of ${\bf N}$;
  \item column sum of a~matrix: given a~matrix ${\bf M}$, let $\matrixsum({\bf M})$ be the vector ${\bf v}$ such that $v_j = \sum_{i = 0}^{\ell} M_{i,j}$. In other words, if ${\bf M} = ({\bf c}_0, {\bf c}_1, \ldots, {\bf c}_\ell)$, then $\matrixsum({\bf M}) = \sum_i {\bf c}_i$;
  \item column upper sum of a~matrix: given a~matrix ${\bf M}$, let $\uppersum({\bf M})$ be the vector ${\bf v}$ such that $v_j = \sum_{i = j}^{\ell} M_{i,j}$.
%  \wninline{I'd say that lowersum would be much more natural name than the uppersum...}
%  \msinline{I prefer not to touch this.}
\end{itemize}
We then have:
\begin{lemma}[\cite{Holm18a}]
  \label{lem:counter-matrices}
  All operations on counter matrices can be performed in time $\BigO(T(\ell) \log \ell)$.
\end{lemma}
Note that \Cref{lem:counter-matrices} \emph{does not} enable us to perform efficient element-wise additions of matrices of the form ${\bf M} + {\bf N}$; we shall avoid this kind of additions in the implementation of our data structure.

We also consider \emph{bit vectors} ${\bf b} = (b_0, b_1, \ldots, b_\ell)$ with $\ell + 1$ entries, where ${\bf 0}$ denotes the all-zero bit vector.
Similarly to counter vectors, we also consider the bitwise OR: $({\bf a} \bitor {\bf b})$, splicing, and extracting single elements of bit vectors.
%Additionally, we support \emph{bit shifts}: given a~bit vector ${\bf b}$ and an~integer $c \in \mathbb{N}$, let ${\bf b} \bitshl c$ \wn{That notation, as it is displayed now, is hideous. Also, if the ``l'' was meant to denote ``left'' then I'd point out it is shift to the right as defined.} be the bit vector whose $i$th coordinate is $b_{i-c}$ if $i \geq c$, or $0$ otherwise.
Let $B(\ell)$ denote the time required to perform a~single bit vector operation; note that $B(\ell) \leq T(\ell)$ since bit vectors can always be simulated via counter vectors.
%\wninline{Counter vectors did not allow for any kinds of shifts, so that inequality is unclear.}
In the setting of combinatorial algorithms we have $B(\ell) = \Oh(\ell)$.
Meanwhile, in the word RAM setting we assume that $B(\ell) = \Oh(1)$ whenever $\ell \in \BigO(\log n)$.
In the same vein, we can define \emph{bit matrices} with $\ell + 1$ bit vectors, replacing all additions in the definition of counter matrices with the corresponding bitwise ORs in a~natural way.
The proof of Lemma~\ref{lem:counter-matrices} applies also to bit matrices, yielding that:
\begin{lemma}[\cite{Holm18a}]
  \label{lem:bit-matrices}
  All operations on bit matrices can be performed in time $\BigO(B(\ell) \log \ell)$.
\end{lemma}

\section{Graph structure} \label{sec:graph-structure}
%\todo[inline]{%
%  \begin{itemize}
%  \protect\item level structure, cover levels, etc
%  \protect\item delete by decreasing one level at a time
%  \protect\item pseudocode
%  \end{itemize}
%}
This section will be devoted to proving the \Cref{lem:first-reduction}, which states that we can reduce all biconnectivity updates and queries to the cover level tree problem. We will build upon the ideas from \cite{HolmDeLichtenbergThorup} and follow some definitions from there, but we will adjust them to new improvements\ms{what are the improvements?} and to provide a cleaner abstraction.

The main reason why biconnectedness proves to be more challenging than 2-edge connectivity seems to be the fact that it is an equivalence relation over edges rather than over vertices. In other words, the edges of $G$ can be partitioned into inclusion-wise maximal subsets of edges forming biconnected components. The key notions for understanding this partitioning will be \emph{covered adjacent pairs} and \emph{transitively covered adjacent pairs}. 

As aforementioned, we will be maintaining a spanning forest $F$ of the graph $G$ and each non-tree edge of the graph will be associated a level $i \in \{0, 1, \ldots, \lmax\}$, where $\lmax \in \Oh(\log n)$. Such an~assignment induces a sequence of graphs $G = G_0 \supseteq G_1 \supseteq \ldots \supseteq G_{\lmax} \supseteq F$, where $G_i$ is the subgraph of $G$ whose edge set contains precisely $E(F)$ and the non-tree edges of level at least $i$.

For $x, y, z\in V(G)$ such that $xy, yz \in E(G)$, we will say that $(xy, yz)$ is an \emph{adjacent pair} (the pairs are treated as unordered). If additionally we have that $xy, yz \in E(F)$, we will say that $(xy, yz)$ is an \emph{adjacent tree pair}. Now, let $uv$ be a non-tree edge at level $i$. Then, $uv$ \emph{covers} all adjacent pairs on a cycle induced by $uv$ in $F$, that is, all adjacent pairs $(xy, yz) \subseteq u \ldots v$, $(vu, us^v(u))$ and $(uv, vs^u(v))$ (we say that $(xy, yz) \in u \ldots v$ if and only if both $xy$ and $yz$ belong to $u \ldots v$).

%\wninline{I am probably a bit reckless with using terms ,,transitively covered at level $i$'', ,,covered at level $i$'' and ,,a pair of edges with cover level $i$'' kinda interchangeably. I hope to take care of that later.} 

Then, we define \emph{transitively covered adjacent pair} as follows. All covered adjacent pairs are transitively covered. Additionally, for $x, y, z, t \in V(G)$ and $xy, xz, xt \in E(G)$, if $(yx, xz)$ and $(zx, xt)$ are transitively covered, then $(yx, xt)$ is transitively covered as well. 
In this section, we will use the following properties of transitively covered adjacent pairs shown in~\cite{HolmDeLichtenbergThorup}:

\begin{lemma}[{\cite[Lemma 18]{HolmDeLichtenbergThorup}}]\label{lem:triples}
    The following properties hold:
	\begin{enumerate}
		\item Biconnectivity is a transitive relation over the neighbors of a vertex $u$, and if two neighbors of $u$ are biconnected, $u$ is in the biconnected component containing them.
		\item An adjacent pair $(xy, yz)$ is transitively covered if and only if $x$ and $z$ are biconnected.
		\item A vertex $y$ is an articulation point\ms{cut-vertex?} if and only if there is an adjacent tree pair $(xy, yz)$ which is not transitively covered.
		\item \label{item:triples-path-biconnected} Two vertices $v$ and $w$ are pseudo-biconnected if and only if for all $(xy, yz) \subseteq v \ldots w$, $(xy, yz)$ is transitively covered.
	\end{enumerate}
\end{lemma}
%\wninline{Currently this Lemma is wrong in the last item for the case where $vw$ is an edge, thanks to the difference in the definitions between the papers}

%\wninline{The following is messed up according to current faulty definition of cover level. Come back here after discussing with Marek.}

Based on \Cref{lem:triples} we note that if the cover level of an adjacent tree pair $(xy, yz)$ equals~$i$, then $i$ is the largest integer such that $(xy, yz)$ is transitively covered in $G_i$. We also say that $(xy, yz)$ is \emph{transitively covered at level $j$} for all $j \le i$.

\Cref{item:triples-path-biconnected} of \cref{lem:triples} motivates the created abstraction of the tree problem (and neighborhood data structure in turn), as maintaining cover levels of adjacent tree edges and corresponding equivalence classes allows us to determine the level at which any two vertices are pseudo-biconnected.
%We recall that the responsibility of the tree cover level data structure is, among others, to maintain the spanning forest $F$, handle updates about changes to cover levels of adjacent tree pairs and answer queries 

Similarly to some of the previous algorithms \cite{HolmDeLichtenbergThorup, Holm18a} for dynamic connectivity problems, we are going to maintain a key invariant:
\begin{description}
	\item[($\dagger$)] Biconnected components in graph $G_i$ have at most $\ceil{\frac{n}{2^i}}$ vertices.
\end{description}
%\wninline{I guess that enumerate is not the best way to highlight a statement?}

As such, the maximum level of an edge cannot exceed $\log_2{n}$. Inserting an edge at level $0$ cannot violate this invariant. We will say that it is \emph{legal} to increase the level of a non-tree edge $e$ to $j$ if this does not violate this invariant, that is, if the biconnected component of $e$ in $G_j \cup \{e\}$ has at most $\ceil{\frac{n}{2^j}}$ vertices. Increasing the level of an edge by one will be called \emph{promoting}. Moreover, throughout the lifetime of a non-tree edge, all operations of increasing its level will precede all operations of decreasing its level (which will happen only just before the edge gets deleted), hence only $\Oh(m \log n)$ level changes will be issued in total.

Deleting an edge may cause the cover levels of multiple adjacent pairs to be lowered and many biconnected classes on various levels to split. Let us focus on the case where the deleted edge $uv$ is a non-tree edge on level $i$. It is the case that before its deletion it caused all adjacent pairs from $u \ldots v$ to be covered at level at least $i$, but it might have been not the only reason: there could have been other non-tree edges at level $i$ or higher covering some of the adjacent pairs on $u \ldots v$. For adjacent pairs covered at level higher than $i$, we know that deleting $uv$ does not affect their level, but it is unclear what happens for adjacent pairs on level exactly $i$. In order to determine what happens with such pairs, we will be \emph{looking through} a hypothetically large number of non-tree edges of level $i$ that might influence cover levels of adjacent pairs on $u \ldots v$, and we will need to do so in a carefully chosen order. The key idea to bound the time complexity of this process --- despite looking at a potentially large number of edges --- is to amortize our work by promoting edges that we are looking through to higher levels. However, promoting edges will not be always legal, as we may break the invariant ($\dagger$) as an effect. Nevertheless, we will show that per single edge deletion, we will be able to promote all the \emph{looked-through} edges, except for at most $\Oh(\lmax)$ of them.

We point out that the tree cover level data structure is unaware of the existence of non-tree edges. The responsibility of maintaining these edges and making the appropriate calls to the tree cover level data structure lies on the main biconnectivity data structure that we implement right now. To this end, we will maintain sets $N^i(v)$ that for each $i \in \{0, 1, \ldots, \lmax\}$ and $v \in V(G)$ store vertices $u$ such that $uv$ is a level $i$ non-tree edge of $G$. 
%However, because of some technical reasons to be explained later, $N^i(v)$ will at some moments temporarily store edges that were already pushed to the higher level.
Such sets can be maintained using any balanced binary tree supporting inserting, deleting, checking existence of an element and providing any of its elements (if nonempty) in $\Oh(\log n)$ time per query. We will also maintain the following invariant:

\begin{itemize}
	\item[($\dagger\dagger$)] A vertex $u$ is $i$-marked if and only if $N^i(u)$ is nonempty. 
\end{itemize}
In order to maintain it, we design an auxiliary function $\Call{UpdateMark}{u, i}$ that checks if $N^i(u)$ is currently empty or not, compares that with whether it is currently $i$-marked or not, and calls $\Call{Mark}{u, i}$ or $\Call{Unmark}{u, i}$ if necessary. 

\paragraph{\textsc{Insert}, \textsc{AreBiconnected}, and \textsc{NextCutVertex} calls}
The \textsc{Insert} and \textsc{AreBiconnected} calls are very easily implementable using the tree data structure.

For an $\Insert{u, v}$ call, we first check whether $u$ and $v$ are connected in $F$ by calling $\Call{Connected}{u, v}$. If they are not, then we simply call $\Call{Link}{u, v}$. If they are, we call $\Call{Cover}{u, v, 0}$, insert $u$ to $N^0(v)$, insert $v$ to $N^0(u)$ and call $\Call{UpdateMark}{u, 0}$ and $\UpdateMark{v, 0}$.

For an $\AreBiconnected{u, v}$ call, we first call $\Call{Connected}{u, v}$. If it returns that $u$ and $v$ are not connected by a tree path, we return that they are not even in the same connected component. Otherwise, we call $\Call{CoverLevel}{u, v}$. If it returns $-1$, then we return that $u$ and $v$ are not biconnected. However, if $\Call{CoverLevel}{u, v}$ returns a~non-negative integer, we conclude that $u$ and $v$ are pseudo-biconnected. It remains to check if $u$ and $v$ are connected by a bridge. However, that boils down to checking if $\FindSize{u, v, 0}=2$. If this is the case, then we return that $u$ and $v$ are not biconnected, or that they are biconnected otherwise.
%\wninline{Come back here after finalizing the definition}

The $\NextCutVertex{v, w}$ call proceeds similarly. If $\CoverLevel{u, v}$ returns $-1$, we return $\Call{MinCoveredPair}{v, w}$ as the articulation point between $v$ and $w$ that is the closest to $v$. Otherwise, we return $w$.

\paragraph{Auxiliary functions}
Before we begin to describe edge deletion function, let us introduce two auxiliary functions: $\PromoteEdge{x, z}$ and $\FindNextEvent{u, v, i}$.

%The goal of $\Call{Promote}{x, z, i}$ is to ,,inform vertex $x$'' that the edge $xz$ got pushed from level $i$ to level $i+1$. During \textsc{UncoverPath} we will actually retrieve many edges twice --- once per each end. Such edge covers (among other triples) an interval of $u \ldots v$ and it is important for us to process it both at $\mathrm{meet}{u, v, x}$ and $\mathrm{meet}{u, v, z}$. We retrieve these events thanks to $x$ and $z$ being $i$-marked, hence, when pushing an edge on a higher level (what happens with first of these two events) we should not update $N^i$ on both ends (and their corresponding $i$-marks) as we will have no way of retrieving the second event. Hence, the function $\Call{Promote}{x, z, i}$ will update $N^{i}(x), N^{i+1}(x)$ and $i$ and $i+1$ marks of $x$, but will not update these for $z$. It may be convenient to think of it in terms of pushing a directed edge.

The goal of $\PromoteEdge{x, z, i}$ is to perform the necessary bookkeeping connected to pushing an edge $xz$ from level $i$ to level $i+1$, which is to update sets $N^i$ and $N^{i+1}$ for $x$ and $z$, update their $i$ and $i+1$ marks accordingly and call $\Cover{x, z, i+1}$. Its pseudocode can be found as \Cref{alg:promote-edge} in \Cref{sec:graph-impl}.

In order to explain the role of $\FindNextEvent{u, v, i}$ we need to introduce a few concepts. For two vertices $x, y$ on $u \ldots v$ we will say that $x$ is on \emph{left} (or \emph{right}, respectively) of $y$ if and only if $x$ is closer to $u$ than $y$ (or $y$ is closer to $u$ than $x$, respectively). For a vertex $x$ on $u \ldots v$ we analogously define its left edge $L(x)$ and right edge $R(x)$ as the edges it is incident to on that path, where the left edge is closer to $u$, that is $L(x) \coloneqq (x, s^u(x)), R(x) \coloneqq (x, s^v(x))$ (note that $u$ does not have the left edge and $v$ does not have the right edge).
For a non-tree edge $xy$ we define the \emph{projection of the edge $xy$ onto $u \ldots v$} as either the projection of $x$ onto $u \ldots v$ or the projection of $y$ onto $u \ldots v$ --- whichever is closer to $u$. We may drop the ``onto $u \ldots v$'' if it is clear from the context. We say that an edge $xy$ is $i$-reachable from $u \ldots v$ if and only if it is of level $i$ and $x$ is $i$-reachable from $u \ldots v$ (it can be easily seen that it is equivalent to $y$ being $i$-reachable from $u \ldots v$). An \emph{event} is a tuple $(e, p, f)$, where $e$ is a level-$i$ non-tree edge, $p$ is the projection of $e$ onto $u \ldots v$ (also called the \emph{projection of an event}) and $f$ is either $L(p)$ or $R(p)$. For a tuple $(e, p, f)$ to be called an event, we additionally require that $e$ is $i$-reachable from $u \ldots v$ --- otherwise $xy$ would not influence cover levels of any adjacent pairs from $u \ldots v$ assuming that the cover level of $u \ldots v$ is equal to $i$.

The goal of $\FindNextEvent{u, v, i}$ is to use $\textsc{FindFirstReach}$ and $\textsc{FindStrongReach}$ from the tree data structure to yield us events in the order from left to right of their projections onto $u \ldots v$.
We do so by extracting: a vertex $x$ that is $i$-marked and $i$-reachable from the leftmost possible edge; the projection $p$ of $x$ onto $u \ldots v$; and an arbitrary element $y \in N^i(x)$. However, we specify that $\FindNextEvent{}$ adheres to the following essential tiebreaking rule:
\begin{description}
	\item[($\dagger\dagger\dagger$)] If we retrieve a vertex $x$ (with projection $p$) that is not strongly $i$-reachable from $L(p)$, we require that there are no vertices that are $i$-marked and strongly $i$-reachable from $L(p)$ through~$p$.
\end{description}
In order to distinguish whether $x$ is strongly $i$-reachable from $L(p)$, as a part of the output $\FindNextEvent{u, v, i}=(xy, p, f)$, we either define that $f=L(p)$ in the case $xy$ is strongly $i$-reachable from $L(p)$, or $f=R(p)$ otherwise. This turns out to be a natural choice: we will later prove that $xy$ and $f$ will eventually belong to the same biconnected component of $G_i$. The pseudocode of $\FindNextEvent{}$ is presented as \cref{alg:find-next-event} in \Cref{sec:graph-impl}.

\subsection{Reducing deletions to non-tree edges}
The main challenge in the implementation of biconnectivity data structures lies in edge deletions. Let us consider a call $\Delete{u, v}$. The edge $uv$ may either be a tree edge or a non-tree edge. We will first reduce to the case where it is a non-tree edge.

Assume that $uv$ is a tree edge. Firstly, we want to determine the highest level of a non-tree edge that connects $F_u$ and $F_v$ that are connected components of $F - \{uv\}$ containing $u$ and $v$ respectively. If no such edge exists, then we can simply call $\Call{Cut}{u, v}$ and conclude the update. However, if such an~edge exists, let us denote the highest level of such edge by $i$, and call any such edge at level $i$ a \emph{replacement edge}. Equivalently, $i$ is the highest level such that there exists a vertex $w$ such that either $(wu, uv)$ or $(wv, vu)$ is an adjacent tree pair covered on level $i$. This level can be determined by repeatedly calling $\FindSize{u, v, j}$ for $j= \lmax, \lmax - 1, \ldots, 0$ and stopping at first $j$ such that $\FindSize{u, v, j} > 2$ (if it did not stop at any $j$, then we conclude there is no such edge at any level). Our goal is to identify the replacement edge $e$ and swap $uv$ with $e$, that is, make $e$ a tree edge and make $uv$ a non-tree edge at level $i$. We will firstly show that such a~swap is a safe operation to perform for our data structure:

\begin{lemma} [\cite{HolmDeLichtenbergThorup}] \label{lem:swap-cover-levels}
	Let $G=G_0 \supseteq G_1 \supseteq \ldots \supseteq G_{\lmax}$ and $G'=G_0' \supseteq G_1' \supseteq \ldots \supseteq G_{\lmax}'$ be respective graphs before and after the swap operation. Then, $G_j$ and $G_j'$ have the same biconnected components and for each pair $(xy, yz)$ that is an adjacent tree pair in both, its cover level is the same in $G$ and $G'$.
\end{lemma}
\begin{proof}
	Let us consider two cases based on whether $j \le i$ or $j>i$.
	\begin{description}
		\item[Case 1:] $j \le i$.
		In this case we have that $G_j = G_j'$, so their biconnected components are clearly the same.

		\item[Case 2:] $j > i$.
		In this case $uv$ is a bridge in $G_j$ and $e$ is a bridge in $G_j'$. Moreover $G_j \setminus \{uv\} = G_j' \setminus \{e\}$. Hence biconnected components in both are the same.
	\end{description}
	As biconnected components are preserved at every level, cover levels of all pairs of edges are preserved too. Hence, if we restrict only to pairs that are adjacent tree pairs in both $G$ and $G'$, we get the desired conclusion.
\end{proof}

The specific realization of swapping a tree edge $uv$ with a non-tree edge $xy$ on level $i$ is as follows. We call $\Cut{u, v}$, then we remove $y$ from $N^i(x)$ and $x$ from $N^i(y)$, call $\Link{x, y}$, add $v$ to $N^i(u)$ and $u$ to $N^i(v)$ and finally call $\Cover{u, v, 0}, \Cover{u, v, 1}, \ldots, \Cover{u, v, i}$. As observed in \Cref{lem:swap-cover-levels}, swapping does not affect cover levels, hence the only adjacent tree pairs in $G'$ with incorrect values of cover levels in our data structure before the final sequence of $\Cover{}$ calls are the ones that are not tree pairs in $G$, that is, the ones involving the edge $xy$. The following sequence of $\Cover{}$ calls correctly updates these and does not affect any other cover levels, so the described sequence of operations correctly executes the desired swapping operation.  

%What remains to be done is to identify reconnecting edge $e$ at the respective level. For that we follow the same logic as \cite{Holm18a} does in Lemma 2.1 in Appendix A, where instead of $\textsc{FindFirstLabel}$ function we use analogous $\textsc{FindFirstReach}$. In short, we first call $\Cut{u, v}$. That causes the biconnected component $B$ at level $i$ containing $u$ and $v$ to split into two parts $B_u$ and $B_v$ in $G_i$. We compare their sizes by calling $\FindSize{u, u, i}$ and $\FindSize{v, v, i}$ and focus on the smaller one (assume without loss of generality that it is $u$). As $|V(B_u)| \le |V(B_v)|$ and $\ceil{\frac{n}{2^i}} \ge |V(B)| = |V(B_u)| + |V(B_v)| \ge 2|V(B_u)| \Rightarrow \ceil{\frac{n}{2^{i+1}}} \ge |V(B_u)|$, if $P$ is some subset of level $i$ edges of $B_u$, if they will all be pushed to level $i+1$, then we will still maintain our ($\dagger$) invariant. Hence, we keep calling $\Call{FindFirstReach}{u, u, i}$ that retrieves the non-tree edges coming out from vertices of $B_u$ one by one. 
%\wninline{Calling $\FindFirstReach{u, u, i}$ is illegal, I have to rethink that. Calling $\FindFirstReach{u, v, i}$ may be worth considering with some additional maintenance.}
%If an edge happens to be within $B_u$, then we push it to the next level and proceed. If it is not, then we know that we found a level $i$ edge connecting $B_u$ and $B_v$ and we stop the search for the replacement edge (it has to happen at some point).

\paragraph{Identifying a replacement edge}
What remains is to identify the replacement edge $e$ at the respective level. For that we follow similar logic as in \cite[Lemma 2.1]{Holm18a}, but instead of $\textsc{FindFirstLabel}$ function we use analogous $\textsc{FindNextEvent}$, and for technical reasons we browse edges from both sides rather than from the smaller one only. We operate in two symmetric phases --- one for $(a, b) \coloneqq (u, v)$ and one for $(a, b) \coloneqq (v, u)$ --- corresponding to searching edges on both sides of $uv$. Consider a single phase (and note that $ab \in E(F)$). We repeatedly call $\FindNextEvent{a, b, i}$ to retrieve a candidate $xy$ for a replacement edge, where $x$ is $i$-reachable from $ab$ and --- as we will see --- it will be reachable from $ab$ through $a$ (as opposed to through $b$). If for a~particular $y$ we have that $s^y(a)=b$, then we know that the edge $xy$ is a~good replacement edge and we stop the search. If $s^y(a) \neq b$, then we know that it is contained within $V(F_a)$ and check if it is legal to promote it. If it is, then we promote it. If it is not, then we stop the current phase.

According to our definition of $i$, a~replacement edge exists and is $i$-reachable from $uv$ both through $u$ and through $v$, so no phase can be stopped by running out of candidates for replacement edges, or even reach the moment when $\FindFirstReach{a, b, i}$ starts returning edges reachable through $b$ rather than from $a$. Hence, a phase can be stopped only by either finding a replacement edge, or by finding an edge that cannot be promoted. Let $B$ be the biconnected component of $G_i$ containing $uv$ and let $B_u = B[F_u]$ and $B_v = B[F_v]$ be subgraphs of $B$ induced by vertices of $F_u$ and $F_v$, respectively. As $\ceil{\frac{n}{2^{i}}} \ge |V(B)| = |V(B_u)| + |V(B_v)|$, we get $\mathrm{min}(|V(B_u)|, |V(B_v)|) \le \ceil{\frac{n}{2^{i+1}}}$. Let us assume that $|V(B_u)| \le |V(B_v)|$, so $|V(B_u)| \le \ceil{\frac{n}{2^{i+1}}}$. We will prove that in this case the phase for which $(a, b) = (u, v)$ cannot stop by not being able to promote an edge, so it has to be stopped by finding a replacement edge. The proof for the case when $|V(B_u)| > |V(B_v)|$ is analogous.
	
	We know that every biconnected component $B'$ of $G_{i+1}$ is either edge-disjoint from $B$ or contained in $B$. As there are no edges of level $i+1$ or higher connecting $F_u$ and $F_v$, any biconnected component $B'$ that is not disjoint from $B$ is either contained in $B_u$, contained in $B_v$, or is the edge $uv$. As such, all hypothetical biconnected components on level $i+1$ formed by promoting any subset of level $i$ edges within $B_u$ will be still contained within $B_u$. Hence, these biconnected components will be of size at most $|V(B_u)| \le \ceil{\frac{n}{2^{i+1}}}$, so it will be legal to promote any subset of level-$i$ edges within $B_u$. Hence, the phase for $(a, b) = (u, v)$ has to be stopped by finding a replacement edge, which concludes the correctness proof of this algorithm.
	
The pseudocode of finding the replacement edge ($\Call{FindReplacement}{}$) and executing the replacement ($\Call{Swap}{}$) can be found in \Cref{alg:swap} in \Cref{sec:graph-impl}.

% We call $\Cut{u, v}$, then we remove $y$ from $N^i(x)$ and $x$ from $N^i(y)$, call $\Link{x, y}$, add $v$ to $N^i(u)$ and $u$ to $N^i(v)$ and finally call $\Cover{u, v, 0}, \Cover{u, v, 1}, \ldots, \Cover{u, v, i}$.

%\wninline{The Swap pseudocode lacked $i=-1$ handling... Check if it's fine now.}

%\wninline{It seems to me I have to make it more complicated than in bridge finding paper, where it is cutting an edge, computing sizes and browsing through edges from the smaller one. The reason is because I need the edge $uv$ to still exist to retrieve candidates for reconnecting edges through $\FindFirstReach{u, v, i}$. Having that edge prevents me from getting to know sizes of left and right parts of the $i$-reachable component of $uv$. I cannot also cut it, compute the sizes and reintroduce it again for the sake of retrieving candidate edges, because its equivalence relations classes will be destroyed and I don't think there is an easy way to fix them. There are probably some easy-ish ways of enriching the tree structure with helpful operations (either roll back Cut effect or get any $i$-reachable neighbor of $u$ from $uv$), but better not do that if I can handle that myself} 

\subsection{Deleting non-tree edges}
Hence, from now point on, we can assume that $uv$ is a non-tree edge.
Assume that this edge is on level $i$. The issue is that it is currently covering on level $i$ all pairs of adjacent edges on the cycle induced by $uv$ and it is not immediately clear how to remove its contribution from the maintained transitive covering of adjacent pairs. Naturally, if an adjacent pair on that cycle is transitively covered at level higher than $i$, then it will stay that way. But an~adjacent pair transitively covered at level exactly $i$ may be covered at any level between $-1$ and $i$ after the removal of $uv$, depending on whether it is covered by some other non-tree edge.

Before deleting an edge we will remove its contribution to the maintained transitive covering of adjacent pairs step by step. We will progressively decrease its level one step at a time and issue corresponding updates to the tree structure. We will design a procedure $\UncoverPath{u, v, i}$, whose goal is to update the cover levels of all adjacent tree pairs in the tree structure after dropping the level of edge $uv$ from $i$ to $i-1$. The plan is to call $\UncoverPath{u, v, j}$ for $j=i, i-1, \ldots, 0$ and remove the edge $uv$. Since $\UncoverPath{}$ will use the operation $\UniformUncover{}$ of the (restricted) tree cover level data structure, we have to first call $\Expose{u, v}$ in this data structure before we start calling $\textsc{UncoverPath}$. Additionally, we also need to remember about updating $N^j(\cdot)$ sets accordingly. This is expressed through the pseudocode in \cref{alg:delete} in \Cref{sec:graph-impl}.

%\wninline{This pseudocode is unnecessarily verbose in terms of updating $N^j(u)$ sets. I could just remove $u$ from $N^i(v)$ and call $\UpdateMark{v, i}$ (and vice versa for $v$ and $u$) and be done with it, but I feel the current version better fits lowering the level one by one}
%\msinline{\footnotesize I think that's whatever}

\paragraph{Uncovering a path}

In order to properly investigate changes in cover levels when dropping a level of a non-tree edge $uv$ from $i$ to $i-1$ we will look through other non-tree edges that influence cover levels of pairs of edges incident to vertices on the $u \ldots v$ path. The main idea is to again amortize the work by promoting the looked-through edges. However, doing so carelessly could break the invariant ($\dagger$). In order to remedy this, we will firstly browse the edges in the direction from $u$ to $v$ until we encounter first moment when pushing an edge would break the invariant ($\dagger$).
%More specifically, we will browse $i$-reachable from $u \ldots v$ ends of level $i$ edges in the order of their projections to $u \ldots v$ (with some tiebreaking rules to be specified) and we accordingly think of $u \ldots v$ as determining our timeline of events.
We retrieve these edges by utilizing the $\FindNextEvent{u, v, i}$ function, which does so in an already specified order, which is in the order of their projections to $u \ldots v$ with the tiebreaking rule $(\dagger\dagger\dagger)$ applied when needed. We intuitively think of the $u \ldots v$ path as prescribing our timeline of events.
Then, we repeat the same procedure, but this time in the direction from $v$ to $u$. Along the way, based on various cases, we will call $\textsc{LocalUncover}$ and $\textsc{UniformUncover}$ with various arguments in order to cancel the influence brought to the transitive covering by $uv$, and call $\Cover{\cdot, \cdot, i+1}$ to promote some level-$i$ edges to level $i+1$. We will prove that doing so will be sufficient to fully accommodate changes to cover levels stemming from lowering the level of $uv$ from $i$ to $i-1$.

Let $G$ be the graph with $uv$ at level $i$ and $G'$ be the graph after lowering the level of $uv$ to $i-1$. 

\msinline{\footnotesize If time: recall what $\LL^x_j$ means. It's the first time this notation appears in this section}

\begin{lemma} \label{lem:pomegranade}
Let $x \in V(G)$ and $j \in \{0, \ldots, \lmax\}$.
\begin{itemize}
	\item If $x \not\in u \ldots v$ or $x \in \{u, v\}$ or $j \neq i$, then $\LL^x_j$ remains unchanged after lowering the level of $uv$ from $i$ to $i-1$.
	\item Otherwise, it either remains unchanged, or $\LL^x_i(L(x))$ should be split into exactly two subclasses, where one of them contains $L(x)$ and the other contains $R(x)$.
\end{itemize}
\end{lemma}
\begin{proof}
	The first part of the lemma is trivial. For the other part, consider the inverse of the operation just performed, that is, view $G$ as the graph obtained from $G'$ by increasing the level of $uv$ from $i-1$ to $i$. In that case the only pair of edges incident to $x$ whose (non-transitive) covering may change is $L(x)R(x)$. If $L(x)$ and $R(x)$ belong to the same equivalence class on level $i$ in $G'$, then the transitive covering around $x$ does not change. If they do not, then transitive covering around $x$ changes by merging respective equivalence classes of $L(x)$ and $R(x)$ on level $i$, which
	%\ms{\footnotesize what $\to$ which, also elsewhere}
	proves the lemma.
\end{proof}

We are finally ready to proceed with describing $\textsc{UncoverPath}$. 

During this procedure, some vertices on $u \ldots v$ that will be projections of some processed events (let us recall that the projection of an event is the projection of one of its two ends --- the one that is closer to $u$) --- we will call these vertices \emph{important}. The vertices that are between two consecutive important vertices, before the projection of the first event, or after the projection of the last event (if the procedure was not stopped) will be called \emph{skipped through}. In the special case of no important vertices, all internal vertices of $u \ldots v$ will be classified as skipped through. In order to properly adjust the cover levels of left and right edges of important vertices we will call $\textsc{LocalUncover}$, while in order to properly adjust cover levels of skipped through vertices we will call $\textsc{UniformUncover}$. We exclude $u$ and $v$ from that classification as $\LL^u_i$ and $\LL^v_i$ are not affected by lowering the level of $uv$. To recall, we proceed in two symmetric phases --- one going from $u$ to $v$ and one going from $v$ to $u$. We assume in the following description that we are in the first phase as the analysis of the second phase is analogous. We also assume that the level of $uv$ has already been lowered to $i-1$, but the state of the data in the corresponding tree cover level data structure maintaining classes $\LL$ is yet to accurately reflect that.

\subparagraph{Skipped through vertices}
Let $p$ be a vertex that was skipped through and let us consider the moment when we are skipping through it, which is either the moment of retrieving the first event whose projection is
on the right of $p$, or after the very last event; and let $G''$ be our graph at that moment.
\newcommand{\KK}{\mathcal{K}}
Let $\LL^p$ denote the equivalence classes for $p$ currently represented in the tree cover level data structure and $\KK^p$ denote the desired final state of $\LL^p$ after lowering the level of $uv$ from $i$ to $i-1$.

We note that no update affecting $\LL_i^p$ has been issued during this phase yet, so in particular we still have that $\LL_i^p(L(p)) = \LL_i^p(R(p))$.
We claim the following:
\begin{lemma} \label{lem:skip-clean}
	If we call $\UniformUncover{\cdot, \cdot, i}$ on a subpath of $u \ldots v$ containing $L(p)R(p)$, then $\LL^p_i$ will get correctly updated in $G''$.
\end{lemma}
\begin{proof}

%Let $F_R, F_L$ be subtrees of $F$ rooted at $p$ such that $F_R$ consists of all vertices reachable in $F$ from $p$ through $R(p)$\ms{\footnotesize from $u \ldots v$ at $R(p)$ through $p$?} and $F_L$ consists of all vertices reachable in $F$ from $p$ through any other edges of $\LL^p_i(L(p))$\ms{\footnotesize to stress: we define reachability only at edges of $u \ldots v$} (we note that $R(p) \in \LL^p_i(L(p))$).
Let $F_R, F_L$ be subtrees of $F$ rooted at $p$ such that $F_R$ consists of $p$ and all vertices $x$ such $(p, s^x(p)) = R(p)$ and $F_L$ consists of $p$ and all vertices $x$ such that $(p, s^x(p)) \in \LL^p_i(L(p)) \setminus \{R(p)\}$ (we note that $R(p) \in \LL^p_i(L(p))$).
According to the order of processing events, we know that there is no $i$-marked vertex in $V(F_L) \setminus \{p\}$ that is $i$-reachable from $u \ldots v$. As such, there are no level-$i$ edges covering a pair of edges from $\LL^p_i(L(p))$ since all $i$-marked vertices $x$ that are $i$-reachable from $u \ldots v$ have $s^x(p) = s^v(p)$ (or in other words, $(p, s^x(p)) = R(p))$. Therefore, the cover level of $L(p)R(p)$ cannot be $i$ after the update. If $L(p)R(p)$ is transitively covered at level higher than $i$ in $G''$, then the tree data structure has that information updated already --- either because it was transitively covered at level higher than $i$ already in $G$ (which lowering the level of $uv$ does not affect), or because it got transitively covered at level $i+1$ during one of $\Cover{\cdot, \cdot, i+1}$ calls that were issued when promoting some level-$i$ edges earlier in the current call of $\textsc{UncoverPath}$.
As we know that the cover level of $L(p)R(p)$ cannot be $i$ after the update, but it should be at least $i-1$ (as $uv$ has level $i-1$) --- if $L(p)R(p)$ is not already transitively covered at level higher than $i$ in our tree cover level data structure --- the only remaining option is that it should be lowered to $i-1$.

Based on the first part of \cref{lem:pomegranade} we conclude that
\begin{description}
	\item[\textbf{($*$)}] $\KK^p_{i+1}(L(p)) = \LL^p_{i+1}(L(p))$ and $\KK^p_{i+1}(R(p)) = \LL^p_{i+1}(R(p))$
\end{description} 
We remind that we have that $\KK^p_{i+1}$ is a refinement of $\KK^p_{i}$, where $\KK^p_i$ can be obtained from $\KK^p_{i+1}$ by merging classes containing pairs of elements that are covered by some level-$i$ non-tree edges.
But, as already argued, there are no level-$i$ non-tree edges covering any pair of edges from $\LL^p_i(L(p))$, so in particular there are no level-$i$ edges covering any pair of edges from $\KK^p_{i+1}(L(p))$, since $\KK^p_{i+1}(L(p)) \stackrel{(*)}{=} \LL^p_{i+1}(L(p)) \subseteq \LL^p_{i}(L(p))$ and any pair of edges from $\KK^p_{i+1}(R(p))$, since $\KK^p_{i+1}(R(p)) \stackrel{(*)}{=} \LL^p_{i+1}(R(p)) \subseteq \LL^p_{i}(R(p)) = \LL^p_{i}(L(p))$, so we have that
\begin{description}
	\item[\textbf{($**$)}] $\KK^p_i(L(p)) = \KK^p_{i+1}(L(p)) \stackrel{(*)}{=} \LL^p_{i+1}(L(p))$ and $\KK^p_i(R(p)) = \KK^p_{i+1}(R(p)) \stackrel{(*)}{=} \LL^p_{i+1}(R(p))$
\end{description} 

%\ms{\footnotesize why the equality with $R(p)$?}
%\wn{We have that $\LL^p_i(L(p)) = \LL^p_i(R(p))$, so that's the same argument}.
If the current cover level of $L(p)R(p)$ in the tree data structure is $i$, it means that $\KK^p_i(L(p)) \stackrel{(**)}{=}\LL^p_{i+1}(L(p)) \neq \LL^p_{i+1}(R(p))\stackrel{(**)}{=} \KK^p_i(R(p))$.
%, so $\KK^p_i(L(p)) = \LL^p_{i+1}(L(p))$ and similarly $\KK^p_i(R(p)) = \LL^p_{i+1}(R(p))$.
As $\KK^p_i(L(p))$ and $\KK^p_i(R(p))$ are different subsets of $\LL^p_i(L(p))$, based on the second part of \cref{lem:pomegranade} we conclude that they constitute a~partition of $\LL^p_i(L(p))$, or in other words --- the pair $(L(p), R(p))$ is uniform and calling $\UniformUncover{\cdot, \cdot, i}$ on a subpath of $u \ldots v$ containing $L(p)R(p)$ will correctly resolve $\LL^p_i$ to $\KK^p_i$.
\end{proof}
%\msinline{\footnotesize Low prio for above: use different notation than $\LL'^{p}$, it looks really meh.}

Knowing this, what remains is to call $\textsc{UniformUncover}$ with appropriate arguments. Namely, it suffices to call $\UniformUncover{u, p_{\rm first}, i}$, where $p_{\rm first}$ is the leftmost important vertex, then call $\UniformUncover{p_l, p_r, i}$ for each pair $(p_l, p_r)$ of consecutive important vertices and --- if the procedure was not stopped --- call $\UniformUncover{p_{\rm last}, v, i}$, where $p_{\rm last}$ is the rightmost important vertex; or call $\UniformUncover{u, v, i}$ if there were no important vertices at all.

\subparagraph{Important vertices}

Let $p$ be an important vertex. Consider a moment right after we processed all events that were strongly $i$-reachable from $L(p)$.
%\ms{\footnotesize $i$-reachable vertices (somehow corresponding to the events)?}
Such a~moment can be identified during the $\textsc{UncoverPath}$ call as either: the first moment of processing an event whose projection is $p$, but $\textsc{FindNextEvent}$ returned $R(p)$ as the tree edge; the first moment when we process an event with a~different projection; or the moment $\FindNextEvent{}$ ran out of events to process.
%\ms{\footnotesize also if we receive $\bot$?}\wn{Good point. Fixed.}
 Let us call our graph at that moment as $G''$. We define $\KK^p$ similarly as in the case of skipped through vertices. We claim the following:

\begin{lemma} \label{lem:important-general}
If $\CoverLevel{L(p), R(p)} \neq i$, we do not have to update $\LL^p_i$. Otherwise, if we call $\LocalUncover{L(p), R(p), i}$ on $G''$, then $\LL^p_i$ will get correctly updated.
\end{lemma}
\begin{proof}
	Based on the first part of \Cref{lem:pomegranade} we have that $\KK_{i+1}^p(L(p)) = \LL_{i+1}^p(L(p))$ and $\KK_{i+1}^p(R(p)) = \LL_{i+1}^p(R(p))$.
	
	If $\CoverLevel{L(p), R(p)} \neq i$, then $\CoverLevel{L(p), R(p)} \ge i+1$ and lowering level of $uv$ from $i$ to $i-1$ will not affect $\LL^p_i$, hence we assume that $\CoverLevel{L(p), R(p)} = i$, what means that $\KK_{i+1}^p(L(p)) = \LL_{i+1}^p(L(p)) \neq  \LL_{i+1}^p(R(p)) = \KK_{i+1}^p(R(p))$ and in particular $\LL_{i+1}^p(L(p)) \subsetneq \LL_{i+1}^p(L(p)) \cup \LL_{i+1}^p(R(p)) \subseteq \LL_{i}^p(L(p)) = \LL_{i}^p(R(p))$.
	
%	Let $F_L$ and $F_R$ denote subtrees of $F$ rooted at $p$, where $F_L$ consists of vertices reachable in $F$ from $p$ through edges of $\LL_{i+1}^p(L(p))$\ms{} and $F_R$ consists of vertices reachable in $F$ from $p$ through edges of $\LL_{i}^p(L(p)) \setminus \LL_{i+1}^p(L(p))$ (note that this set of edges is non-empty as it contains $R(p)$).
	Let $F_L$ and $F_R$ denote subtrees of $F$ rooted at $p$, where $F_L$ consists of $p$ and vertices $x$ such that $(p, s^x(p)) \in \LL_{i+1}^p(L(p))$ and $F_R$ consists of $p$ and vertices $y$ such that $(p, s^y(p)) \in \LL_{i}^p(L(p)) \setminus \LL_{i+1}^p(L(p))$ (note that $\LL_{i}^p(L(p)) \setminus \LL_{i+1}^p(L(p))$ is non-empty as it contains $R(p)$).
	We defined $F_L$ and $F_R$ exactly so that the vertices of $V(F_L) \setminus \{p\}$ that are $i$-reachable from $L(p)$ through $p$
%	 \ms{\footnotesize the vertices of $F_L$ that are $i$-reachable from xyz}
	are actually strongly $i$-reachable from $L(p)$ through $p$, but vertices of $V(F_R) \setminus \{p\}$ are not strongly $i$-reachable from $L(p)$ through $p$.
%	 are not strongly $i$-reachable from it\ms{\footnotesize is "it" = $p$?}.
	Thanks to the tiebreaking rule $(\dagger\dagger\dagger)$ we know that at that moment there are no $i$-marked vertices in $V(F_L) \setminus \{p\}$ that are $i$-reachable from $L(p)$ through $p$, hence there are also no level-$i$ edges connecting $V(F_L) \setminus \{p\}$ and $V(F_R) \setminus \{p\}$, which means that there is no level-$i$ edge covering a~pair of edges $(e_1, e_2)$ such that $e_1 \in \LL_{i+1}^p(L(p)) = \KK_{i+1}^p(L(p))$ and $e_2 \notin \LL_{i+1}^p(L(p)) = \KK_{i+1}^p(L(p))$, so $\KK_i^p(L(p)) = \KK_{i+1}^p(L(p)) = \LL_{i+1}^p(L(p))$.
%	There are no edges of level higher than $i$ that connect $V(F_L) \setminus \{p\}$ and $V(F_R) \setminus \{p\}$ as any such edge would cover at level higher than $i$ a pair of edges $(e_1, e_2)$, where $e_1 \in \LL_{i+1}^p(L(p))$ and $e_2 \notin \LL_{i+1}^p(L(p))$, contradicting the definition of $\LL_{i+1}^p(L(p))$.

As $\KK_i^p(L(p)) = \LL_{i+1}^p(L(p)) \subsetneq \LL_{i}^p(L(p))$, thanks to \cref{lem:pomegranade} we infer that $\LL_i^p(L(p)) = \KK_i^p(L(p)) \cup \KK_i^{p}(R(p))$ and since $\KK_i^p(L(p)) \neq \LL_i^p(L(p))$, we have to be in the second case of that lemma, where $\LL_i^p(L(p)) = \KK_i^p(L(p)) \sqcup \KK_i^{p}(R(p))$, which means that $\KK_i^p(R(p)) = \LL_i^p(L(p)) \setminus \KK_i^p(L(p)) = \LL_i^p(L(p)) \setminus \LL_{i+1}^p(L(p))$. Hence, applying $\LocalUncover{L(p), R(p), i}$ will correctly update $\LL_i^p$.
\end{proof}
\paragraph{$\UncoverPath{}$ correctness analysis} 
Let us sum up the $\UncoverPath{}$ procedure. We proceed in two symmetric phases --- first we go from $u$ to $v$ and then we go from $v$ to $u$. In each phase we retrieve non-tree level-$i$ edges that are $i$-reachable from $u \ldots v$. We do so in an order from left to right of their left projections with a carefully chosen tiebreaking rule. For each such edge we first call $\UniformUncover{}$ and $\LocalUncover{}$ with appropriate arguments. Then, if it is legal to push the considered non-tree edge to level $i+1$, we do so. If it is not, we stop this phase. The full pseudocode of this function is presented as \cref{alg:uncover-path} in \Cref{sec:graph-impl}.
\msinline{\footnotesize Do we ever precisely state the tiebreaking rule outside of the pseudocode? I think this section mentions a~few times that there's some rule, but I cannot actually find it (I may be blind, tho)}
\wninline{Yes, it was stated explicitly, visibly exposed and even given the mark of three daggers. But I admit that I've just changed a sentence stating that we will explain the tiebreaking rule "later" when it was already defined earlier :p}

\begin{lemma} \label{lem:uncover-correctness}
	The procedure $\UncoverPath{u, v, i}$ correctly resolves all cover levels affected by lowering the level of $uv$ edge from $i$ to $i-1$.
\end{lemma}
\begin{proof}
%\msinline{\footnotesize FYI: you can use \texttt{\textbackslash{}Cref\{ref1,ref2\}}}
If any of the two phases is completed, then all the cover levels are clearly correctly updated thanks to our previous analysis of \Cref{lem:skip-clean,lem:important-general}, so from that point on, we assume that none of them was completed (actually it cannot be the case that exactly one of them was completed, what will become clear later on).

Let $B_1, B_2, \ldots, B_c$ be the biconnected components of $G_i$ that contain at least one edge of $u \ldots v$, where we assume that the level of $uv$ has already been lowered from $i$ to $i-1$, so that $uv \notin G_i$. The intersection of each of them with $u \ldots v$ is a subpath of $u \ldots v$ and these subpaths form an~edge partitioning of $u \ldots v$, hence they can be naturally linearly ordered from left to right. Without loss of generality assume that this ordering is $B_1, B_2, \ldots, B_c$ (where $B_1$ is the only one that contains~$u$). $B_1, B_2, \ldots, B_c$ cover whole $u \ldots v$, so in particular each pair $(B_j, B_{j+1})$ has a unique common vertex. Before lowering the level of $uv$ from $i$ to $i-1$ all these biconnected components formed one biconnected component $B$ in $G_i$ and splitting it to $B_1, B_2, \ldots, B_c$ is the only change in biconnected components such an~operation causes and all we need to do is to resolve cover levels around common vertices of biconnected components that are consecutive in the order. We also have $\ceil{\frac{n}{2^i}} \ge |V(B)| = |V(B_1) \cup V(B_2) \cup \ldots \cup V(B_c)|$ thanks to the invariant $(\dagger)$.

Components that have at most $\ceil{\frac{n}{2^{i+1}}}$ vertices will be called \emph{small} and other ones will be called \emph{large}. We claim that there is at most one large component among $B_1, \ldots, B_c$. Assume otherwise, that is, there exist $1 \le j_1 < j_2 \le c$ such that $|V(B_{j_1})|, |V(B_{j_2})| > \ceil{\frac{n}{2^{i+1}}}$. Note that $|V(B_{j_1}) \cap V(B_{j_2})| \le 1$ (and it equals~$1$ if and only if $j_1+1=j_2$), hence $|V(B_{j_1}) \cup V(B_{j_2})| \ge |V(B_{j_1})| + |V(B_{j_2})| - 1$. We get that $\ceil{\frac{n}{2^i}} \ge |V(B)| \ge |V(B_{j_1}) \cup V(B_{j_2})| \ge |V(B_{j_1})| + |V(B_{j_2})| - 1 \ge 2(\ceil{\frac{n}{2^{i+1}}} + 1) - 1 > \ceil{\frac{n}{2^i}}$, which is a contradiction. 

All non-tree edges that we will retrieve during the $\UncoverPath{u, v, i}$ call will be contained within one of the $B_1, \ldots, B_c$. If such an edge belongs to a small component, then it is always legal to promote it from level $i$ to level $i+1$. Hence, if all components are small, then none of the two phases will be stopped and we correctly resolve all the cover levels, so we assume that there exists a large component --- let us call it $B_j$. Then, we will not be able to promote all level $i$ edges within it to level $i+1$, hence both phases of $\UncoverPath{u, v, i}$ will be stopped at some point of retrieving level $i$ edges from it.

%\wninline{The paragraph below really calls for better arguments... I will take care of that. In particular, it could use the argument that if $\FindNextEvent$ returns $(e, p, f)$, then $e$ and $f$ will be in the same biconnected component. That property is not true with the current definition that allows vertices on $u \ldots v$ to be strongly reachable from $u \ldots v$ what makes the current version wrong}

We need the auxiliary claim:
\begin{claim} \label{cl:retrieving-order}
	Assume that $e_1$ and $e_2$ are two edges retrieved during first phase of $\UncoverPath{u, v, i}$ such that $e_1 \in E(B_x)$, $e_2 \in E(B_y)$ and $e_1$ was retrieved before $e_2$. Then $x \le y$.
\end{claim}
\begin{proof}
	We know that the projection of $e_1$ onto $u \ldots v$ is contained within $V(B_x) \cap u \ldots v$ and the projection of $e_2$ onto $u \ldots v$ is contained within $V(B_y) \cap u \ldots b$. However, sets $V(B_1) \cap u \ldots v, \ldots, V(B_c) \cap u \ldots v$ are ordered such that if $a + 1 < b$ then all vertices of $V(B_a) \cap u \ldots v$ lie strictly on the left of all vertices of $V(B_b) \cap u \ldots v$ and if $a+1=b$, then it is also true with the exception that the rightmost vertex of $V(B_a) \cap u \ldots v$ is the leftmost vertex of $V(B_b) \cap u \ldots v$. Since we retrieve events in the order of their projections from left to right, if the projection of $e_1$ is on strictly on the left of the projection of $e_2$, then our claim follows easily. It also easily follows if they project to the same vertex that is not the common vertex $V(B_a) \cap V(B_{a+1})$ of two consecutive biconnected components.
	
	The only remaining case to exclude is that $e_1 \in E(B_{a+1}), e_2 \in E(B_a)$ and both $e_1$ and $e_2$ project to the unique common vertex of $B_a$ and $B_{a+1}$ for some $1 \le a \le c$ --- denote that vertex by $p$. We have that $L(p) \in E(B_a)$ and $R(p) \in E(B_{a+1})$.
	
	Assume that the edge $e_1 = gh$ was retrieved as an effect of retrieving $g$ as the $i$-marked vertex that was strongly $i$-reachable from $L(p)$ through $p$. Then, the first two edges of its strong reachability witness were $L(p)$ and $(p, s^g(p))$. Hence, the cover level of that pair is at least $i+1$, so they have to belong to the same biconnected component of $G_i$ after lowering the level of $uv$ from $i$ to $i-1$, or in other words $(p, s^g(p)) \in E(B_a)$ as $L(p) \in E(B_a)$. However, all edges on the cycle induced by $e_1$ in $F$ belong to the same biconnected component of $G_i$ too, and $(p, s^g(p))$ is one of them, so~$e_1 \in E(B_{a+1}) \Rightarrow (p, s^g(p)) \in E(B_{a+1})$ --- a contradiction.
	
	Hence, at the moment of retrieving $e_1$, it was not retrieved as a strongly $i$-reachable edge from $L(p)$ through $p$. According to our tiebreaking rule $(\dagger\dagger\dagger)$ it means that $e_2$ was not retrieved as a~strongly $i$-reachable edge from $L(p)$ through $p$. Moreover, as $e_1$ was not retrieved as a strongly $i$-reachable edge from $L(p)$ through $p$, we know that the $\LocalUncover{L(p), R(p), i}$ was issued right after retrieving $e_1$ at the latest (note that as $L(p)$ and $R(p)$ belong to different biconnected components of $G_i$, $\LocalUncover{L(p), R(p), i}$ must have been issued), which is before retrieving~$e_2$. Thanks to the definition of $\LocalUncover{}$, after the call $\LocalUncover{L(p), R(p), i}$ we have that $\LL_i^p(L(p)) = \LL_{i+1}^p(L(p))$, which means that all $i$-reachable from $L(p)$ through $p$ vertices, except for $p$, are actually strongly $i$-reachable from $L(p)$ through $p$. Hence, the property that there are no $i$-marked and strongly $i$-reachable from $L(p)$ through $p$ vertices, which is ensured by the tiebreaking rule $(\dagger\dagger\dagger)$, implies that there are no more $i$-marked and $i$-reachable from $L(p)$ through $p$ vertices, except for $p$.
	
	However, if $e_2=gh$ and $e_2 \in E(B_a)$, then we have that both $g$ and $h$ project to neither something on the left of $p$ (because otherwise this edge would have been retrieved earlier) nor something on the right of $p$ (because otherwise we would have $R(p) \in E(B_a)$), so they both project to $p$, hence they are $i$-reachable from $L(p)$ through $p$. They are both $i$-marked and least one of them is different than $p$ that, what contradicts the fact that there are no more $i$-marked and $i$-reachable from $L(p)$ through $p$ vertices.
\end{proof}

Moreover, thanks to understanding from the \Cref{cl:retrieving-order}, we note that the following claim holds as well:
\begin{claim}
	Assume that $e \in E(B_x)$ was retrieved during first phase. Then, after processing hypothetical $\UniformUncover{}$ and $\LocalUncover{}$ calls that might have been issued as an effect of that, all cover levels around common vertices of $V(B_1) \cap V(B_2), \ldots, V(B_{x-1}) \cap V(B_x)$ have been correctly resolved, resulting in correctly splitting $B_1, B_2, \ldots, B_x$ from each other.
\end{claim}

%Note that thanks to the order of retrieving events and Lemmas \ref{lem:skip-clean} and \ref{lem:important-general}, when we retrieve an event with a non-tree edge $e_n$ and a tree edge $e_t$, in the moment after issuing $\LocalUncover{}$ and $\UniformUncover{}$ calls corresponding to it, but before checking if $e_n$ can be promoted to level $i+1$, all cover levels around vertices on the left of $e_t$ (which
%%\ms{which}
%includes its left end) were correctly resolved. \wn{I think that is false if we do not change the definition of strong reachability to exclude vertices from $u \ldots v$.} Hence, at the moment of checking if the first retrieved non-tree edge from $B_j$ can be promoted, we have already correctly resolved cover levels around common vertices of
%%\ms{\footnotesize $B_1 \cap B_2$?}
%$B_1 \cap B_2, B_2 \cap B_3, \ldots, B_{j-1} \cap B_j$ resulting in splitting appropriately these components from each other.
As an effect, we conclude that the first phase correctly resolves all cover levels around common vertices of $V(B_1) \cap V(B_2), \ldots, V(B_{j-1}) \cap V(B_j)$ and it will be stopped at some point of processing edges from $B_j$. Analogous claims hold for the second phase, which will correctly resolve cover levels around common vertices of $V(B_c) \cap V(B_{c-1}), \ldots, V(B_{j+1}) \cap V(B_j)$. Hence, both phases together will correctly resolve cover levels around common vertices of $V(B_1) \cap V(B_2), \ldots, V(B_{c-1}) \cap V(B_c)$, which proves the correctness of $\UncoverPath{}$ procedure.

\end{proof}

\subsection{Summary}
Having proven the correctness of $\UncoverPath{}$, in order to prove \cref{lem:first-reduction}, we proceed to analyze the number of calls to the tree data structure and the time complexity overhead we need for them. Per single $\Insert{}$ or $\Delete{}$ call, there is only a constant number of calls to $\Link{}$, $\Cut{}$ and $\Expose{}$. For each edge retrieved by $\textsc{FindNextEvent}$ in both $\Swap{}$ and $\UncoverPath{}$, we call $\UniformUncover{}$, $\LocalUncover{}$, $\FindSize{}$, $\Mark{}$, $\Unmark{}$, $\Cover{}$, $\FindFirstReach{}$ and $\FindStrongReach{}$ a constant number of times. For all but at most two such events for $\Swap{}$ and two for $\UncoverPath{}$, we will call $\PromoteEdge{}$ (which in turn calls $\Cover{}$ once). As the total number of calls of $\PromoteEdge{}$ cannot exceed $m \cdot \lmax$, the total number of calls to the aforementioned functions that can be charged to $\PromoteEdge{}$ calls will be $\Oh(m \cdot \lmax)$ as well. The number of such calls stemming from processing events that cannot be charged to $\PromoteEdge{}$ calls will be at most twice the number of $\UncoverPath{}$ and $\Swap{}$ calls, but $\UncoverPath{}$ is called at most $\lmax$ times per a single $\Delete{}$ update and $\Swap{}$ is called at most once per a single $\Delete{}$ update, hence their number will be $\Oh(m \cdot \lmax)$ as well. Additionally, each $\Swap{}$ calls $\FindSize{}$ and $\Cover{}$ at most $\lmax$ outside of its $\Call{FindReplacement}{}$ subroutine, but the total number of such calls is again bounded by $\Oh(m \cdot \lmax)$. In order to perform the necessary navigation and bookkeeping that comes with it, we need to call $s^{\cdot}(\cdot)$ function and update $N^{\cdot}(\cdot)$ sets. The number of such operations can be bounded as $\Oh(m \cdot \lmax)$ in the same way. As they take $\Oh(\log n)$ time each, the required time overhead is $\Oh(m \cdot \lmax \cdot \log n) = \Oh(m \cdot \log^2 n)$. 
%\wninline{Take into account swapping as well.}
 
% and let $e_1e_2$ be a pair of adjacent edges on the $u \ldots v$ path whose cover level in $G$ is $i$ and let $w$ be the common vertex of them.

%%We will consider a few cases:
%
%%\begin{description}
%%	\item[Case 1:] there exists a different non-tree edge at level $i$ that covers $e_1e_2$
%%	
%%	In this case $e_1e_2$ clearly remains covered at level $i$ in $G'$
%%	
%%	\item[Case 2:] there does not exist a different non-tree edge at level $i$ covering $e_1e_2$ and there does not exist any non-tree edge $qr$ such that $q$ is $i$-reachable from $e_1$ in $G$ (or equivalently --- from $e_2$)
%%	
%%	In this case there is no other tree edge $e_3$ adjacent to $w$ such that its cover level with $e_1$ and $e_2$ is $i$. Hence, the pair $e_1e_2$ is uniform. Moreover, as there is no other edge than $uv$ covering it at level $i$ in $G$, its cover level in $G'$ is $i-1$. Hence, if we apply $\UniformUncover{x, y, i}$ on $G$ for some $x, y$ such that $e_1e_2 \subseteq x \ldots y$, then $L_i^w$ will be updated as it should be in $G'$.
%%	
%%	\item[Case 3:] there does not exist a different non-tree edge at level $i$ covering $e_1e_2$, but there exist some non-tree edges $qr$ such that $q$ is $i$-reachable from $e_1$ in $G$
%%	
%%	
%%\end{description}
%

%\wninline{This pseudocode gets a bit ugly with some too long lines :<}
%\msinline{The actual manuscript will have a~wider text area.}

\section{Tree structure} \label{sec:tree-structure}
In this section we will implement a~data structure for the dynamic tree cover data structure.
The end result shall be \Cref{lem:top-tree-reduction}, which we restate below for convenience:

\TopTreeReductionLemma*

Together with \Cref{lem:neighborhood-ds}, this also yields the following result:

\begin{lemma}
  \label{lem:tree-structure}
  Let $\lmax \in \BigO(\log n)$ and $\hat{T} \coloneqq T(\lmax) \cdot \log \lmax$.
  There exists a~restricted dynamic tree cover level data structure with $\lmax$ levels that processes each operation of the form:
  \begin{itemize}
    \item $\Link{}$, $\Cut{}$ and $\Select{}$ in amortized $\BigO(\log^2 n \cdot \log^2 \log n)$ time;
    \item $\Connected{}$ in worst-case $\BigO(\log n)$ time,
    \item $\Cover{}$, $\UniformUncover{}$, $\LocalUncover{}$, $\CoverLevel{}$, $\MinCoveredPair{}$, $\FindSize{}$, $\Mark{}$, $\FindFirstReach{}$ and $\FindStrongReach{}$ in amortized $\BigO(\log n \cdot \log^2 \log n)$ time.
  \end{itemize}
\end{lemma}

The structure of the proof is as follows.
We first formally introduce the top trees data structure of Alstrup, Holm, de Lichtenberg, and Thorup~\cite{TopTreesOriginal} and present various structural properties of these trees that will be used in our tree data structure (\Cref{ssec:tree-structure-toptrees}).
Then, we present the anatomy of the data structure and design a~framework along which updates and queries will be implemented (\Cref{ssec:tree-structure-overview}).
We follow by showing how to maintain cover level information in our data structure in order to facilitate queries of the form $\CoverLevel{}$, $\MinCoveredPair{}$, $\Cover{}$, $\UniformUncover{}$ and $\LocalUncover{}$ (\Cref{ssec:tree-structure-coverlevels}).
In the following parts, we will implement the queries counting $i$-reachable vertices ($\FindSize{}$, \Cref{ssec:tree-structure-counting}) and finding $i$-marked $i$-reachable vertices ($\FindFirstReach{}$, $\FindStrongReach{}$, \Cref{ssec:tree-structure-findfirst,ssec:implicit-refine}).
We conclude the proof in \Cref{ssec:tree-structure-final}.

\subsection{Top trees}
\label{ssec:tree-structure-toptrees}
A~\emph{top tree} is a~data structure representing and maintaining dynamic information about dynamic trees.\ms{cite more?}
We proceed to introduce this data structure formally; the treatment will be slightly non-standard and it will essentially follow the exposition by Holm and Rotenberg~\cite{DynamicPlanarEmbeddings}.

Let $T$ be a~tree with a~designated set $\bnd T$ of one or two vertices, named \emph{external boundary vertices} of $T$.
If $S$ is a~subtree of $T$, we define the \emph{boundary} of $S$, denoted $\bnd S$, as the set comprising those vertices of $S$ that either are external boundary vertices of $T$, or are incident to some edge outside of $S$.
The boundary of $S$ is always nonempty.
If $|\bnd S| \leq 2$, we define that $S$ is a~\emph{cluster}: a~\emph{point cluster} if $|\bnd S| = 1$, and a~\emph{path cluster} if $|\bnd S| = 2$.
A~\emph{top tree} $\mathcal{T}$ is a~rooted tree expressing how the original tree $T$ is decomposed recursively into single edges of $T$. Each node of $\mathcal{T}$ represents a~cluster of $T$ (without worry about confusion, we will identify nodes of a~top tree with the clusters represented by these nodes).
The root of $\mathcal{T}$ represents the original tree $T$, every leaf node is a~single-edge cluster of $T$, and for a~nonleaf cluster $C$, children clusters $C_1, \ldots, C_k$ form an~edge partitioning of $C$.

For our convenience and in order to simplify various case studies in this section, we will use a~standard trick where we add to each vertex $v$ of $T$ a~dummy edge with one endpoint at $v$ and the other endpoint at a~fresh leaf of $T$.

Given a~cluster $C$, the \emph{cluster path} $\pi(C)$ is the shortest path in $C$ connecting all boundary vertices of $C$; so for path clusters $\pi(C)$ contains at least one edge, and for point clusters $\pi(C)$ contains just one vertex.
For a~path cluster $C$ and $v \in \bnd C$, let also $\firstedge_{C,v}$ be the unique edge on $\pi(C)$ incident to $v$. We assume that $\firstedge_{C,v}$ is oriented away from $v$.

In this work, we will require that path clusters are \emph{slim}: if $\bnd S = \{v, w\}$, then both $v$ and $w$ are leaves of $S$.
This creates some issues when the underlying tree $T$ has two external boundary vertices; in this case the entire tree $T$ cannot be described by a~slim path clusters.
We work around this issue by proclaiming that in the case of $\bnd T = \{p, q\}$, the root of the top tree of $T$ consists of \emph{three} clusters: the slim path cluster $R_{pq}$ with boundary $\{p, q\}$, and two point clusters $R_p, R_q$ with boundaries $\{p\}$ and $\{q\}$, respectively.

The structure of $\mathcal{T}$ can be altered through the updates of the form $\Link{v, w}$ (add an~edge $(v, w)$), $\Cut{v, w}$ (remove the edge $(v, w)$), and $\Expose{v}$ or $\Expose{v, w}$ (set $\bnd T$ to $\{v\}$ or $\{v, w\}$, respectively).
Each update is applied to the forest of top trees representing $T$ through a~sequence of local updates: $\Create{}$ (create a~leaf top tree node), $\Destroy{}$ (destroy the leaf top tree node), $\Merge{}$ (create a~root nonleaf top tree node representing a~cluster of $T$ by merging given root top tree nodes into a~single cluster), and $\Split{}$ (delete a~root nonleaf top tree node, replacing it with its children).
It turns out that $\mathcal{T}$ can be efficiently maintained in this regime:

\begin{theorem}[\cite{DynamicPlanarEmbeddings}]
    \label{thm:top-trees}
    We can represent a~fully dynamic $n$-vertex forest under $\Link{}$, $\Cut{}$ and $\Expose{}$ using a~forest of top trees of height $\BigO(\log n)$.
    Each update is applied to the forest in worst-case time $\BigO(\log n)$ and is carried through a~sequence of $\BigO(1)$ $\Create{}$ and $\Destroy{}$ modifications and $\BigO(\log n)$ $\Merge{}$ and $\Split{}$ modifications.
    Moreover, the required sequence of modifications can be determined in worst-case time $\BigO(\log n)$.
\end{theorem}

\paragraph{Additional information in top trees}
The top tree data structure can be amended to maintain additional information for each node of the top tree.
First, for every cluster $C$ we can store \emph{summary} information about $C$ (some aggregate information about $C$ or $\pi(C)$, e.g., the size of the subtree, the maximum value stored at a~vertex of the cluster path, etc.).
We usually require that such summary about $C$ should be effectively computable during $\Merge{}$ from the summaries stored in the children of $C$ in the top tree.
Similarly, along $C$ we can keep an~information on a~\emph{lazy update}: an~update that has been applied to the cluster $C$ (e.g., recoloring of all edges in the subtree, an~increase in the value stored at each vertex of the cluster path, etc.), but not yet propagated to the descendants of $C$ in the top tree. Lazy updates are propagated from a~node of the top tree to its children on each $\Split{}$; this requires that the lazy updates should be composable (i.e., any sequence of lazy updates can be squashed to a~single lazy update) and that summaries of clusters can be maintained efficiently under lazy updates.

\paragraph{Structural properties of top trees}
We now list some properties of structural decompositions provided by top trees that we will implicitly use later on.

%\msinline{make better pictures eventually}
%
%\begin{figure}[h]
%  \centering
%  \begin{subfigure}[b]{0.28\textwidth}
%    \centering
%    \includegraphics[scale=0.2]{images/toptree-point-to-pointpoint.pdf}
%    \caption{}
%    \label{sfig:case-point-to-pointpoint}
%  \end{subfigure}%
%  \begin{subfigure}[b]{0.28\textwidth}
%    \centering
%    \includegraphics[scale=0.2]{images/toptree-point-to-pathpoint.pdf}
%    \caption{}
%    \label{sfig:case-point-to-pathpoint}
%    \label{sfig:case-point-to-pointpath}
%  \end{subfigure}
%  \begin{subfigure}[b]{0.38\textwidth}
%    \centering
%    \includegraphics[scale=0.15]{images/v2-toptree-path-3way.pdf}
%    \caption{}
%    \label{sfig:case-path}
%  \end{subfigure}
%  \caption{All possible ways in which a~cluster $C$ can split into child clusters.}
%  \label{fig:toptree-cases}
%\end{figure}

\begin{figure}[h]
  \centering
  \begin{tikzpicture}[x=0.5cm,y=0.5cm,scale=0.65]
	\begin{scope}[
		every path/.style={
			%thick,
		},
		every node/.style={
			font=\tiny,
			text=black,
			inner sep=1pt,
		},
		every label/.style={
			label distance=2mm,
		},
		% Define a custom style for the vertex sets
		vertex set/.style={
			dashed,
		},
		% Define a custom style for the vertices
		vertex/.style={
			draw,
			circle,
			fill=white,
			minimum size=2mm,
			inner sep=0pt,
			outer sep=0pt,
		},
		boundary vertex/.style={vertex,fill=black},
		% Define custom styles for the edges
		edge/.style={blue,thick},
		undirected edge/.style={edge},
		directed edge/.style={edge,->,>=stealth'},
		]
		
		% homogeneous path cluster (1) now (3)
		\begin{scope}[shift={(30,0)}]%[shift={(0,0)}]
			\path[use as bounding box] (-7,4) rectangle (7,-2);
			\node at (-6,3) {\tiny \CasePath};
			
			\node[boundary vertex,label={above:\tiny $v$}] (a) at (-6,0) {};
			\node[vertex,label={above:\tiny $x$}] (c) at (0,0) {};
			\node[boundary vertex,label={above:\tiny $w$}] (b) at (6,0) {};
			
			\draw (a.north) to[bend left=35] (c.north) to[bend left=35] (b.north);
			\draw (b.south) to[bend left=35] (c.south) to[bend left=35] (a.south);
			\node at (-3,0) {\tiny $A$};
			\node at (3,0) {\tiny $B$};
			\draw
			(c.south)
			.. controls ($(c)+(-3,-4)$) and ($(c)+(3,-4)$)
			.. (c.south);
            \node at (0,-2) {\tiny $P$};

			\draw[color=blue,thick] (a.north) to[bend left=55] (b.north);
			\draw[color=blue,thick] (b.south) to[bend left=65] (a.south);
			\node[blue] at (3.5,3) {\tiny $C$};
		\end{scope}
		
		% heterogeneous point cluster (3) now (2)
		\begin{scope}[shift={(15,3)}]
			\path[use as bounding box] (-4,2) rectangle (5,-12);
			%% \node[vertex] (pb1) at (-4,-13) {};
			%% \node[vertex] (pb2) at (4,-13) {};
			
			\node at (-3,0) {\tiny \CasePointToPathPoint};
			
			\draw
			(0,0) node[boundary vertex,label={right:\tiny $v$}] (a) {}
			to[bend left=35]
			(0,-6) node[vertex,label={left:\tiny $x$}] (c) {}
			to[bend left=35]
			(a);
			
			\draw
			(c)
			.. controls ($(c)+(-4,-6)$) and ($(c)+(4,-6)$)
			.. (c);
			
			\node at (0,-3) {\tiny $A$};
			\node at (0,-9) {\tiny $B$};

			\draw[blue]
			(a)
			.. controls ($(c)+(-10,-9)$) and ($(c)+(10,-9)$)
			.. (a);

			\node[blue] at (2.7,-10.7) {\tiny $C$};
		\end{scope}
		
		% homogeneous point cluster (4) now (1)
		\begin{scope}[shift={(0,3)}]
			\path[use as bounding box] (-7,2) rectangle (4,-5);
			%% \node[vertex] (pa) at (-8,-4) {};
			%% \node[vertex] (pb) at ( 8,-4) {};
			%% \node[vertex] (pc) at (0,-8) {};
			
			\node at (-6,0) {\tiny \CasePointToPointPoint};
			
			\draw
			(0,0) node[boundary vertex,label={above:\tiny $v$}] (c) {}
			.. controls ($(c)+(-7,-4)$) and ($(c)+(0,-8)$)
			.. (c);
			\draw
			(c) {}
			.. controls ($(c)+(0,-8)$) and ($(c)+(7,-4)$)
			.. (c);
			
			\node at (-2,-3) {\tiny $A$};
			\node at (2,-3) {\tiny $B$};

			\draw[blue]
			(c) {}
			.. controls ($(c)+(-15,-8)$) and ($(c)+(15,-8)$)
			.. (c);
			\node[blue] at (4.6,-5.7) {\tiny $C$};

		\end{scope}
				
	\end{scope}
\end{tikzpicture}
  \caption{All possible ways in which a~cluster $C$ can split into child clusters.}
  \label{fig:toptree-cases}
\end{figure}
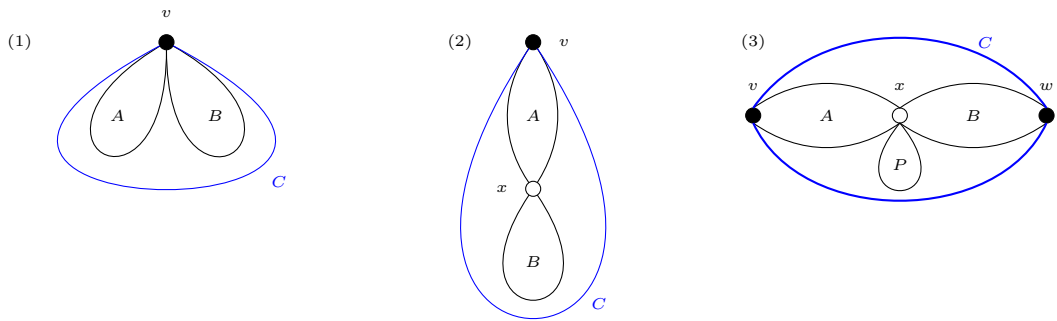

\subparagraph{Children of nonleaf clusters}
We define that a~nonleaf cluster $C$ can split into children clusters in three possible ways, listed below:
\begin{itemize}
    \item If $C$ is a~point cluster with $\bnd C = \{v\}$, then $C$ has precisely two children $A$, $B$.
    It can happen that $\bnd A = \bnd B = \{v\}$ (case \CasePointToPointPoint in \Cref{fig:toptree-cases}).
    Naturally, this case can only hold if $v$ is a nonleaf vertex of $C$.

    \item It can also happen that $\bnd C = \{v\}$ and precisely one of $A, B$ is a~(slim) path cluster (without loss of generality $A$); see case \CasePointToPathPoint.
    Then it must be the case that $\bnd A = \{v, x\}$ and $\bnd B = \{x\}$ for some vertex $x \in V(T)$.
    Since $A$ is a~slim path cluster, this case can only hold if $v$ is a leaf of $C$.
    It also follows that $v$ and $x$ are leaves of $A$, so all edges incident to $x$ in the underlying tree $T$ --- except from $\firstedge_{A,x}$ --- belong to $B$.

    \item Finally, if $C$ is a~(slim) path cluster with $\bnd C = \{v, w\}$, then $C$ has three children: two path clusters $A, B$ and a~point cluster $P$, where $\bnd A = \{v, x\}$, $\bnd B = \{x, w\}$ and $\bnd P = \{x\}$; see case \CasePath.
    Here, both $v$ and $w$ are leaves of $C$; both $v$ and $x$ are leaves of $A$; and both $x$ and $w$ are leaves of $B$.
    Hence all edges incident to $x$ in $T$ --- except from $\firstedge_{A,x}$ and $\firstedge_{B,x}$ --- belong to $P$.
\end{itemize}

In cases \CasePointToPathPoint and \CasePath, we call $x$ the \emph{midpoint} of $C$.

We will often implicitly rely on the following observation about point clusters $C$ (with $\bnd C = \{v\}$) splitting into two point clusters according to case \CasePointToPointPoint: suppose $e_1 = \vec{vv_1}, \ldots, e_k = \vec{vv_k}$ are the edges of $C$ incident to $v$ ($k \geq 2$).
Then, for every $i \in \{1, \ldots, k\}$, the subtree $T_i \coloneqq T[e_i] \cup \{e_i\}$ --- i.e., the subtree of $T$ containing the edge $e_i$ and all vertices closer to $v_i$ than $v$ --- is a~point cluster of $\Tc$ that splits into child clusters following case \CasePointToPathPoint.
Moreover, $T_i$ is a~strict descendant of $C$ in $\Tc$, and moreover, every cluster $C' \neq T_i$ on the vertical path between $C$ and $T_i$ in $\Tc$ is actually a~point cluster with $\bnd C' = \{v\}$ splitting into child clusters according to case \CasePointToPointPoint.

%\wninline{I'd explicitly remark here that every point cluster with multiple edges $e_1, \ldots, e_k$ adjacent to its boundary vertex has all clusters $T[e_1], \ldots, T[e_k]$ as its descendants. It's an important intuition.}
%\msinline{ok, added the paragraph above}

\subparagraph{Internal vertices of clusters}
Let $x \in V(T)$. We say that $x$ is an~internal vertex in a~cluster $C \in \mathcal{T}$ if $x \in C \setminus \bnd C$.
If $x \in \bnd T$, then $x$ is not internal to any cluster of $\Tc$.
Otherwise, the set of clusters of $\mathcal{T}$ containing $x$ as an~internal vertex forms a~vertical path in $\mathcal{T}$ containing the root cluster; in particular, there is the unique deepest cluster $I^\Tc_x$ --- call it the \emph{enclosing cluster} for $x$ with respect to $\Tc$ --- that contains $x$ internally.
We will drop $\Tc$ from the superscript if $\Tc$ is clear from context.
For $x \in \bnd T$, we make a~convention that $I^\Tc_x$ is the root of $\mathcal{T}$.
We can assume that the top tree data structure holds a~mapping from every $x$ to $I^\Tc_x$.
Note that for $x \notin \bnd T$, the enclosing cluster $I^\Tc_x$ splits into two or three clusters, each containing $x$ as a~boundary vertex, according to one of the cases \CasePointToPathPoint or \CasePath above, and in both cases, $x$ is the midpoint of $I^\Tc_x$.

\subparagraph{Cluster and interface edges}
Let $x \in V(T)$. We define that an~oriented edge $e = \vec{xy}$ is a~cluster edge with respect to $\Tc$ if $e$ belongs to some point cluster $C \in \Tc$ with $\bnd C = \{x\}$; otherwise, $e$ is an~interface edge.
Equivalently, $e$ is a~cluster edge with respect to $\Tc$ if and only if $\Tc$ contains a~point cluster $C = T[e]$.

For $x \in T$, let $\interface^\Tc_x$ denote the set of (oriented) interface edges whose tail is $x$, and $\cluster^\Tc_x$ be the set of oriented cluster edges with tail $x$.
The following characterization of $\interface^\Tc_x$ is immediate.
\begin{lemma}
  \label{lem:interface-firstedge}
  For $x \in T$, the set $\interface^\Tc_x$ is equal to:
  \begin{itemize}
    \item $\emptyset$ if $\bnd T = \{x\}$;
    \item $\{\firstedge_{T,x}\}$ if $\bnd T = \{x, x'\}$;
    \item $\{\firstedge_{A,x}\}$ if $x \notin \bnd T$, and $I^\Tc_x$ splits into clusters $A$, $B$ according to case \CasePointToPathPoint in \Cref{fig:toptree-cases};
    \item $\{\firstedge_{A,x}, \firstedge_{B,x}\}$ if $x \notin \bnd T$, and $I^\Tc_x$ splits into $A$, $B$, $P$ according to case \CasePath.
  \end{itemize}
\end{lemma}
Again, we will write $\interface_x$ and $\cluster_x$ if $\Tc$ is known from context.

\subsection{Overview of the cover level data structure}
\label{ssec:tree-structure-overview}
We now overview how top trees (\Cref{thm:top-trees}) and neighborhood data structures (\Cref{lem:neighborhood-ds}) cooperate in the implementation of the cover level data structure.

Recall that the cover level data structure holds a~dynamic tree $T$.
For every node $v$ of $T$, we maintain an~instance $N_v$ of neighborhood data structure with all supported extensions, with $\lmax$ levels and the ground set equal to the set of edges  incident to $v$.
The instance $N_v$ naturally represents the sequence of equivalence relations $\mathcal{L}^v_0, \mathcal{L}^v_1, \ldots, \mathcal{L}^v_{\lmax}$.
%\msinline{I think I want to actually do the biasing extension later, but leave selection here}

Naturally, $T$ is represented by a~top tree data structure $\mathcal{T}$.
We preserve the following invariants for $N_v$:
\begin{itemize}
    \item \textbf{Selected items.} In $N_v$, the selected items are precisely the interface edges with tail $v$ in the top tree if $|\interface_v| = 2$. Otherwise, no items are selected.
%    This is allowed by the selection extension of $N_v$ since $|\interface_v| \leq 2$ for all vertices $v$ of the tree.

    \item \textbf{Weights.} In $N_v$, the weight of an~item $\vec{vw}$, denoted $\Weight_{\Tc}(\vec{vw})$, depends on the type of the edge with respect to $\Tc$, and we define it as follows:
  
    \[
        \Weight_{\Tc}(\vec{vw}) = \begin{cases}
            |E(T[\vec{vw}])| & \text{if $\vec{vw}$ is a~cluster edge}, \\
            1 + \sum_{\vec{vx} \in \cluster_v} \Weight_{\Tc}(\vec{vx}) & \text{if $\vec{vw}$ is an~interface edge}.
        \end{cases}
    \]

    For convenience, we denote by $\Weight_{\Tc}(v) \coloneqq \sum_{\vec{vx} \in \delta_{\mathrm{out}}(v)}$ the total weight of all items in $N_v$.
    Since $|\interface_v| \leq 2$, we have $\Weight_{\mathcal{T}}(v) \leq 2 + 3 \cdot \sum_{\vec{vx} \in \cluster_v} |E(T[\vec{vw}])|$; in particular, $\Weight_{\Tc}(v) \in \BigO(n)$.
    By the same token, whenever $\vec{vw} \in \interface_v$, we have $\Weight_{\Tc}(\vec{vw}) > \frac13 \Weight_{\Tc}(v)$.
\end{itemize}

The sets of interface edges and the weights only change under the updates of the tree changing the structure of $\mathcal{T}$ --- $\Insert{}$, $\Cut{}$ and $\Expose{}$.
The types and weights of edges can be maintained by the top tree data structure within the same time complexity guarantees as in Theorem~\ref{thm:top-trees}.
In fact, by \Cref{lem:interface-firstedge} it is enough to recompute the type of the edges incident to a~vertex $v$ only when an~enclosing cluster for $v$ is created during \textsc{Merge} or \textsc{Create}.
Hence every \textsc{Merge} and every \textsc{Create} changes the type and the subtree size for only a~constant number of edges.
We refer to the pseudocode in \Cref{ssec:tree-impl-bookkeeping} for details.

%\msinline{weight $\to$ cost or sth}
%Given the weights of the items in the neighborhood data structure, we can define the \emph{costs} of operations on these structures.
%Namely, in $N_v$:
%\begin{itemize}
%  \item queries $\Insert{}$, $\Delete{}$, $\Unzip{}$, $\Select{}$ and $\SetWeight{}$ have weight $\log n$;
%  \item queries $\Zip{x, y, \cdot}$ and $\Level{x, y, \cdot}$ have weight $1 + \log \frac{\Weight_\Tc(v)}{\Weight_\Tc(x)} + \log \frac{\Weight_\Tc(v)}{\Weight_\Tc(y)}$;
%  \item queries $\UpdateCounters{x, \cdot}$, $\SumCounters{x}$, $\UpdateMarks{x, \cdot}$ and $\OrMarks{x}$ have weight $1 + \log \frac{\Weight_\Tc(v)}{\Weight_\Tc(x)}$;
%  \item query $\FindMarked{x, \cdot}$ has weight $1 + \log \frac{\Weight_\Tc(v)}{\Weight_\Tc(x)}$ if it returns $\bot$, and $1 + \log \frac{\Weight_\Tc(v)}{\Weight_\Tc(x)} + \log \frac{\Weight_\Tc(v)}{\Weight_\Tc(y)}$ if it returns a~vertex $y$.
%  \item queries $\LongZip{}$, $\LongUnzip{}$ and $\SelectedLevel{}$ have weight $1$.
%\end{itemize}
%By \Cref{lem:neighborhood-ds}, a~query of weight $W$ can be processed by $N_v$ in time $\BigOTilde(W)$, where $\BigOTilde$ hides factors of the form $\log \lmax$, $T(\lmax)$ and $B(\lmax)$.
%
%\wninline{Just checking ... it's intentional that there is no mention about Select here, right?}
%\msinline{It's not! Thanks for catching that.}
%
%\msinline{\footnotesize For the future: make sure $\Select{}$ can indeed have weight $\log n$, i.e., make sure $\Select{}$ can be performed by the nbd ds in time $\BigOTilde(\log n)$}

\paragraph{Cluster costs}
Let $\Tc$, $\Uc$ be two top trees defined over the same underlying tree $T$.
Whenever $C$ is a~cluster of $\Uc$, we define the \emph{cost} of $C$ with respect to $\Tc$, denoted $\ncost_\Tc(C)$, as follows:
\begin{itemize}
  \item $\ncost_\Tc(C) = 1$ if $C$ is a~leaf of $\Tc$, or $C$ splits according to case \CasePointToPointPoint in \Cref{fig:toptree-cases};
  \item $\ncost_\Tc(C) = 1 + \log \frac{\Weight_\Tc(v)}{\Weight_\Tc(\firstedge_{A,v})} + \log \frac{\Weight_\Tc(x)}{\Weight_\Tc(\firstedge_{A,x})} $ if $C$ splits according to case \CasePointToPathPoint;
  \item $\ncost_\Tc(C) = 1 + \log \frac{\Weight_\Tc(x)}{\Weight_\Tc(\firstedge_{A,x})} + \log \frac{\Weight_\Tc(x)}{\Weight_\Tc(\firstedge_{B,x})} $ if $C$ splits according to case \CasePath.
\end{itemize}

The intuition behind this cost function is as follows: suppose that the weights of items in the neighborhood data structures are defined with respect to $\Tc$, but the tree cover level data structure holds the top tree $\Uc$. (Usually, we will have $\Tc = \Uc$, but, as we will see later, this may temporarily not be the case while the updates or queries are underway.)
Then, the cost of a~cluster $C$ will be proportional to the total normalized cost of operations in neighborhood data structures performed when: $C$ is split or merged; a~lazy update is propagated from $C$ to its children; or suitable information about cluster $C$ is computed from the corresponding information stored in its children.

Our data structure will crucially rely on the fact that when $\Tc = \Uc$, the total cost of clusters along a~vertical path of a~top tree is logarithmic in the size of the tree:

\begin{lemma}[Vertical Path Telescoping Lemma]
    \label{lem:top-tree-telescope}
    Let $\mathcal{C}$ be a~family of clusters of $\Tc$, all lying on a~single vertical path of $\Tc$.
    Then $\sum_{C \in \mathcal{C}} \ncost_\Tc(C) \in \BigO(\log n)$.
\end{lemma}
\begin{proof}
    Let $C_1, C_2, \ldots, C_k$ be a~leaf-to-root path in $\Tc$ encompassing $\mathcal{C}$; assume $C_1$ is a~leaf of $\Tc$ and $C_k$ is the root of $\Tc$.
    We will prove by induction that, for each $h \in \{1, \ldots, k\}$, we have
    \begin{equation}
        \label{eq:ncost-potential}
        \sum_{i=1}^h \ncost_\Tc(C_i) \leq 6h + \phi(C_h),
    \end{equation}
    where $\phi(C)$ is defined for a~cluster $C$ as follows:
    \[
        \phi(C) = \begin{cases}
            0 & \text{if }C\text{ is a~leaf of }\Tc, \\
            \log \Weight_\Tc(v) & \text{if }C\text{ is a~point cluster and }\bnd C = \{v\}, \\
            \log |E(C)| & \text{if }C\text{ is a~path cluster.}
        \end{cases}
    \]
    Since $k \in \BigO(\log n)$ and $\Weight_\Tc(x) \in \BigO(n)$ for all $x \in V(T)$, the lemma will follow from \eqref{eq:ncost-potential}.
    
    The claim holds for $h = 1$ as $\ncost_\Tc(C_1) = 1$, hence assume $h \geq 2$.
    Then $C_h$ is a non-leaf cluster and $C_{h-1}$ is a~child of $C_h$.
    We consider cases, depending on how $C_h$ splits into children in $\Tc$.
    Note that in each case, it is enough to prove that $\ncost_\Tc(C_h) \leq 6 + \phi(C_h) - \phi(C_{h-1})$.
    \begin{itemize}
        \item If $C_h$ is a~point cluster splitting according to case \CasePointToPointPoint in \Cref{fig:toptree-cases}, then $C_{h-1}$ is a~point cluster with the same boundary as $C_h$, and so $\phi(C_h) = \phi(C_{h-1})$.
        Then \eqref{eq:ncost-potential} follows from $\ncost_\Tc(C_h) = 1$.

        \item Suppose $C_h$ splits into $A$ and $B$ according to case \CasePointToPathPoint and $\bnd C_h = \{v\}$, $\bnd A = \{v, x\}$, $\bnd B = \{x\}$.
        Since $A$ is a~slim path cluster, we have $C = T[\firstedge_{A,v}]$, $\interface_x = \{\firstedge_{A,x}\}$ (by \Cref{lem:interface-firstedge}) and $B = \bigcup_{\vec{xy} \in \cluster_x} T[\vec{xy}]$.

        Then $\phi(C_h) = \log \Weight_\Tc(v)$ and $\ncost_\Tc(C_h) = 1 + \alpha_v + \alpha_x$, where $\alpha_v = \log \frac{\Weight_\Tc(v)}{\Weight_\Tc(\firstedge_{A,v})}$ and  $\alpha_x = \log \frac{\Weight_\Tc(x)}{\Weight_\Tc(\firstedge_{A,x})}$. % \geq \log \Weight_\Tc(\firstedge_{A,v}) = \log |E(T[\firstedge_{A,v}])| = \log |E(C_h)|$.
        Moreover, it follows from $\interface_x = \{\firstedge_{A,x}\}$ that $\Weight_\Tc(\firstedge_{A,x}) > \frac13 \Weight_\Tc(x)$, hence $\alpha_x < \log 3$.
        Additionally, we find that $\Weight_\Tc(\firstedge_{A,x}) = 1 + \sum_{\vec{xy} \in \cluster_x} \Weight_\Tc(\vec{xy}) = 1 + |E(B)|$.

        \smallskip

        If $C_{h-1} = A$, then $\phi(C_{h-1}) = \log |E(A)| \leq \log |E(C_h)| = |E(T[\firstedge_{A,v}])| = \Weight_\Tc(\firstedge_{A,v})$.
        Thus $\alpha_v \leq \phi(C_h) - \phi(C_{h-1})$, from which it follows that $\ncost_\Tc(C_h) \leq 1 + \log 3 + \phi(C_h) - \phi(C_{h-1})$.

        \smallskip

        If $C_{h-1} = B$, then $\phi(C_{h-1}) = \log \Weight_\Tc(x) < \log(3 \Weight_\Tc(\firstedge_{A,x})) = \log 3 + \log(1 + |E(B)|) \leq \log 3 + \log |E(C_h)| = \log 3 + \Weight_\Tc(\firstedge_{A,v})$.
        So $\alpha_v \leq \log 3 + \phi(C_h) - \phi(C_{h-1})$ and therefore $\ncost_\Tc(C_h) \leq 1 + 2 \log 3 + \phi(C_h) - \phi(C_{h-1})$.
        
        \item Finally, suppose $C_h$ is a~(slim) path cluster splitting into (slim) path clusters $A, B$ and a~point cluster $P$ according to case \CasePath.
        Let $\bnd A = \{v, x\}$, $\bnd B = \{x, w\}$, $\bnd P = \{x\}$, so that $\interface_x = \{\firstedge_{A,x}, \firstedge_{B,x}\}$ by \Cref{lem:interface-firstedge}.
        
        Then $\phi(C_h) = \log |E(C_h)|$ and $\ncost_\Tc(C_h) = 1 + \alpha_x + \beta_x$, where $\alpha_x = \log\frac{\Weight_\Tc(x)}{\Weight_\Tc(\firstedge_{A,x})}$ and $\beta_x = \log\frac{\Weight_\Tc(x)}{\Weight_\Tc(\firstedge_{B,x})}$.
        We have $\Weight_\Tc(\firstedge_{A,x}) > \frac13 \Weight_\Tc(x)$ and so $\alpha_x < \log 3$; analogously $\beta_x < \log 3$.
        Hence $\ncost_\Tc(C_h) < 1 + 2 \log 3$.
        Additionally, $\Weight_\Tc(\firstedge_{A,x}) = \Weight_\Tc(\firstedge_{B,x}) = 1 + \sum_{\vec{xy} \in \cluster_x} \Weight_\Tc(\vec{xy}) = 1 + |E(P)|$.
        
        \smallskip
        
        If $C_{h-1} = A$, then $\phi(C_{h-1}) = \log |E(A)| \leq \log |E(C_h)| = \phi(C_h)$ and thus $\ncost_\Tc(C_h) < 1 + 2 \log 3 + \phi(C_h) - \phi(C_{h-1})$.
        The same argument applies when $C_{h-1} = B$.
        
        \smallskip
        
        If $C_{h-1} = P$, then $\phi(C_{h-1}) = \log \Weight_\Tc(x) < \log(3 \Weight_\Tc(\firstedge_{A,x}) = \log 3 + \log(1 + |E(P)|) < \log 3 + \log |E(C_h)| = \log 3 + \phi(C_h)$.
        Hence $\ncost_\Tc(C_h) < 1 + 3 \log 3 + \phi(C_h) - \phi(C_{h-1})$.
    \end{itemize}
    We have $\ncost_\Tc(C_h) \leq 6 + \phi(C_h) - \phi(C_{h-1})$, so \eqref{eq:ncost-potential} holds by induction.
\end{proof}

\paragraph{Lifetime of a~query}
We now sketch how we implement operations --- updates and queries --- in the cover level data structure.
We do not yet specify what information is stored in the nodes of top tree or how this information is maintained under the updates; we delay this exposition to the following sections.

First, consider \emph{structural} updates altering the structure of the underlying data structure --- $\Link{}$, $\Cut{}$ and $\Expose{}$.
Each of these is simply implemented by forwarding the corresponding call to the top tree $\mathcal{T}$.
On every node split, any lazy updates stored in the node is forwarded to the children of this node %\wn{``of this node'' rather than ``of a top tree''?}\ms{fixed};
and similarly, when a~new node is created by a~merge, the information stored in the node is computed based on data stored in the children of the node.
As mentioned, we also recompute the types and weights of edges of $T$, remembering to update this information in the neighborhood data structures.
Each structural update will split and merge $\BigO(\log n)$ nodes of the top tree and recompute the types and weights of $\BigO(\log n)$ edges of $T$.
Since each of these steps take $\BigOTilde(\log n)$ time each, we will find that a~structural update can be performed in time $\BigOTilde(\log^2 n)$.

Now consider all the remaining operations.
A~typical framework for processing such operations is as follows: suppose an~operation pertains to a~vertex $u$ of $T$ or a~path $v \ldots w$ in $T$.
We begin by exposing the vertex $u$ (resp., the path $v \ldots w$), causing the set $\bnd T$ of external boundary vertices of $T$ to become equal to $\{u\}$ (resp., $\{v, w\}$).
Then the operation can be performed essentially \emph{for free}: a~query can be resolved by examining the information stored in the root clusters of the top tree, and an~update is performed by modifying the lazy update state in these clusters.

Sadly, this scheme will not be efficient enough for our purposes: our tree cover level data structure will be ultimately consumed by \Cref{lem:first-reduction}, which in turn provides a~biconnectivity data structure that processes any sequence of $m$ updates via $\BigO(m)$ calls to $\Link{}$, $\Cut{}$, and $\Expose{}$, but as many as $\BigO(m \log n)$ calls to the remaining operations of our tree data structure.
Thus, the implementation of these remaining operations should not invoke $\Expose{}$: as mentioned before, $\Expose{}$ is a~time-expensive structural update that may require as much as $\BigOTilde(\log^2 n)$ time to complete.
%\wninline{I'd explain why this is too slow (it may seem that we are shooting exactly for $\log^2 n$, so that seems fine at the first sight)}
%\msinline{fixed}
This inefficiency essentially comes from the fact that the types and weights of edges of $T$ need to be recomputed and propagated to the neighborhood data structures.
Fortunately, we can work around this problem: we only \emph{transiently expose} the vertex $u$ (resp., the path $v \ldots w$) \emph{without} replacing the weights of edges or the sets of selected edges in the neighborhood data structures.
This produces a~transient version $\Uc$ of $\Tc$, where the types and weights of the edges (and thus also the costs of the clusters) are defined in the neighborhood data structures with respect to the original snapshot $\Tc$ of the top tree.
In particular, the costs of operations on a~cluster $C$ of $\Uc$ in these structures will be proportional to the cost $C$ with respect to $\Tc$ (i.e., $\ncost_\Tc(C)$) instead of $\Uc$.
Then, after a~requested operation on the cover level data structure is performed, $\Uc$ is reverted back to $\Tc$; we call this process \emph{transient unexpose}.

The following lemma formalizes this description and argues that the transient expose can be carried out so that further operations on the transiently exposed top tree $\Uc$ remain efficient:

\begin{lemma}[Transient Expose Lemma]
    \label{lem:transient-expose-lemma}
    A~top trees data structure of \Cref{thm:top-trees} can be extended by implementing $\TransientExpose{u}$ (resp., $\TransientExpose{v, w}$), temporarily altering the top tree $\Tc$ to $\Uc$ so that the set of external boundary vertices of $T$ in $\Uc$ becomes $\{u\}$ (resp., $\{v, w\}$), but the types and weights of edges in $\Uc$ are inherited from $\Tc$. 
    Moreover:
    \begin{enumerate}
        \item \label{item:tel-expose-is-quick} The update is carried through a~sequence of $\BigO(\log n)$ $\Merge{}$ and $\Split{}$ modifications, which can be identified in total time $\BigO(\log n)$.
        \item \label{item:tel-original-prefix-cost} If $\mathcal{C}_\Tc \coloneqq \Tc \setminus \Uc$ is the set of clusters split during the transient expose, then $\sum_{C \in \mathcal{C}_\Tc} \ncost_\Tc(C) \in \BigO(\log n)$.
        \item \label{item:tel-new-prefix-cost} If $\mathcal{C}_\Uc \coloneqq \Uc \setminus \Tc$ is the set of clusters merged during the transient expose, then $\sum_{C \in \mathcal{C}_\Uc} \ncost_\Tc(C) \in \BigO(\log n)$.
        \item \label{item:tel-new-vertical-cost} $\Uc$ still satisfies the Vertical Path Telescoping Lemma, i.e., if $\mathcal{C}$ is a~family of clusters on a~single vertical path of $\Uc$, then $\sum_{C \in \mathcal{C}} \ncost_\Tc(C) \in \BigO(\log n)$.
    \end{enumerate}
\end{lemma}

We will call $\mathcal{C}_\Tc$ the set of \emph{hidden} clusters, and $\mathcal{C}_\Uc$ the set of \emph{transient} clusters.

Then an~operation in the tree cover level data structure will usually run according to the following scheme:
\begin{enumerate}
    \item Identify the sequence of $\BigO(\log n)$ splits (clusters $\mathcal{C}_\Tc$) and merges (clusters $\mathcal{C}_\Uc$) in the top tree $\Tc$ that yields a~transiently exposed top tree $\Uc$. This sequence can be identified in time $\BigO(\log n)$ by \Cref{item:tel-expose-is-quick}.
    \item Perform the identified splits of clusters from $\mathcal{C}_\Tc$. By \Cref{item:tel-original-prefix-cost}, the total cost of split clusters is $\BigO(\log n)$, which will allow us to argue that all lazy updates stored in the split clusters can be propagated to the children in time $\BigOTilde(\log n)$.
    \item Perform the identified merges of clusters from $\mathcal{C}_\Uc$ without updating the weights of the edges or the sets of selected edges in the neighborhood data structures. By \Cref{item:tel-new-prefix-cost}, the total cost of merged clusters (with respect to the weights induced by the original tree $\Tc$) is $\BigO(\log n)$, so we will be able to compute the information stored in $\mathcal{C}_\Uc$ in time $\BigOTilde(\log n)$.
    \item Run the operation in the transiently exposed top tree $\Uc$. In most cases, the operation can be performed only by touching information stored in the root clusters of $\Uc$ in time $\BigOTilde(\log n)$. However, in some cases, we will additionally need to analyze a~vertical path of $\Uc$. \Cref{item:tel-new-vertical-cost} will allow us to argue that this can, too, be done in time $\BigOTilde(\log n)$.
    \item Perform a~\emph{transient unexpose}, i.e., revert the transient expose: split the transient clusters of the top tree and restore (merge) the hidden clusters. Again, by \Cref{item:tel-original-prefix-cost} and \Cref{item:tel-new-prefix-cost}, this can be achieved in time $\BigOTilde(\log n)$.
\end{enumerate}

The proof of \Cref{lem:transient-expose-lemma} follows from a~careful analysis of $\Expose{}$ and an~algebraic analysis similar to that of \Cref{lem:top-tree-telescope}.
%Therefore, we defer this proof to \Cref{sec:app-transient-expose}.

\begin{proof}[Proof of \Cref{lem:transient-expose-lemma}]
    Recall that $\Tc$ represents the tree $T$ with additional dummy edges: to each vertex $v \in V(T)$ we add a~dummy edge, say $e_v$, which has one endpoint at $v$ and the other at a~fresh leaf.
    We define $\mathcal{C}_\Tc$ as the set of strict ancestors of $e_u$ (if a~lone vertex $u$ is to be exposed) or $e_v$ and $e_w$ (if $v$ and $w$ are being exposed), and we split the bags in $\mathcal{C}_\Tc$ in the top-down order.
    Let $\Tc' \subseteq \Tc$ be the family of root clusters produced by this sequence of splits.

    Naturally, $\mathcal{C}_\Tc$ can be determined in $\BigO(\log n)$ time, and $\sum_{C \in \mathcal{C}_\Tc} \ncost_\Tc(C) \in \BigO(\log n)$ by \Cref{lem:top-tree-telescope} (all clusters of $\mathcal{C}_\Tc$ lie on two vertical paths in $\Tc$).
    Hence $\mathcal{C}_\Tc$ satisfies \Cref{item:tel-original-prefix-cost}.
    
    Now, given a~family of root clusters $\Tc'$, we construct a~top tree $\Uc$ with (transient) boundary vertices $u$ (resp.\ $v, w$) by arbitrarily merging clusters according to cases in \Cref{fig:toptree-cases}.
    This can be done via a~straightforward (and standard) case distinction; note here that we do not attempt to optimize the height or any other potential functions of the resulting transient top tree $\Uc$.
    Let $\mathcal{C}_\Uc$ be the set of merges produced by this procedure.
    
    We now prove that $\mathcal{C}_\Uc$ satisfies \Cref{item:tel-expose-is-quick,item:tel-new-prefix-cost}.
    First, every merge decreases the number of root top tree clusters, so $|\mathcal{C}_\Uc| < |\Tc'|$; and it decreases this number by at most $2$, hence $|\Tc'| \leq \BigO(1) + 2|\mathcal{C}_\Tc|$.
    Therefore $|\mathcal{C}_\Uc| \leq \BigO(1) + 2|\mathcal{C}_\Tc| \in \BigO(\log n)$.
    So the number of merges performed is $\BigO(\log n)$, and it is straightforward to find all of them in $\BigO(\log n)$ time.
    
    In order to bound $\sum_{C \in \mathcal{C}_\Uc} \ncost_\Tc(C)$, we introduce the following \emph{potential value} for families of root clusters.
    If $\mathcal{A}$ is a~family of clusters edge-partitioning $T$, we define $\zeta(\mathcal{A}) \coloneqq \sum_{C \in \mathcal{A}} \zeta(C)$, where
    \[ \zeta(C) = \begin{cases}
        \log\frac{\Weight_\Tc(v)}{\Weight_\Tc(\firstedge_{C,v})} + \log\frac{\Weight_\Tc(w)}{\Weight_\Tc(\firstedge_{C,w})} & \text{if }C\text{ is a~path cluster, }\bnd C = \{v, w\}, \\
        0 & \text{if }C\text{ is a~point cluster.}
    \end{cases} \]
    Obviously, $\zeta(\Tc), \zeta(\Uc) \in \BigO(\log n)$.
    Moreover:
    \begin{claim}
        Let $C \in \mathcal{A}$ and suppose $\mathcal{A}'$ is produced from $\mathcal{A}$ by splitting $C$ into child clusters.
        Then $\zeta(\mathcal{A}') = \zeta(\mathcal{A}) + (\ncost_\Tc(C) - 1)$.
    \end{claim}
    \begin{proof}[Proof of the claim]
        Straightforward case analysis.
    \end{proof}
    Applying the claim above multiple times, we get
    \begin{align*}
        \zeta(\Tc') &= \zeta(\Tc) + \sum_{C \in \mathcal{C}_\Tc} (\ncost_\Tc(C) - 1) \stackrel{(\ref{item:tel-original-prefix-cost})}{\leq} \BigO(\log n), \\
        \zeta(\Tc') &= \zeta(\Uc) + \sum_{C \in \mathcal{C}_\Uc} (\ncost_\Tc(C) - 1) \geq \sum_{C \in \mathcal{C}_\Uc} \ncost_\Tc(C) - |\mathcal{C}_\Uc|.
    \end{align*}
    Hence, $\sum_{C \in \mathcal{C}_\Uc} \in \BigO(\log n)$, fulfilling \Cref{item:tel-new-prefix-cost}.
    
    Finally, for \Cref{item:tel-new-vertical-cost}, observe that every vertical path in $\Uc$ can be split in two parts: the subpath closer to the root of $\Uc$ comprising transient clusters (whose total cost with respect to $\Tc$ is bounded by $\BigO(\log n)$ by \Cref{item:tel-new-prefix-cost}), and the remaining subpath containing clusters present in the original top tree $\Tc$ and forming a~vertical path in $\Tc$ (so its total cost is bounded by $\BigO(\log n)$ by \Cref{lem:top-tree-telescope}).
    Hence the total cost, with respect to $\Tc$, of clusters on any vertical path in $\Uc$ is bounded by $\BigO(\log n)$.
\end{proof}

\subsection{Cover level queries}
\label{ssec:tree-structure-coverlevels}
We now give an~implementation of the basic variant of the dynamic tree cover level data structure, allowing updates through $\Link{}$, $\Cut{}$, $\Connected{}$, $\Cover{}$, $\UniformUncover{}$ and $\LocalUncover{}$ and queries through $\CoverLevel{}$ and $\MinCoveredPair{}$, as well as laying foundations for the remaining types of updates and queries.
Formally, we prove the following:
% Now we claim the following data structure for cover levels:

\begin{lemma}
  \label{lem:ds-tree-cover-levels}
  There exists a~partial implementation \textsc{CL} of the restricted tree cover level data structure supporting $\Link{}$, $\Cut{}$, $\Select{}$, $\Connected{}$, $\Cover{}$, $\UniformUncover{}$, $\LocalUncover{}$, $\CoverLevel{}$ and $\MinCoveredPair{}$ in which:
%  \wninline{I think I'd prefer having explicitly worst-case/amortized time distinction wherever applicable}
%  \msinline{\footnotesize Fair point, I kinda just assumed that writing nothing is equivalent to claiming worst-case}
  \begin{itemize}
    \item Each transient expose and unexpose takes worst-case $\BigO(\log n)$ time, plus calls to $\Level{}$, $\LongZip{}$ and $\LongUnzip{}$ in the neighborhood data structures of worst-case total normalized cost $\BigO(\log n)$.
    % \item Additionally, each merge of a~non-transient cluster takes $\BigO(1)$ time, plus $\BigO(1)$ calls to $\Select{}$, $\SetWeight{}$.
    \item $\Link{}$, $\Cut{}$, $\Expose{}$ take worst-case $\BigO(\log n)$ time, plus calls to $\Level{}$, $\LongZip{}$, $\LongUnzip{}$, $\Select{}$ and $\SetWeight{}$ in the neighborhood data structures of worst-case total normalized cost $\BigO(\log^2 n)$.
    \item $\Connected{}$ takes worst-case $\BigO(\log n)$ time.
    \item $\Cover{}$, $\UniformUncover{}$, $\LocalUncover{}$, $\CoverLevel{}$ and $\MinCoveredPair{}$ take worst-case $\BigO(\log n)$ time, plus one transient expose and queries to the neighborhood data structures of the form $\Zip{}$, $\Unzip{}$, $\Level{}$ of worst-case total normalized cost $\BigO(\log n)$.
  \end{itemize}
\end{lemma}

\Cref{lem:ds-tree-cover-levels} satisfies the time bounds prescribed by \Cref{lem:top-tree-reduction} for $\Link{}$, $\Cut{}$, $\Expose{}$, $\Cover{}$, $\UniformUncover{}$, $\LocalUncover{}$, $\CoverLevel{}$ and $\MinCoveredPair{}$.
The remaining operations ($\FindSize{}$, $\Mark{}$, $\Unmark{}$, $\FindFirstReach{}$, $\FindStrongReach{}$) are not implemented here; the implementations of these operations will be given in the following sections, assuming access to the basic data structure \textsc{CL} from \Cref{lem:ds-tree-cover-levels}.

An~observant reader may notice that the time complexity bounds on the operations claimed by \Cref{lem:tree-structure}, excluding the calls to the neighborhood data structures, are slightly better than those stated in \Cref{lem:top-tree-reduction} (precisely by a~factor of $\hat{T} = \BigO(T(\lmax) \cdot \log \lmax)$).
This phenomenon has a~simple explanation: the extensions of the base data structure to be presented in the following sections will slightly degrade the efficiency of cluster merges.
Namely, on each cluster merge --- whether it is a~part of a~transient (un)expose or a~structural update of a~top tree --- we will need to perform some additional bookkeeping, which will increase the base time complexity of each transient expose, transient unexpose, link, cut, and expose from $\BigO(\log n)$ to $\BigO(\log n \cdot \hat{T})$.

%\wninline{The above statement seems far from true as \Cref{lem:tree-structure} mentions many other functions like FindSize or FindFirstReach etc. that are not mentioned in the above Lemma.}
%\msinline{Check if it's better this way?}

The remaining part of \Cref{ssec:tree-structure-coverlevels} is devoted to the proof of \Cref{lem:ds-tree-cover-levels}.
We refer the reader to the implementation of the data structure given in \Cref{ssec:tree-impl-cover}.

\subsubsection{Contents of the cover level data structure}

Given two adjacent edges $e_1, e_2$ with a~common endpoint $x$, let $c(e_1e_2)$ denote the cover level of the pair $e_1, e_2$, i.e., the largest integer $i$ such that $e_1$ and $e_2$ are in the same equivalence class of $\mathcal{L}^x_i$.
(If $e_1, e_2$ are in different equivalence classes of $\mathcal{L}^x_0$, we define that $c(e_1e_2) = -1$.)
Then for each path cluster $C$, we store the following value $\cover_C$:
\[
  \cover_C \coloneqq \min \left( \{c(e_1e_2) \,\mid\, e_1e_2 \subseteq \pi(C)\} \,\cup\, \{\lmax\} \right).
\]
Moreover, if $\bnd C = \{u, v\}$, let $\argcover_{C,u}$ be the pair of edges $e_1e_2 \subseteq \pi(C)$ closest to $u$ such that $\cover_C = c(e_1e_2)$, and $\argcover_{C,v}$ be the analogous pair of edges closest to $v$.
Here, we place $\argcover_{C, u} = \argcover_{C, v} = \bot$ if $\pi(C)$ has fewer than $2$ edges.

If $C$ is a~nonleaf path cluster, then $C$ is a~merge of two path clusters $A, B$ and a~point cluster $P$ (see case \CasePath).
Hence $\cover_C$ can be computed efficiently given $A$ and $B$: assuming that $\bnd C = \{v, w\}$, $\bnd A = \{v, x\}$, $\bnd B = \{x, w\}$, we have
\[
  \cover_C = \min \{ \cover_{A},\, \cover_{B},\, c(\firstedge_{A,x}\firstedge_{B,x}) \};
\]
and $\argcover_{C,\, \cdot}$ can be computed analogously.
Note that $c(\firstedge_{A,x}\firstedge_{B,x})$ can be computed by calling $N_x.\Level{\firstedge_{A,x}, \firstedge_{B,x}}$.

Observe now that the cover level queries can be resolved by transiently exposing the path $p \ldots q$ in $\mathcal{T}$.
This temporarily produces a~top tree $\mathcal{T}'$ containing a~root path cluster $C$ with $\bnd C = \{p, q\}$.
Then $\CoverLevel{p, q} = \cover_C$ and $\MinCoveredPair{p, q} = \argcover_{C, p}$.
After computing the answer to a~query we transiently unexpose the queried path.

\paragraph{Cover and uniform uncover}
In order to facilitate updates through $\Cover{}$ and $\UniformUncover{}$, for each path cluster $C$ with $\bnd C = \{v, w\}$ we additionally store two integers $\coverfrom_C, \covertop_C \in \{-1, 0, \ldots, \lmax\}$ representing pending lazy updates, with the following semantics:

\medskip

\emph{
  Assume that the cover level of the path $v \ldots w$ before performing the pending lazy updates was exactly $\coverfrom_C$.
  Then $\covertop_C \geq \max\{\coverfrom_C, \cover_C\}$ and the cluster path has pending lazy updates
  \[ \Cover{v, w, i} \quad\text{for } i = \coverfrom_C + 1, \ldots, \covertop_C \]
  followed by
  \[ \UniformUncover{v, w, i} \quad\text{for } i = \covertop_C, \ldots, \cover_C + 1. \]
  The lazy updates have been already applied to the cluster $C$ (so the cover level of $\pi(C)$ is $\cover_C$), but not to the descendants $C$.
}

\medskip

Note that the sequence of lazy updates above causes each pair $e_1e_2 \subseteq v \ldots w$ at cover level at most $\covertop_C$ to have cover level precisely $\cover_C$. On the other hand, the pairs of edges at cover level strictly above $\covertop_C$ preserve their level.

Lazy updates are propagated from cluster $C$ of the top tree to its children whenever $C$ is split (either permanently or transiently).
The propagation is facilitated by the following observation: let $v, w$ be two different vertices of the tree and let $i \in \{0, \ldots, \lmax\}$ be such that $c(v\ldots w) \geq i$ (so that the call to $\UniformUncover{v, w, i}$ is legal).
Then the sequence of updates $\UniformUncover{v, w, i}, \Cover{v, w, i}$ leaves the cover levels of the pairs of all edges unchanged.

So suppose that a~path cluster $C$ is split.
This requires us to apply to each path child $D$ of $C$ the lazy updates pending in $C$; note that $D$ may have some pending updates itself, so we have to combine the sequence of lazy updates in $D$ with the sequence of lazy updates pending in $C$.
Note that we have $\cover_D \geq \coverfrom_C$, since $\cover_D$ and $\coverfrom_C$ denote the cover levels of $\pi(D)$ and $\pi(C)$, respectively, before applying the updates pending in $C$.
Then:
\begin{itemize}
  \item If $\covertop_C < \cover_D$, then no lazy updates pending in $C$ influence the cover levels of any pair of edges on the cluster path of $D$, so the values $\cover_D$, $\coverfrom_D$, $\covertop_D$ are left unchanged.
  \item If $\covertop_C \geq \cover_D$, then the lazy updates pending in $C$, when applied after the lazy updates pending in $D$, will cause all pairs of consecutive edges of $\pi(D)$ at cover level at most $\max\{\covertop_C, \covertop_D\}$ to be at cover level $\cover_C$.
  Hence we can set $\cover_D \coloneqq \cover_C$, $\covertop_D \coloneqq \max\{\covertop_C, \covertop_D\}$.
  This correctly propagates lazy updates from $C$ to $D$.
\end{itemize}

Also, if $C$ has two path clusters $D_1, D_2$ as children (case \CasePath in \Cref{fig:toptree-cases}, with $\bnd C = \{v, w\}$, $\bnd D_1 = \{v, x\}$, $\bnd D_2 = \{x, w\}$), we need to propagate the lazy updates to the pair of edges on $\pi(C)$ incident to $x$.
To this end, let $e_1 = \firstedge_{D_1,x}$ and $e_2 = \firstedge_{D_2,x}$ be the edges of $\pi(C)$ incident to $x$ that are part of cluster paths $\pi(D_1), \pi(D_2)$, respectively.
Note that by \Cref{lem:interface-firstedge}, we are guaranteed that $e_1, e_2$ are the two selected edges in the neighborhood data structure $N_x$; so we first determine $c(e_1e_2)$ by calling $N_x.\SelectedLevel$.
Then, if $c(e_1e_2) \leq \covertop_C$, we adjust $c(e_1e_2)$ to $\cover_C$ by calling either $\LongZip{c(e_1e_2), \cover_C}$ if $c(e_1e_2) < \cover_C$, or $\LongUnzip{c(e_1e_2), \cover_C}$ if $c(e_1e_2) > \cover_C$.

%In all cases, after a~split, we indicate that $C$ has no pending updates by resetting both $\coverfrom_C$ and $\covertop_C$ to $\cover_C$.
%\wninline{I am a bit confused by the above. Does $C$ exist at all after being split? Do you mean that C may be a hidden cluster that is not part of ``the current top tree'' after being split, but it will be revived shortly?}
%\msinline{I commented out the sentence above. Does it work for you? (I think it's ok for me)}

\subparagraph{Remark}
\Cref{lem:interface-firstedge} may not hold if $C$ is a~transient cluster created temporarily during a~transient expose.
Due to that, we specify that transient nodes cannot hold lazy updates, i.e., we require that $\cover_C = \coverfrom_C = \covertop_C$ for all such clusters $C$.

\bigskip

We move to the implementation of $\Cover{p, q, i}$.
First suppose that the set of external boundary vertices of $T$ happens to be $\{p, q\}$.
Let $R$ be the root path cluster of the top tree, with $\bnd R = \{p, q\}$.
The preconditions of $\Cover{}$ guarantee that $\cover_R \geq i - 1$.
Hence if $\cover_R = i - 1$, we set $\cover_R \coloneqq i$ and $\covertop_R \coloneqq \max\{\covertop_R, i\}$.

In the general case, we begin by transiently exposing the path $p \ldots q$.
This temporarily rebuilds a~prefix of the top tree, producing in particular a~transient root path cluster $R$ with $\bnd R = \{p, q\}$.
We then cover the path recursively: let $C$ be a~path cluster with $\pi(C) \subseteq p \ldots q$; initially, $C = R$.
If $C$ comes from the original top tree $\Tc$ (i.e., it is not transient), we repeat our procedure above: if $\cover_C = i - 1$, we set $\cover_C \coloneqq i$ and $\covertop_C \coloneqq \max\{\covertop_C, i\}$.
On the other hand, if $C$ is transient, we cannot apply lazy updates to $C$.
Since $C$ is a~path cluster, it has two path clusters $A, B$ and a~point cluster $P$ as children; assume that $\bnd C = \{v, w\}$, $\bnd A = \{v, x\}$, $\bnd B = \{x, w\}$, $\bnd P = \{x\}$.
We recursively cover the cluster paths of $A$ and $B$; it only remains to update, if necessary, the cover level of the pair of edges $e_1 \coloneqq \firstedge_{A,x}$, $e_2 \coloneqq \firstedge_{B,x}$.
So if $N_x.\Level{e_1, e_2} = i - 1$, we run $N_x.\Zip{e_1, e_2}$.
Finally, after the recursive scheme finishes, we transiently unexpose the path $p \ldots q$.

\smallskip

$\UniformUncover{p, q, i}$ is similar: if the set of external boundary vertices is $\{p, q\}$, then we are guaranteed that $\cover_R \geq i$.
If $\cover_R = i$, we set $\cover_R \coloneqq i - 1$.
In the general case, we will exploit the fact that the cover level data structure is \emph{restricted}, meaning that $p \ldots q$ is a~subpath of the currently exposed path $a \ldots b$.
We will use it through the following observation:

\begin{lemma}
  \label{lem:restricted-interface}
  Let $vx$, $xw$ be two consecutive edges of the path $a \ldots b$.
  Then $\interface_x = \{\vec{xv}, \vec{xw}\}$.
\end{lemma}
\begin{proof}
  Let $R_{\text{orig}}$ be the root path cluster of the top tree, with $\bnd R_{\text{orig}} = \{a, b\}$.
  Split $R_{\text{orig}}$ recursively into progressively smaller path clusters (via case \CasePath of \Cref{fig:toptree-cases}) until producing a~path cluster $C = I_x$.
  We necessarily have $\pi(C) \subseteq a \ldots b$ and hence, by \Cref{lem:interface-firstedge}, we also have $\interface_x = \{\vec{xv}, \vec{xw}\}$.
\end{proof}

We again begin by transiently exposing the path $p \ldots q$, and follow by describing a~recursive scheme for uncovering a~cluster path of a~path cluster $C$ with $\pi(C) \subseteq p \ldots q$.
If $C$ comes from the original top tree $\Tc$, we test whether $\cover_C = i$, and if so, we set $\cover_C \coloneqq i - 1$.
Otherwise, let $A, B, P$ be the children of $C$ in the top tree, with $x = \bnd A \cap \bnd B$.
We recursively uncover the cluster paths of $A$ and $B$, and it remains to uncover, if necessary, the pair of edges $e_1 \coloneqq \firstedge_{A,x}$, $e_2 \coloneqq \firstedge_{B,x}$.
Observe now that since $e_1e_2 \subseteq p \ldots q \subseteq a \ldots b$, we have $\interface_x = \{e_1, e_2\}$ by \Cref{lem:restricted-interface}; hence, if $N_x.\SelectedLevel{\xspace} = i$, we apply $N_x.\LongUnzip{i, i - 1}$. %\footnote{Note that this allows us to avoid the computationally expensive call of the form $N_x.\Unzip{}$.}

\paragraph{Local uncover}
We are left to implement $\LocalUncover{e_1, e_2, i}$.
Let $e_1 = px, e_2 = xq$ and recall that the pair $e_1e_2$ is at cover level precisely $i$.
Observe now that the update at hand may change the values $\cover_C$, $\argcover_C$ only for those clusters $C$ for which $x$ is an~internal vertex of the cluster path $\pi(C)$.
Therefore, we transiently expose the vertex $x$.
This way, the resulting transient top tree has no clusters with $x$ as an~internal vertex.
This allows us to update the cover levels by simply calling $N_x.\Unzip{e_1, e_2, i}$.
Finally, we perform a~transient unexpose.

%\wninline{The above looks to me correct, but also, idk, suspicious :p? Why not just call $\Expose{x}$?}
%\msinline{That's simpler, thanks.}

\subsubsection{Complexity analysis of the data structure}
We now analyze the time complexity guarantees of all elements of our data structure.

\paragraph{Transient exposes and unexposes}
Observe that each split of a~hidden cluster under a~transient expose requires the recomputation of several values and pointers to edges and at most one call to $\LongZip{}$ or $\LongUnzip{}$ in the corresponding neighborhood data structure.
A~split of a~transient cluster during a~transient unexpose requires no additional bookkeeping since no lazy updates are ever stored in such clusters.
The total number of hidden and transient clusters is $\BigO(\log n)$ by \Cref{lem:transient-expose-lemma}(\ref{item:tel-expose-is-quick}), so all splits can be done in total time $\BigO(\log n)$.

Now consider a~merge of a~cluster $C$ (either a~transient cluster during a~transient expose or a~hidden cluster during a~transient unexpose).
This recomputes several values and pointers to edges and, if $C$ is a~path cluster with path children $A, B$ and midpoint $x$, issues at most one call to $N_x.\Level{\firstedge_{A,x}, \firstedge_{B,x}}$.
These operations are performed, by \Cref{lem:neighborhood-ds} and the definition of a~cluster cost, in total time $\BigO(\ncost_\Tc(C))$.
Hence by \Cref{lem:transient-expose-lemma}(\ref{item:tel-original-prefix-cost}) and \Cref{lem:transient-expose-lemma}(\ref{item:tel-new-prefix-cost}), all merges can be performed in total time $\BigO(\log n)$, satisfying the requirements of \Cref{lem:ds-tree-cover-levels}.

\paragraph{Structural updates (link, cut and expose)}
Each $\Link{}$, $\Cut{}$, and $\Expose{}$ is processed by the top trees data structure in worst-case time $\BigO(\log n)$ and produces a~sequence of $\BigO(\log n)$ splits and merges.
However, each split or merge additionally updates the set of interface edges incident to a~constant number nodes of the tree, as well as the weights of a~constant number of nodes.
%\wninline{I'd specify if the constant number of nodes is \textit{per each update} or \textit{in total}. I suppose per each update?}
%\msinline{Fixed? Please check}
Hence, on top of the computations above, we also need to additionally perform $\Select{}$ and $\SetWeight{}$ in neighborhood data structures a~constant number of times per split or merge.
Each call to the neighborhood data structure has normalized cost at most $\BigO(\log n)$, so all in all, each structural update takes time $\BigO(\log n)$, plus calls to the neighborhood data structures of total normalized cost $\BigO(\log^2 n)$.

\paragraph{Connectivity query}
Each $\Connected{}$ query is resolved entirely by top trees; this takes worst-case time $\BigO(\log n)$.

\paragraph{Cover level queries}
Both $\CoverLevel{}$ and $\MinCoveredPair{}$ resolve to a~single transient expose and reading appropriate information from the resulting transient root clusters.
This obviously satisfies the bounds of \Cref{lem:ds-tree-cover-levels}.

\paragraph{Path cover level updates}
For $\Cover{}$, we perform a~single transient expose, but now we additionally issue calls of the form $\Level{}$ and $\Zip{}$ in the neighborhood data structure.
Precisely, if $C$ is a~transient path cluster with path cluster children $A$, $B$ and midpoint $x$, then we call $N_x.\Level{\firstedge_{A,x}, \firstedge_{B,x}}$ and $N_x.\Zip{\firstedge_{A,x}, \firstedge_{B,x}}$.
This calls have total normalized cost $\BigO(\ncost_\Tc(C))$; hence all calls to the neighborhood data structures across the invocation of $\Cover{}$ have total normalized cost $\BigO(\log n)$ by \Cref{lem:transient-expose-lemma}(\ref{item:tel-new-prefix-cost}).

For $\UniformUncover{}$, the argument is even simpler: past a~transient expose, we only issue $\BigO(\log n)$ calls of the form $\SelectedLevel{}$ and $\LongUnzip{}$ to the neighborhood data structures, hence these calls have total normalized cost $\BigO(\log n)$.

\paragraph{Local uncover}
Finally, $\LocalUncover{}$ requires a~single transient expose.
This again is consistent with \Cref{lem:ds-tree-cover-levels}.

%\wninline{There isn't any mention of pseudocode in \Cref{ssec:tree-impl-cover}}

\subsection[Counting i-reachable vertices]{Counting $i$-reachable vertices}
\label{ssec:tree-structure-counting}
In this section we will implement $\FindSize{}$ in our tree cover level data structure.
We will first describe additional information we will store for each cluster and show how this information is maintained, and then we will present how to use this information to implement $\FindSize{}$.
Namely, we will prove the following statement:

\begin{lemma}
  \label{lem:ds-tree-find-size}
  The data structure \textsc{CL} from \Cref{lem:ds-tree-cover-levels} can be extended to a~restricted tree cover level data structure \textsc{CR} additionally supporting $\FindSize{}$ so that:
  \begin{itemize}
    \item Each transient expose and unexpose takes additionally worst-case $\BigO(\log n \cdot T(\lmax) \log \lmax)$ time, plus calls to $\UpdateCounters{}$ and $\SumCounters{}$ in the neighborhood data structures of worst-case total normalized cost $\BigO(\log n)$.
    \item $\FindSize{}$ takes worst-case $\BigO(T(\lmax) \log n)$ time, plus a~single transient expose and unexpose, and calls to $\Level{}$ and $\SumCounters{}$ in the neighborhood data structures of worst-case total normalized cost $\BigO(\log n)$.
  \end{itemize}
\end{lemma}

The pseudocode implementation of the counting extension can be found in \Cref{ssec:tree-impl-count}.

In the following description we assume that at the point of the query, the tree $T$ is represented by a~top tree $\Tc$, with the neighborhood data structures holding the types of edges and weights defined with respect to $\Tc$.
The first step of the query $\FindSize{p, q, i}$ is, however, to transiently expose the pair $\{p, q\}$ in $T$.
%\wninline{The above sounds like ... buttery butter :D}
This causes $\Tc$ to be temporarily replaced with another top tree $\Uc$, also representing $T$, though the neighborhood data structures still hold the weights defined with respect to $\Tc$.

\paragraph{Counters for neighborhood data structures}
For every cluster edge $\vec{vw} \in \cluster^\Uc_v$ we store a~counter vector $\clustercntvec_{\vec{vw}}$, where $\clustercnt_{\vec{vw}, i}$ for $i \in \{0, \ldots, \lmax\}$ denotes the number of vertices in $T[\vec{vw}]$ that are $i$-reachable from $vw$.
This counter vector is stored together with the element $\vec{vw}$ of the neighborhood data structure $N_v$.
Meanwhile, for every interface edge $\vec{vw}$ with respect to $\Uc$, we store the all-zero counter vector ${\bf 0}$ for the element $\vec{vw}$ of $N_v$.
Whenever any counter vector $\clustercntvec_{\vec{vw}}$ is updated, it is stored in $N_v$ through the call $N_v.\UpdateCounters{vw, \clustercntvec_{\vec{vw}}}$.

We will access the counter vectors in $N_v$ using $\SumCounters{}$ through the following lemma.

\begin{lemma}
  \label{lem:excl-sum-counters}
  Let $v \in V(T)$ be a~vertex of $T$ currently represented by a~top tree $\Uc$.
  For any $A_+, A_- \subseteq \interface^\Uc_v$, we can determine a~counter vector ${\bf c}$ such that, for every $i \in \{0, \ldots, \lmax\}$, $c_i$ is the number of vertices in $\bigcup_{e \in \cluster^\Uc_v} T[e]$ that are $i$-reachable from at least one edge of $A_+$, but not $i$-reachable from any edge of $A_-$.
  This query takes worst-case time $\BigO(T(\lmax))$, plus queries to the neighborhood data structures of worst-case total normalized cost
  \begin{equation}
    \label{eq:reach-query-time}
    \BigO\left(\sum_{e \in A_+ \cup A_-} \left(1 + \log \frac{\Weight_\Tc(v)}{\Weight_\Tc(e)} \right) \right).
  \end{equation}
\end{lemma}

We leave a~few technical remarks before proving \Cref{lem:excl-sum-counters}.
First, it might be enticing to reduce \Cref{lem:excl-sum-counters} to the case where $A_- = \emptyset$ by observing that $c_i = x_i - y_i$, where $x_i$ (resp., $y_i$) is the number of vertices in $\bigcup_{e \in \cluster^\Uc_v} T[e]$ that are $i$-reachable from at least one edge of $A_+ \cup A_-$ (resp., $A_-$).
However, this reduction cannot be applied here since counters in counter vectors cannot be subtracted.
The technical contribution of \Cref{lem:excl-sum-counters} is to present how ${\bf c}$ can be determined using only the operations supported by counter vectors.
Second, even though the sets $A_+$, $A_-$ have total bounded size ($|A_+ \cup A_-| \leq |\interface^\Uc_v| \leq 2$), we find it more illuminating to give a~proof that proceeds inductively on the size of $A_+$.

\begin{proof}
%	\wninline{I don't know what to suggest, but I find it very unsettling and potentially making the reader question if they understand correctly, that we are doing an induction on the size of a set whose size is at most two, or iterating over all pairs from $A_+ \times A_-$, where this set has at most one element :p...}
%	\wninline{I was wondering why it's not simpler to first design a function $N_v.\SumCounters{A_+}$ and then say that $N_v.\SumCounters{A_+, A_-} = N_v.\SumCounters{A_+ \cup A_-} - N_v.\SumCounters{A_-}$ and only after a long time realized that this would be illegal cause we are not allowed to subtract, because of apx counters (and because that would not generalize to boolean algebra later too). Do we want to preemptively remind the reader of this technical detail in this context here?}
%	\msinline{Added some remarks above. Do they satisfy you?}
  Induction on the size of $A_+$.
  If $A_+ = \emptyset$, then the answer is ${\bf 0}$.
  Now assume that $|A_+| \geq 1$ and pick $e \in A_+$.
  Let ${\bf c}_{\neq e}$ be the recursively computed counter vector for the set $A_+ \setminus \{e\}$, and let ${\bf c}_e = N_v.\SumCounters{e}$.
  Observe that a~vertex $x \in T[f]$ for a~cluster edge $f \in \cluster^\Uc_v$ is $i$-reachable from $e$ but not from $(A_+ \cup A_-) \setminus \{e\}$ if all of the following conditions hold:
  \begin{itemize}
    \item $x$ is $i$-reachable from $f$;
    \item the cover level of the pair $e, f$ is at least $i$;
    \item the cover level of all pairs $e', e$ for $e' \in (A_+ \cup A_-) \setminus \{e\}$ is strictly less than $i$. (Equivalently, assuming the previous condition, the cover level of all pairs $e', f$ is strictly less than $i$.)
  \end{itemize}
  The number of vertices $i$-reachable from $f$ in $T[f]$ is simply $\clustercnt_{f,i}$, and so the number of vertices that satisfy the first two conditions above is $c_{e,i}$.
  Then, letting $p$ to be the maximum cover level of $e$ with all the remaining edges of $A_+ \cup A_-$ (which can be determined by $|A_+ \cup A_-| - 1$ calls $N_v.\Level{e, \cdot}$), we find that the number of vertices satisfying all three conditions is $c_{e,i}$ if $i > p$, or $0$ otherwise.
  
  Therefore, the final counter vector ${\bf c}$ is precisely ${\bf c} = {\bf c}_{\neq e} + [{\bf 0} : (p+1) : {\bf c}_e]$.
  Since $|A_+|, |A_-| \in \BigO(1)$, it can be verified that the total query time is bounded from above by $\BigO(T(\lmax))$, plus queries to the neighborhood data structures of total normalized cost \eqref{eq:reach-query-time}.
\end{proof}
Extending the notation slightly, we say that $N_v.\SumCounters{A_+, A_-}$ denotes the operation described in Lemma~\ref{lem:excl-sum-counters}.
We will frequently have $A_- = \emptyset$ and in this case we will use the shorthand notation $N_v.\SumCounters{A_+}$.

\paragraph{Counters for path clusters}
Let $C$ be a~(slim) path cluster with $\bnd C = \{v, w\}$.
We say that a~vertex $y \in C$ such that $\meet(v, w, y) \notin \{v, w\}$ is \emph{$(\geq j, i)$-reachable from $\firstedge_{C,v}$} for $i, j \in \{0, \ldots, \lmax\}$ if there exists an~edge $e \in \pi(C)$ incident to $\meet(v, w, y)$ such that $c(\firstedge_{C,v}, e) \geq j$ and $y$ is $i$-reachable from $e$.
(Note there are two edges on $\pi(C)$ incident to $\meet(v, w, y)$; one separates $y$ from $v$, and the other separates $y$ from $w$.)
We also say that $y \in C$ is \emph{$(=j, i)$-reachable from $\firstedge_{C,v}$} if it is $(\geq j, i)$-reachable but not $(\geq j+1, i)$-reachable from $\firstedge_{C,v}$.
Then we define the following counter vectors.
\begin{itemize}
  \item $\totalcntvec_C$, where $\totalcnt_{C, i}$ for $i \in \{0, \ldots, \lmax\}$ is the number of vertices $y \in C$ such that $y \notin \{v, w\}$ and $y$ is $i$-reachable from $v \ldots w$;
  \item $\totalcntvec_{C,v}$, where $\totalcnt_{C, v, i}$ is the number of vertices $y \in C$ such that $y \notin \{v, w\}$ and $y$ is $i$-reachable from $\firstedge_{C, v}$; and $\totalcntvec_{C,w}$, defined analogously;
  \item $\diagcntvec_{C,v,j}$ for $j \in \{0, \ldots, \lmax\}$, where $\diagcnt_{C,v,j,i}$ is the number of vertices $y \in C$ such that $y \notin \{v, w\}$ and $y$ is $(=j, i)$-reachable from $\firstedge_{C,v}$.
  We also define $\diagcntvec_{C,w,j}$ analogously.
  Clearly, we can represent $\diagcntvec_{C,v}$ and $\diagcntvec_{C,w}$ as counter matrices.
\end{itemize}

For technical reasons, we only store $\totalcntvec_C$ explicitly and do not store $\totalcntvec_{C, \cdot}$ or $\diagcntvec_{C, v}$.
Instead, we maintain the counter matrices $\diagcntvec^\star_{C,v}$ and $\diagcntvec^\star_{C,w}$ that represent the corresponding values of $\diagcntvec$ in $C$ \emph{before} the application of pending lazy updates, if any, to $C$.
Then the corresponding counter matrices $\diagcntvec_{C,v}$, $\diagcntvec_{C,w}$ representing the state of $C$ \emph{after} applying these updates can be determined ``on the fly'' from $\diagcntvec^\star_{C,v}$ as:
\[
  \diagcnt_{C,v,j,i} =
  \begin{cases}
    0 & \text{if }j \leq \covertop_C\text{ and }j \neq \cover_C, \\
    \sum_{j' = 0}^{\covertop_C} \diagcnt^\star_{C,v,j',i} & \text{if }j = \cover_C, \\
    \diagcnt^\star_{C,v,j,i} & \text{if }j > \covertop_C.
  \end{cases}
\]

Also we have that $\totalcnt_{C,v,i} = \sum_{j = i}^{\lmax} \diagcnt_{C,v,j,i}$.
Observe thus that both $\diagcntvec$ and $\totalcntvec$ can be computed from $\diagcntvec^\star$ in a~constant number of counter vector and matrix operations, which takes time $\BigO(T(\lmax) \log \lmax)$ by \Cref{lem:counter-matrices}.

\paragraph{Computing counters}
The values of all counters described above are computed for a~cluster during a~merge.
For leaf clusters, we can see that all counter vectors are identically zero.
We now distinguish cases on how a~cluster $C$ splits into child clusters $A$, $B$ (and possibly a~point cluster $P$).
\begin{itemize}
  \item If $\bnd A = \bnd B = \bnd C = \{v\}$ (case \CasePointToPointPoint in \Cref{fig:toptree-cases}), no counters need to be computed or updated.
  \item If $\bnd A = \{v, x\}$, $\bnd B = \{x\}$, $\bnd C = \{v\}$ (case \CasePointToPathPoint), then $C$ is the deepest cluster in the top tree that entirely contains $T[\firstedge_{A,v}]$; hence the vector $\clustercntvec_{\firstedge_{A,v}}$ should be updated.
  Recall that $\interface^\Uc_x = \{\firstedge_{A,x}\}$.
  Let ${\bf c}^x = N_x.\SumCounters{\{\firstedge_{A,x}\}}$, so that $c^x_i$ is the number of vertices $y \in \bigcup_{\vec{xx'} \neq \firstedge_{A,x}} T[\vec{xx'}]$ that are $i$-reachable from $\firstedge_{A,x}$.
  Then we can compute $\clustercnt_{\firstedge_{A,v}, i}$ for $i \in \{0, \ldots, \lmax\}$ as the sum of the following values:
    \begin{itemize}
      \item $\totalcnt_{A, v, i}$ (the number of vertices $y$ with $\meet(v, x, y) \notin \{v, x\}$ that are $i$-reachable from $\firstedge_{A, v}$);
      \item $1 + c^x_i$ if $\cover_A \geq i$ (the number of vertices $y$ with $\meet(v, x, y) = x$ that are $i$-reachable from $\firstedge_{A, v}$).
      The $+1$ summand comes from the vertex $y = x$.
    \end{itemize}
    Hence $\clustercntvec_{\firstedge_{A,v}}$ can be computed from $\totalcntvec_{A,v}$ and ${\bf c}^x$ in a~constant number of counter vector operations.
  \item If $\bnd A = \{v, x\}$, $\bnd B = \{x, w\}$, $\bnd P = \{x\}$ and $\bnd C = \{v, w\}$ (case \CasePath), then we first calculate counter vectors:
    \begin{itemize}
      \item ${\bf c}^{AB} = N_x.\SumCounters{\{\firstedge_{A,x}, \firstedge_{B,x}\}}$;
      \item ${\bf c}^A = N_x.\SumCounters{\{\firstedge_{A,x}\}}$ (and ${\bf c}^B$ analogously);
      \item ${\bf c}^{B \setminus A} = N_x.\SumCounters{\{\firstedge_{B,x}\}, \{\firstedge_{A,x}\}}$ (and ${\bf c}^{A \setminus B}$ analogously).
    \end{itemize}
  Noting that $\interface^\Uc_x = \{\firstedge_{A,x}$, $\firstedge_{B,x}\}$, we can compute $\totalcnt_{C,i}$ for $i \in \{0, \ldots, \lmax\}$ as the sum of:
  \begin{itemize}
    \item $\totalcnt_{A, i}$,
    \item $\totalcnt_{B, i}$,
    \item $1 + c^{AB}_i$.
  \end{itemize}
  (These values represent the number of vertices $y$ that are $i$-reachable from $v \ldots w$ and such that $\meet(v, w, y)$ is, respectively: strictly between $v$ and $x$, strictly between $x$ and $w$, and equal to $x$.)
  Hence $\totalcntvec_C$ is computed from $\totalcntvec_A$, $\totalcntvec_B$, and ${\bf c}^{AB}$ within a~constant number of counter vector operations.
  
  Then, letting $\ell_{AB}$ be the cover level of the pair of edges $\firstedge_{A,x}$, $\firstedge_{B,x}$ (i.e., $\ell_{AB} = N_x.\Level{\firstedge_{A,x}, \firstedge_{B,x}}$), %\wn{isn't this also $N_x.\SelectedLevel$?}
  we compute $\diagcnt^\star_{C, v, j, i}$ as the sum of the following values:
  \begin{itemize}
    \item $\diagcnt_{A, v, j, i}$,
    \item $\diagcnt_{B, x, j, i}$ if $\min\{\cover_A, \ell_{AB}\} > j$,
    \item $\sum_{j'=j}^{\lmax} \diagcnt_{B, x, j', i}$ if $\min\{\cover_A, \ell_{AB}\} = j$,
    \item $1 + c^A_i$ if $\cover_A = j$ and $\ell_{AB} < j$,
    \item $1 + c^{AB}_i$ if $\cover_A = j$ and $\ell_{AB} \geq j$,
    \item $c^{B \setminus A}_i$ if $\cover_A > j$ and $\ell_{AB} = j$.
  \end{itemize}
  (Again, the values $\diagcnt_{A, v, \cdot, \cdot}$; the values $\diagcnt_{B, v, \cdot, \cdot}$; and the values $(1 +) c^\star_i$ represent the number of vertices $y$ that are $i$-reachable from $\firstedge_{A,v}$ and such that $\meet(v, w, y)$ is, respectively: strictly between $v$ and $x$; strictly between $x$ and $w$; and equal to $x$.)

%\wninline{I expanded the explanation below. Please check.}
  For efficiency, we also observe that $\diagcnt_{A,v,j,i} = 0$ for $j < \cover_A$. This allows us to express the sum of $\diagcntvec_{A, v}$ and the part of $\diagcntvec_{B, x}$ for $j < \min\{\cover_A, \ell_{AB}\}$ as the result of a matrix splicing operation $\matrixsplice{\diagcntvec_{B, x}}{\min\{\cover_A, \ell_{AB}\}}{\diagcntvec_{A, v}}$ and avoid an~expensive element-wise addition of matrices.
  This enables us to compute $\diagcntvec_{C,v}$ in a~constant number of counter vector and matrix operations; see \Cref{ssec:tree-impl-count} for technical details.
  Then $\diagcntvec^\star_{C,w}$ is computed analogously.
\end{itemize}

\paragraph{Time complexity of a~merge}
By Lemma~\ref{lem:counter-matrices}, for any path cluster $D$ and $v \in \bnd D$, both $\diagcntvec_{D,v}$ and $\totalcntvec_{D,v}$ can be computed on the fly from $\diagcntvec^\star_{D,v}$ in $\BigO(T(\lmax) \log \lmax)$ time.
After this computation, we can assume access to these values in child clusters $A$, $B$ of $C$.

In each case of the merge, we call $N_x.\SumCounters{}$ a~constant number of times, each time with some set of interface edges for $x$ with respect to $\Uc$.
Hence by Lemma~\ref{lem:excl-sum-counters}, all these calls take total time $\BigO(T(\lmax))$, plus calls to the neighborhood data structures of total normalized cost $\BigO\left(\sum_{e \in \interface^\Uc_x} \log \frac{\Weight_\Tc(x)}{\Weight_\Tc(e)} \right)$.

If $\bnd A = \{v, x\}$, $\bnd B = \{x\}$, $\bnd C = \{v\}$ (case \CasePointToPathPoint), the counter vector $\clustercntvec_e$ is computed from vectors $\totalcntvec_{A,v}$ and ${\bf c}^x$ in a~constant number of vector operations in time $\BigO(T(\lmax))$.
Moreover, we call $N_x.\SumCounters{\{\firstedge_{A,x}\}}$ and $N_v.\UpdateCounters{\firstedge_{A,v}, \cdot}$.
By \Cref{lem:excl-sum-counters,lem:neighborhood-ds}, these take total time $\BigO(T(\lmax))$, plus calls to the neighborhood data structures of total normalized cost $\BigO(\ncost_\Tc(C))$.

Finally, for the case $\bnd A = \{v, x\}$, $\bnd B = \{x, w\}$, $\bnd C = \{v, w\}$ (case \CasePath), the counter vector $\totalcntvec_C$ is computed in a~constant number of vector operations from $\totalcntvec_{A}$, $\totalcntvec_{B}$ and ${\bf c}^{AB}$. And $\diagcntvec^\star_{C,v}$ is determined in a~constant number of matrix and vector operations from $\diagcntvec_{A,v}$, $\diagcntvec_{B,x}$, ${\bf c}^A$, ${\bf c}^B$ and ${\bf c}^{AB}$.
This again can be done in time $\BigO(T(\lmax) \cdot \log \lmax)$.
Additionally, we call $N_x.\SumCounters{}$ a~constant number of times, which takes time $\BigO(T(\lmax))$, plus calls to the neighborhood data structures of total normalized cost $\BigO(\ncost_\Tc(C))$.

Summing up, a~single merge of $C$ takes time $\BigO(T(\lmax) \cdot \log \lmax)$, plus calls to the neighborhood data structures of total normalized cost $\BigO(\ncost_\Tc(C))$.
Therefore, by the Transient Expose Lemma (\Cref{lem:transient-expose-lemma}), each transient expose and unexpose takes additionally time $\BigO(\log n \cdot T(\lmax) \log \lmax)$, plus calls to the neighborhood data structures of total normalized cost $\BigO(\log n)$.

%\msinline{todo}

\paragraph{Resolving $\FindSize{p, q, i}$}
Armed with the set of counters defined above, it is straightforward to determine the answer to the size query: ``\emph{count the number of vertices of $T$ that are $i$-reachable from any edge on the path $p \ldots q$}''.
Namely, upon a~transient expose of the path $p \ldots q$, we hold a~transient top tree $\Uc$ representing the tree $T$.
Recall that $\Uc$ consists, in fact, of three root clusters: a~slim path cluster $C_{pq}$ with $\bnd C_{pq} = \{p, q\}$ and two point clusters $C_p, C_q$ with $\bnd C_p = \{p\}$, $\bnd C_q = \{q\}$.
We now find the result of $\FindSize{p, q, i}$ to be the sum of the following values:

\begin{itemize}
    \item $\totalcnt_{C_{p,q}, i}$: the number of vertices $y$ with $\meet(y,p,q) \notin \{p, q\}$ that are $i$-reachable from $p \ldots q$;
    \item $1 + \left(N_p.\SumCounters{\{\firstedge_{C_{p,q},p}\}}\right)_i$: the number of vertices $y$ with $\meet(y,p,q) = p$ that are reachable from $p \ldots q$, including $y = p$;
    \item $1 + \left(N_q.\SumCounters{\{\firstedge_{C_{p,q},q}\}}\right)_i$: the analogous count of vertices $y$ satisfying $\meet(y, p, q) = q$.
\end{itemize}

Afterwards, we restore the original set of exposed vertices.
It is easy to see that the whole process requires one transient expose and transient unexpose and additionally time $\BigO(\log n)$, plus calls to the neighborhood data structures of total normalized cost $\BigO(\log n)$.
This satisfies the requirements of \Cref{lem:ds-tree-find-size}.

\subsection[Implementing FindFirstReach and FindStrongReach]{Implementing $\FindFirstReach{}$ and $\FindStrongReach{}$}
\label{ssec:tree-structure-findfirst}
We now move to the implementation of \FindFirstReach{} and \FindStrongReach{}.
Recall that in $\FindFirstReach{p, q, i}$, we have to return a~tuple $(ab, c, y)$, where $ab \in E(p \ldots q)$ with $a$ closer than $b$ to $p$, $c \in \{a, b\}$ and $y$ is an~$i$-marked vertex $i$-reachable from $p \ldots q$ at $ab$ through $c$. Ties are broken by first minimizing the distance between $p$ and $a$ in $T$, and then minimizing the distance between $p$ and $c$.
On the other hand, in $\FindStrongReach{p, q, e, b, i}$, we are required to find an~$i$-marked vertex $y$ that is \emph{strongly} $i$-reachable from $p \ldots q$ at $e$ through $b$.
(In both cases, $\bot$ is to be returned when no such $i$-marked vertex exists.)

As previously, we assume that $T$ is represented by a~top tree $\Tc$, and the weights of the edges in the neighborhood data structures are defined with respect to $\Tc$.
Again, at the start of each query, we transiently expose $\{p, q\}$, causing $\Tc$ to be temporarily replaced with another top tree $\Uc$, but preserving the weights of the edges in the neighborhood data structures as defined with respect to~$\Tc$.

Both types of queries will be resolved in two phases.
First, we will determine whether the answer (a~sought $i$-marked vertex $y$) exists.
In the positive case, we will additionally determine the values of $a, b, c$ (in case of \FindFirstReach{}), as well as a~\emph{container} of $y$.
A~container will be a~nonempty set $Y$ of vertices with the property that every $y \in Y$ is a~correct answer to a~query at hand.
A~container will be of one of the following types:

\begin{itemize}
  \item an~\emph{explicit container} of the form $Y = \{y\}$;
  \item an~implicitly-defined \emph{cluster edge container}, defined by a~cluster edge $\vec{vw}$ with respect to $\Uc$, so that $p, q \notin T[\vec{vw}]$ and $Y$ is the set of $i$-marked vertices $y \in T[\vec{vw}]$ that are $i$-reachable from $vw$;
  \item an~implicitly-defined \emph{path cluster container}, defined by a~tuple $(C, v, w)$, where $C$ is a~path cluster of $\Uc$ with $\bnd C = \{v, w\}$ and $E(C) \cap E(p \ldots q) = \emptyset$, so that $Y$ is the set of $i$-marked vertices $y \in C$ such that $\meet(v, w, y) \notin \{v, w\}$ and $y$ is $i$-reachable from $\firstedge_{C, v}$.
\end{itemize}

In the first case, we can immediately resolve the query by returning $y$.
Otherwise, a~second phase is necessary in which we will recursively \emph{refine} the given container to produce an~example vertex $y \in Y$.
We delay the implementation of this refinement scheme to \Cref{ssec:implicit-refine}; there, we will prove the following statement.

\begin{lemma}
  \label{lem:container-refinement}
  Given an~implicitly-defined container $Y$, we can determine a~vertex $y \in Y$ in worst-case time $\BigO(\log n)$, plus queries to the neighborhood data structures of the form $\Level{}$, $\SelectedLevel{}$, $\OrMarks{}$, $\FindMarked{}$ of worst-case total normalized cost $\BigO(\log n)$.
\end{lemma}

Assuming \Cref{lem:container-refinement}, we will now prove the following statement.
Recall from \Cref{ssec:prelims-counters} that $B(\lmax)$ denotes the time required to perform an~operation on a~bit vector and that $B(\lmax) \leq T(\lmax)$.

\begin{lemma}
  \label{lem:tree-cover-findfirst}
  The data structure \textsc{CL} from \Cref{lem:ds-tree-cover-levels} can be extended to a~restricted tree cover level data structure \textsc{FR} additionally supporting $\Mark{}$, $\Unmark{}$, $\FindFirstReach{}$ and $\FindStrongReach{}$ so that:
  \begin{itemize}
    \item Each transient expose and unexpose takes additionally worst-case $\BigO(\log n \cdot B(\lmax) \log \lmax)$ time, plus calls to $\UpdateMarks{}$ and $\OrMarks{}$ in the neighborhood data structures of worst-case total normalized cost $\BigO(\log n)$.
    \item $\Mark{}$ and $\Unmark{}$ takes worst-case constant time, plus one transient expose and unexpose.
    \item $\FindFirstReach{}$ and $\FindStrongReach{}$ takes worst-case $\BigO(\log n)$ time, plus one transient expose and unexpose, and calls to $\Level{}$, $\OrMarks{}$ and $\FindMarked{}$ in the neighborhood data structures of worst-case total normalized cost $\BigO(\log n)$.
  \end{itemize}
\end{lemma}

The implementation of \FindFirstReach{} and \FindStrongReach{} phase will follow in the most part the blueprint of \FindSize{}, so in this description we will omit some proofs and implementation details that are immediate adaptations of their counterparts in \Cref{ssec:tree-structure-counting}.
We refer the reader to \Cref{ssec:tree-impl-find} for the pseudocode of the operations described in this section.

From now on we assume that the neighborhood data structures implement the marking extension.

\paragraph{Bit vectors for neighborhood data structures}
For every cluster edge $\vec{vw}$ with respect to $\Uc$ we store a~bit vector $\ismarkedvec_{\vec{vw}}$.
For $i \in \{0, \ldots, \lmax\}$, the value $\ismarked_{\vec{vw},i}$ specifies whether any $i$-marked vertex in $T[\vec{vw}]$ is $i$-reachable from $vw$.
We assume that $\ismarkedvec_{\vec{vw}} = {\bf 0}$ for interface edges $\vec{vw}$.
The vector $\ismarkedvec_{\vec{vw}}$ is the mark vector ${\bf b}^{vw}$ stored in the neighborhood data structure $N_v$, hence whenever $\ismarkedvec_{\vec{vw}}$ is updated, it is stored in $N_v$ via the call $N_v.\UpdateMarks{vw, \ismarkedvec_{\vec{vw}}}$.

In order to facilitate the implementation of \FindStrongReach{}, we need to store an~additional \emph{shifted} vector in the neighborhood data structure $N_v$.
Let $\ismarkedvec^+_{\vec{vw}}$ denote the shifted version of $\ismarkedvec_{\vec{vw}}$, where $\ismarked^+_{\vec{vw}, i+1} = \ismarked_{\vec{vw}, i}$ for $i \in \{0, \ldots, \lmax - 1\}$ and $\ismarked^+_{\vec{vw}, 0} = 0$.
Then we have a~copy $N^+_v$ of the neighborhood data structure $N_v$ which stores, for each edge $\vec{vw}$, $\ismarkedvec^+_{\vec{vw}}$ as the mark vector ${\bf b}^{+,vw}$.

We now adapt \Cref{lem:excl-sum-counters} to the setting of bit vectors, producing the following two statements:

\begin{lemma}
  \label{lem:excl-or-marks}
  Let $v \in V(T)$.
  Then, for any $A_+, A_- \subseteq \interface^\Uc_v$ we can determine a~bit vector ${\bf b}$ such that, for every $i \in \{0, \ldots, \lmax\}$, $b_i$ denotes whether there exists an~$i$-marked vertex in $\bigcup_{e \in \cluster^\Uc_v} T[e]$ that is $i$-reachable from any edge in $A_+$, but not $i$-reachable from any edge in $A_-$.
  The query takes worst-case time $\BigO(B(\lmax))$, plus queries to the neighborhood data structures of worst-case total normalized cost $\BigO\left(\sum_{e \in A_+ \cup A_-} \left(1 + \log \frac{\Weight_\Tc(v)}{\Weight_\Tc(e)} \right) \right)$.
\end{lemma}
\begin{proof}
  In the absence of the set $A_-$, the result of the query would simply be a~bitwise-or of bit vectors of the form ${\bf c}_e \coloneqq N_v.\OrMarks{e}$ for $e \in A_+$.
  However, if $A_-$ is nonempty, these bit vectors must be zero-truncated similarly to the proof of \Cref{lem:excl-sum-counters}: for each $e \in A_+$, let $p_e$ be the maximum cover level of $e$ with the edges of $A_-$, and ${\bf d}_e \coloneqq [{\bf 0} : (p_e + 1) : {\bf c}_e]$.
  Then, adapting the argument from the proof of \Cref{lem:excl-sum-counters}, we find that $d_{e,i}$ denotes whether there exists a~vertex in $\bigcup_{e \in \cluster^\Uc_v} T[e]$ that is $i$-reachable from $e$, but not $i$-reachable from any edge in $A_-$.
  Hence the result of the query is actually the bitwise-or of all vectors ${\bf d}_e$.
  It is straightforward to verify the time and normalized cost bounds of the implementation of the query.
\end{proof}

\begin{lemma}
  \label{lem:excl-find-marked}
  Let $v \in V(T)$.
  Then, for any $A \subseteq \interface^\Uc_v$ and $i \in \{0, \ldots, \lmax\}$, we can determine a~cluster edge $\vec{vx} \in \cluster^\Uc_v$ (respectively, $\vec{vx^+} \in \cluster^\Uc_v$) such that there exists an~$i$-marked vertex $y \in T[\vec{vx}]$ (resp., $y \in T[\vec{vx^+}]$) that is $i$-reachable (resp., strongly $i$-reachable) from some edge of $A$; we place $\vec{vx} = \bot$ (resp., $\vec{vx^+} = \bot$) if such an~edge does not exist.
  The query takes worst-case time $\BigO(1)$, plus queries to the neighborhood data structures of worst-case total normalized cost bounded by the sum of:
  \begin{itemize}
    \item $\BigO\left(\sum_{e \in A} \left(1 + \log \frac{\Weight_\Tc(v)}{\Weight_\Tc(e)} \right) \right)$; and
    \item $\BigO\left(\log \frac{\Weight_\Tc(v)}{\Weight_\Tc(\vec{vx})} \right)$ (resp., $\BigO\left(\log \frac{\Weight_\Tc(v)}{\Weight_\Tc(\vec{vx^+})} \right)$) if the query returns an~edge $\vec{vx}$ (resp., $\vec{vx^+}$).
  \end{itemize}
\end{lemma}
\begin{proof}
  Naturally, $\vec{vx}$ is found by iterating all $e \in A$ and querying $N_v.\FindMarked{e, i}$.
  If at least one of these operations returns an~edge as a~result, we forward it as the result of our query.
  Similarly, $\vec{vx^+}$ is determined by performing queries of the form $N^+_v.\FindMarked{e, i + 1}$.
  The entire process takes time $\BigO(1)$, plus queries to the neighborhood data structures of required total normalized cost.
\end{proof}

%\wninline{I think that absolutely all calls to $N^+_{.}.\FindMarked$ in the whole paper (pseudocode too) should have $i+1$ as the second argument rather than $i$. Am I right?}
%\msinline{original $\FindMarked$: yes. The invocations of \Cref{lem:excl-find-marked} should probably use $i$ to indicate strong $i$-reachability, so I renamed the function below from $\FindMarked$ to $\FindStrongMarked$. Does this make sense?}

We extend the interface of the neighborhood data structure $N_v$ so that $N_v.\OrMarks{A_+, A_-}$ returns the vector ${\bf b}$ described in the statement of \Cref{lem:excl-or-marks}.
Then $N_v.\FindMarked{A, i}$ denotes the operation of \Cref{lem:excl-find-marked} returning the cluster edge $\vec{vx}$, and $N^+_v.\FindStrongMarked{A, i}$ returns the cluster edge $\vec{vx^+}$.
Moreover, we define that $\BigO\left(\sum_{e \in A} \left(1 + \log \frac{\Weight_\Tc(v)}{\Weight_\Tc(e)} \right) \right)$ is the \emph{base cost} of a~call to $\FindMarked{}$ (resp., $\FindStrongMarked{}$), and the optional summand $\BigO\left(\log \frac{\Weight_\Tc(v)}{\Weight_\Tc(\vec{vx})} \right)$ (resp., $\BigO\left(\log \frac{\Weight_\Tc(v)}{\Weight_\Tc(\vec{vx^+})} \right)$) is the \emph{recovery cost} of the call.

\paragraph{Bit vectors maintained by the data structure}
For each vertex $v$ of the tree we maintain a~bit vector $\bmarksvec_v$ such that $\bmarks_{v,i}$ denotes whether the vertex $v$ is $i$-marked.

Also, for each (slim) path cluster $C$ with $\bnd C = \{v, w\}$, we store the following bit vectors.

\begin{itemize}
  \item $\totalmarksvec_C$, where $\totalmarks_{C, i}$ for $i \in \{0, \ldots, \lmax\}$ is true whenever there exists an~$i$-marked vertex $y \in C \setminus \{v, w\}$ that is $i$-reachable from $v \ldots w$;
  \item $\totalmarksvec_{C,v}$, where $\totalmarks_{C, v, i}$ is true whenever there exists $y \in C \setminus \{v, w\}$ that is $i$-reachable from $\firstedge_{C, v}$; and $\totalmarksvec_{C,w}$, defined analogously;
  \item $\diagmarksvec_{C,v,j}$ and $\diagmarksvec_{C,w,j}$ for $j \in \{0, \ldots, \lmax\}$ which are the bit vector analogs of $\diagcntvec_{C,v,j}$ and $\diagcntvec_{C,w,j}$.
\end{itemize}

These vectors can be computed analogously to their counter counterparts during both transient and permanent merges, and it can be easily verified that the time required to compute all these vectors during a~single merge of $C$ is bounded by $\BigO(B(\lmax) \cdot \log \lmax)$, plus calls to the neighborhood data structures of total normalized cost $\BigO(\ncost_\Tc(C))$.
Hence by the Transient Expose Lemma (\Cref{lem:transient-expose-lemma}), each transient expose and unexpose takes additionally time $\BigO(\log n \cdot B(\lmax) \log \lmax)$, plus calls to the neighborhood data structures of total normalized cost $\BigO(\log n)$ --- as required.

\paragraph{Updates via $\Mark{}$ and $\Unmark{}$}
Suppose a~vertex $u$ is to be $i$-marked or $i$-unmarked.
We begin by transiently exposing the vertex $u$, temporarily producing a~top tree $\Uc$.
Observe that after the expose, $u$ is not part of any subtree $T[\vec{vw}]$ for a~cluster edge $\vec{vw}$ with respect to $\Uc$; and moreover, the marks on $u$ do not influence any bit vectors or bit matrices stored together with path clusters of $\Uc$.
Hence marking or unmarking $u$ resolves to updating contents of $\bmarksvec_u$.
Afterwards, the vertex $u$ is transiently unexposed.
Obviously, this all takes constant time, plus a~single transient expose and unexpose.

\paragraph{Implementing \FindFirstReach{p, q, i}}
Following the blueprint laid down by $\FindSize{}$, we first perform a~transient expose on $\{p, q\}$, producing a~transient top tree $\Uc$ representing $T$.
$\Uc$ has three root clusters that partition $E(T)$: a~(slim) path cluster $C_{pq}$ with $\bnd C_{pq} = \{p, q\}$ and two point clusters $C_p$, $C_q$ such that $\bnd C_p = \{p\}$, $\bnd C_q = \{q\}$.

We first test if the sought $y$ is a~vertex of $C_p$, or equivalently, we test if $a = c = p$.
Observe that this is true precisely when either $p$ is $i$-marked (in which case we return an~explicit container $\{p\}$), or $f \coloneqq N_v.\FindMarked{\{\firstedge_{C_{pq}, p}\}, i} \neq \bot$ (in which case we return an~implicit cluster edge container defined by $f$).
In both positive cases, we have $ab = \firstedge_{C_{pq}, p}$ and $c = p$.

Otherwise, we check if a~sought $y$ belongs to $C_{pq}$, but is neither $p$ nor $q$ (equivalently, $c \neq q$ holds).
This is true if and only if $\totalmarks_{C,i}$ holds.
In this case, we will determine a~container for $y$ via the following recursive auxiliary function:

\smallskip

\emph{$\FindInPath{C, v, w}$: Let $C$ be a~(slim) path cluster such that $\pi(C) = v \ldots w \subseteq p \ldots q$.
Find a~tuple $(ab, c, Y)$ such that $ab \in \pi(C)$, $c \in \{a, b\} \setminus \bnd C$ and $Y$ is a~container for an $i$-marked vertex $y$ that is $i$-reachable from $\pi(C)$ at $ab$ through $c$.
In case of ties, minimize $\dist(a, v)$ first and then minimize $\dist(c, v)$.
We assume that the sought $y$ exists in $C$.}

\smallskip

We implement $\FindInPath{}$ as follows.
Since $C$ is a~(slim) path cluster (and not a~leaf of $\Uc$ since otherwise the sought $y$ cannot exist), it splits into two slim path clusters $A$, $B$ and a~point cluster $P$ according to case \CasePath of \Cref{fig:toptree-cases}.
Suppose that $\bnd A = \{v, x\}$, $\bnd B = \{x, w\}$ and $\bnd P = \{x\}$.
We perform the following steps in order.

\begin{enumerate}
    \item Test if the sought $y$ is in $A \setminus \bnd A$; this is true if and only if $\totalmarks_{A,i}$ holds.
    In this case, the answer to the query is $\FindInPath{A, v, x}$.
    
    \item Test if the sought $y$ is in $P$. This holds either if $x$ is $i$-marked (in which case we return $(ab, c, Y) = (\firstedge_{A,x}, x, \{x\})$), if $f_A \coloneqq N_x.\FindMarked{\firstedge_{A,x}} \neq \bot$ (and then we return $(\firstedge_{A,x}, x, f_A)$, where $f_A$ defines an~implicit cluster edge container), or if $f_B \coloneqq N_x.\FindMarked{\firstedge_{B,x}} \neq \bot$ (and then we return $(\firstedge_{B,x}, x, f_B)$).
    
    \item Test if the sought $y$ is in $B \setminus \bnd B$; this holds if and only if $\totalmarks_{B,i}$ holds.
    Then, the answer to the query is $\FindInPath{B, x, w}$.
    Note that this test must succeed due to the assumption that the sought $y$ belongs to $C$.
\end{enumerate}

Finally, if $y$ belongs to neither $C_p$ nor $C_{pq} \setminus \{p, q\}$, we verify if $y$ can be found in $C_q$.
This is done analogously to the case of $C_p$: either $q$ is $i$-marked (and we return $(\firstedge_{C_{pq},q}, q, \{q\})$), or $f \coloneqq N_v.\FindMarked{\firstedge_{C_{pq},q}, i} \neq \bot$ (and then we return $(\firstedge_{C_{pq},q}, q, f)$).

\paragraph{Time analysis of \FindFirstReach{}}
A~single transient expose is performed.
Aside from this, the time complexity of the query is dominated by recursive calls to $\FindInPath{}$.
Observe that $\FindInPath{}$ is invoked with slim path clusters $C$ lying on a~vertical path in $\Uc$.
Each recursive call performs calls to $N_x.\FindMarked{}$ a~constant number of times, with arguments from $\interface^\Uc_x$, where $x$ is the midpoint of $C$.
By \Cref{lem:excl-find-marked}, such calls have base cost $\BigO\left(\sum_{e \in A} \left(1 + \log \frac{\Weight_\Tc(v)}{\Weight_\Tc(e)} \right) \right)$, or equivalently $\BigO(\ncost_\Tc(C))$.
Moreover, the recursion is short-circuited as soon as any invocation of $\FindMarked{}$ returns an~edge.
In other words, across all recursive calls in total, at most one call to $\FindMarked{}$ incurs additional recovery cost (and this cost is obviously bounded by $\BigO(\log n)$).
Hence by the \Cref{item:tel-new-vertical-cost} of the Transient Expose Lemma (\Cref{lem:transient-expose-lemma}), the total normalized cost of calls to the neighborhood data structures is bounded by $\BigO(\log n)$, as required.
As announced, the second stage (refinement of implicitly-defined containers) is deferred to \Cref{ssec:implicit-refine}; but assuming it can be performed efficiently (\Cref{lem:container-refinement}), we conclude that the time requirements of \Cref{lem:tree-cover-findfirst} hold in case of $\FindFirstReach{}$.

\paragraph{Implementing $\FindStrongReach{p, q, e, b, i}$}
Firstly, we transiently expose the vertices $p, q$, yielding a~temporary top tree $\Uc$.
After the expose, the set of edges with tail $b$ that do not lie on the path $p \ldots q$ is precisely $\cluster^\Uc_b$.
Hence it is enough to figure out if there exists a~cluster edge $\vec{bx^+}$ with respect to $\Uc$ with tail $b$ such that there is some $i$-marked vertex $y \in T[\vec{bx^+}]$ that is strongly $i$-reachable from $e$.
This can be then done by calling $N_b^+.\FindStrongMarked{\{e\}, i}$, which has normalized cost at most $\BigO(\log n)$ by \Cref{lem:excl-find-marked}.
If the return value is $f \neq \bot$, we return $y$ in the form of an~implicit cluster edge container defined by $f$.
After the query, we transiently unexpose $p, q$.

Plainly, assuming \Cref{lem:container-refinement}, the entire process takes time $\BigO(\log n)$, plus a~single transient expose and unexpose.

\subsection{Refining implicit containers}
\label{ssec:implicit-refine}
If \FindFirstReach{} or \FindStrongReach{} returns an~implicit container $Y$ in lieu of $y$, we need to perform the second stage of the query: refine $Y$ to produce an~explicit vertex $y \in Y$ to return as the result of the query.
This refinement should be efficient so as to meet the time requirements of \Cref{lem:container-refinement}.
We shall implement the refinement through the following methods:

\begin{itemize}
  \item $\RefineClusterEdgeContainer{\vec{vw}, i}$: Refine a~cluster edge container defined by a~cluster edge $\vec{vw} \in \cluster^\Uc_v$;
  \item $\RefinePathClusterContainer{C, v, w, i}$: Refine a~path cluster container defined by $(C, v, w)$, where $C$ is a~slim path cluster in $\Uc$.
\end{itemize}

The implementation of these methods will be mutually recursive.
They will altogether implement a~scheme of recursive descent on the (transiently exposed) top tree $\Uc$, successively progressing through the clusters of the top tree lying on a~single root-to-leaf path in the top tree.
This will allow us to argue that each of these helper functions is efficient.

\smallskip

For $\RefineClusterEdgeContainer{\vec{vw}, i}$, let $C$ be the (unique) point cluster of the $\Uc$ such that $\bnd C = \{v\}$ and $vw$ is the unique edge of $C$ incident to $v$.
Then, $C$ splits into a~path cluster $A$ and a~point cluster $B$ according to case \CasePointToPathPoint in \Cref{fig:toptree-cases}, so that $\bnd A = \{v, x\}$ and $\bnd B = \{x\}$.
The result to the query is now found via the following case study.

\begin{itemize}
  \item If $\totalmarks_{A, v, i}$ is true, then a~sought vertex $y$ exists in $A$ and we can recursively return $\RefinePathClusterContainer{A, v, x, i}$.
  \item Otherwise, the path from $v$ to sought $y$ must pass through $x$. Therefore, it must be the case that $\cover_A \geq i$ (otherwise no such vertex $u$ can exist). Then, if $x$ is itself $i$-marked, we can immediately return $x$ as the result.
  \item Otherwise, there must exist some cluster edge $f \neq \firstedge_{A,x}$ with tail $x$ such that the sought $y$ is in $T[f]$.
  Since $\interface^\Uc_x = \{\firstedge_{A,x}\}$, such $f$ will be found by letting $f \coloneqq N_x.\FindMarked{\{\firstedge_{A,x}\}, i}$.
  Then, we return $\RefineClusterEdgeContainer{f, i}$ recursively.
\end{itemize}

\smallskip
For $\RefinePathClusterContainer{C, v, w, i}$, recall that $C$ splits into two path clusters $A$, $B$ and a~point cluster $P$ (case \CasePath in \Cref{fig:toptree-cases}).
Letting $\bnd A = \{v, x\}$, $\bnd B = \{x, w\}$, $\bnd P = \{x\}$, we have $\interface^\Uc_x = \{\firstedge_{A,x}, \firstedge_{B,x}\}$.
We perform another case study:

\begin{itemize}
  \item If $\totalmarks_{A, v, i}$ holds, some sought vertex $y$ exists in $A$; and moreover, this vertex satisfies $\meet(v, c, u) \notin \{v, c\}$. Hence we can return $\RefinePathClusterContainer{A, v, c, i}$.
  \item Otherwise, the path between $v$ and the sought $y$ must contain $x$, so it must be the case that $\cover_A \geq i$. Again, if $x$ itself is $i$-marked, we immediately return $x$ as the result.
  \item Otherwise, a~sought $y$ exists in $P$ if and only if $T$ contains a~cluster edge $f \in \cluster^\Uc_x$ such that $y$ is in $T[f]$; so the cover level of the pair $f, \firstedge_{A,x}$ must be at least $i$, and $y$ must be $i$-reachable from $f$.
  Such $f$ is found --- if it exists --- by letting $f \coloneqq N_x.\FindMarked{\{\firstedge_{A,x}\}, i}$.
  If $f \neq \bot$, we return $\RefineClusterEdgeContainer{f, i}$.
  \item Otherwise, a~sought $y$ must exist in $B$; in particular, the cover level of the pair $\firstedge_{A,x}, \firstedge_{B,x}$ must be at least $i$, and $\totalmarks_{B, x, i}$ must hold.
  This means that we can recursively return $\RefinePathClusterContainer{B, x, w, i}$.
\end{itemize}

The implementation of these functions can be found in \Cref{ssec:tree-impl-find}.

\paragraph{Time analysis}
We first observe that the recursion descends along clusters on a~single root-to-leaf path in $\Uc$.
This is straightforward: if $C, C'$ are two clusters of $\Uc$ pertaining to two consecutive recursive calls (either the slim path cluster itself in case of a~path cluster container, or the point cluster $T[\vec{vw}]$ in case of a~cluster edge container defined by a~cluster edge), then $C' \subsetneq C$ by a~simple case analysis.
Each of the two recursive functions performs a~constant number of inspections of bits in bit vectors, and a~constant number of calls of the form $N_x.\Level{}$ and $N_x.\FindMarked{}$ (where $x$ is the midpoint of the cluster $C$ pertaining to the recursive call).
Thus:
\begin{itemize}
  \item Excluding the calls to the neighborhood data structures, each recursive call takes time $\BigO(B(\lmax))$ per visited cluster.
  \item When a~cluster $C$ is examined, the calls to $N_x.\Level{}$ have normalized cost $\BigO(\ncost_\Tc(C))$, and the calls to $N_x.\FindMarked{}$ have base cost $\BigO(\ncost_\Tc(C))$.
  Hence by \Cref{item:tel-new-vertical-cost} of the Transient Expose Lemma (\Cref{lem:transient-expose-lemma}), the sum of these costs is $\BigO(\log n)$.
  \item It remains to bound the total recovery cost of calls to $N_x.\FindMarked{}$.
  Suppose we call $N_x.\FindMarked{}$ when considering a~cluster $C$, and the call returns a~cluster edge $\vec{xy}$.
  This call incurs additional recovery cost bounded by $\BigO(\log \frac{\Weight_\Tc(x)}{\Weight_\Tc(\vec{xy})})$, and immediately afterwards, we recursively return $\RefineClusterEdgeContainer{\vec{xy}, \cdot}$.
  So let $C'$ be the (unique) point cluster of $\Uc$ with $\bnd C' = \{x\}$, where $xy$ is the unique edge of $C'$ incident to $x$.
  Then $C'$ splits into child clusters according to case \CasePointToPathPoint and so, by definition, $\ncost_\Tc(C') \geq 1 + \log\frac{\Weight_\Tc(x)}{\Weight_\Tc(\firstedge_{C', x})} > \log \frac{\Weight_\Tc(x)}{\Weight_\Tc(\vec{xy})}$.
  Thus the recovery cost of the call is upper-bounded by $\ncost_\Tc(C')$.
  Therefore, since all clusters $C'$ lie on a~single vertical path in $\Uc$, \Cref{item:tel-new-vertical-cost} of the Transient Expose Lemma applies and the sum of recovery cost of calls to $N_x.\FindMarked{}$ across all recursive calls is also bounded by $\BigO(\log n)$.
\end{itemize}

Therefore the entire recursion takes time $\BigO(\log n \cdot B(\lmax))$, plus calls to the neighborhood data structures of total normalized cost $\BigO(\log n)$.
Consequently, our implementation meets the time complexity bounds postulated by \Cref{lem:container-refinement}.

\subsection{Proof of \Cref{lem:top-tree-reduction}}
\label{ssec:tree-structure-final}
It remains to conclude that \Cref{lem:top-tree-reduction} follows from \Cref{lem:ds-tree-cover-levels,lem:ds-tree-find-size,lem:tree-cover-findfirst}.
Indeed, consider the data structure of \Cref{lem:ds-tree-cover-levels} extended by \Cref{lem:ds-tree-find-size,lem:tree-cover-findfirst}.
\Cref{fig:ds-tree-time-analysis} lists, for every supported operation, the time complexity of the query, excluding transient (un)exposes and calls to the neighborhood data structures. 
Note that all supported operations except $\Connected{}$, $\Link{}$, $\Cut{}$, and $\Expose{}$, perform one transient expose and one transient unexpose as a~subroutine.
Thus \Cref{lem:top-tree-reduction} holds.
Moreover, \Cref{lem:top-tree-reduction,lem:neighborhood-ds} imply \Cref{lem:tree-structure}.

 \begin{figure}
   \begin{center}
   \footnotesize
   \begin{tabular}{c|c|c|c|c}
     \textbf{Operation} &
     \textbf{T} &
     \textbf{NC} &
     \textbf{Subqueries} &
     \textbf{Reference} \\ \hline
     {\footnotesize transient expose} & $\BigO(\log n \cdot T(\lmax) \log \lmax)$ & $\BigO(\log n)$ & \begin{tabular}{@{}c@{}} $\Level{}$ \\ $\LongZip{}$ \\ $\LongUnzip{}$ \\ $\UpdateCounters{}$ \\ $\SumCounters{}$ \\ $\UpdateMarks{}$ \\ $\OrMarks{}$ \end{tabular} & \Cref{lem:ds-tree-cover-levels,lem:ds-tree-find-size,lem:tree-cover-findfirst} \\ \hline
     $\Connected{}$ & $\BigO(\log n)$ & --- & --- & \Cref{lem:ds-tree-cover-levels} \\ \hline
     \begin{tabular}{@{}c@{}} $\Link{}$ \\ $\Cut{}$ \\ $\Select{}$ \end{tabular} & $\BigO(\log n \cdot T(\lmax) \log \lmax)$ & $\BigO(\log^2 n)$ & \begin{tabular}{@{}c@{}} $\Level{}$ \\ $\LongZip{}$ \\ $\LongUnzip{}$ \\ $\Select{}$ \\ $\SetWeight{}$ \end{tabular} & \Cref{lem:ds-tree-cover-levels} \\ \hline
     \begin{tabular}{@{}c@{}} $\Cover{}$ \\ $\UniformUncover{}$ \\ $\LocalUncover{}$ \\ $\CoverLevel{}$ \\ $\MinCoveredPair{}$ \end{tabular} & $\BigO(\log n)$ & $\BigO(\log n)$ & \begin{tabular}{@{}c@{}} $\Zip{}$ \\ $\Unzip{}$ \\ $\Level{}$ \\ transient expose \end{tabular} & \Cref{lem:ds-tree-cover-levels} \\ \hline
     $\FindSize{}$ & $\BigO(\log n)$ & $\BigO(\log n)$ & \begin{tabular}{@{}c@{}} $\SumCounters{}$ \\ transient expose \end{tabular} & \Cref{lem:ds-tree-find-size} \\ \hline
     \begin{tabular}{@{}c@{}} $\Mark{}$ \\ $\Unmark{}$ \end{tabular} & $\BigO(\log n)$ & --- & transient expose & \Cref{lem:tree-cover-findfirst} \\ \hline
     \begin{tabular}{@{}c@{}} $\FindFirstReach{}$ \\ $\FindStrongReach{}$ \end{tabular} & $\BigO(\log n)$ & $\BigO(\log n)$ & \begin{tabular}{@{}c@{}} $\Level{}$ \\ $\OrMarks{}$ \\ $\FindMarked{}$ \\ transient expose \end{tabular} & \Cref{lem:tree-cover-findfirst}
   \end{tabular}
   \end{center}

   \caption{Time complexity bounds of the combined data structure from \Cref{lem:ds-tree-cover-levels,lem:ds-tree-find-size,lem:tree-cover-findfirst}.
   Here, {\bf T} denotes the time complexity of the query, excluding calls to transient exposes and calls to the neighborhood data structures performed as subroutines.
   Next, {\bf NC} is the total normalized cost of the calls to the neighborhood data structures performed by the query, excluding the calls stemming from a~transient expose; and the {\bf Subqueries} column immediately right of it lists the types queries to these neighborhood data structures performed by the considered operations, and determines whether a~transient expose is performed as a~subroutine.}
   \label{fig:ds-tree-time-analysis}
 \end{figure}

\section{Neighborhood structure} \label{sec:nbd}
This section is dedicated to the proof of Lemma~\ref{lem:nbhtree}. Our proof is phrased as, first, a description of a data structure, followed, then, by the pseudocode for using and updating the data structure upon queries and updates. The pseudocode is annotated with running times of each line, most of which are straightforward to see from the description of the data structure. There are a few exceptions, where the running times need a brief argument; these are isolated in \cref{cla:parent,cla:nca,cla:lvla,cla:assign-dchildren}.

In general, we will analyse everything as if we had single values in the leaves. However, in the end, the leaves will hold counter and bit vectors of length $\lmax$, and thus, we will spend an extra factor of $T(\lmax)$ time on each operation, where $T(\lmax)$ is the time needed to perform operations on an $\lmax$-sized vector.

\paragraph{Description of data structure}
Let $\Xsel\subseteq X$ be the set of at most $2$ selected
edges. Then instead of maintaining the collection of nested sets
$\set{\mathcal L_i}_{i\in[\lmax]}$ explicitly, we maintain a slightly
modified collection of nested sets $\set{\mathcal L_i'}_{i\in[\lmax]}$
where
\begin{align*}
  \mathcal L_i(v)&=
  \begin{cases}
    \mathcal L_i'(v) &\text{if $i=-1$ or $\abs{\mathcal L_i(v)\cap \Xsel}\leq 1$}\\
    \mathcal \cup_{x\in \Xsel} \mathcal L_i'(x) &\text{otherwise}
  \end{cases}
\end{align*}
In other words, for $i\geq 0$ every set in $\mathcal L_i$ that contains both selected edges is the disjoint union of $2$ sets in $\mathcal L'_i$.

We represent the neighborhood data structure as a rooted tree, where
the leaves are the elements of the set $X$, and all leaves have depth
$\lmax+1$.  Each internal node at depth $i+1$ corresponds to a set in
$\mathcal L_i'$. In particular the root corresponds to the set $X$
where $\mathcal L_{-1}'=\set{X}$.  Together with a variable
representing $\SelectedLevel{}$ this is sufficient to reconstruct
$\set{\mathcal L_i}_{i\in[\lmax]}$. The main benefit of representing
the pair $(\mathcal L',\SelectedLevel{})$ instead of $\mathcal L$
directly, is that it makes $\LongZip$ and $\LongUnzip$ run in worst-case
constant time since their preconditions ensure that only
$\SelectedLevel{}$ needs to be changed.

Let $\set{s_1,s_2}=\Xsel$ be the selected edges. We call a
node \emph{$\Xsel$-marked} if it is an ancestor of $s_1$ or $s_2$,
and \emph{$\Xsel$-unmarked} otherwise; we will explicitly store in
each node which (if any) of $s_1,s_2$ it is ancestor to.  Note that this is unrelated to the ``marking'' extension that we will cover later in \Cref{sec:nt-extensions}. Also observe that we have introduced an arbitrary ordering between $s_1$ and $s_2$, which we use for disambiguation below.

\paragraph{Solid-path decomposition}
%% \jhinline{I have added the new headline here, and rewritten the following to make it clear that we distinguish between heavy/light edges, which are implicitly defined by the weights, and solid/dashed edges which we explicitly maintain ourselves. }
%% \wninline{That's great. At the end of my work on this a week ago I had a feeling that the initial version speaking about both heavy and solid was a conscious decision rather than a confusion, so it's good to have the confirmation that it is the case and have it written explicitly. I think that some places in the text still mention heavy where it should be solid though.}
To get the remaining operations to have the desired cost, we use a
version of \emph{heavy-light decomposition} similar to~\cite{SleatorT81,HolmR20}.  Define the weight of the
subtree $T_v$ rooted at a node $v$ representing some set $S\in\mathcal
L_i'$ as $w(T_v)=\sum_{s\in S}w(s)$.  Call an edge $(c,p)$ between a
child $c$ and a parent $p$ \emph{heavy} if either 1) $c$ is
$\Xsel$-marked and an ancestor of $s_1$; or 2) $c$ is $\Xsel$-marked and
$p$ is not an ancestor of $s_1$; or 3) $p$ is not $\Xsel$-marked and the
total weight $w(T_p)<2 w(T_c)$. If $(c,p)$ is heavy we will say that $c$ is a \emph{heavy child}. Edges that are not heavy are
called \emph{light} and if $(c, p)$ is a light edge, then we say that $c$ is a \emph{light child}. Since each node has at most one heavy
child, the heavy edges form a system of \emph{heavy paths}.

Suppose we mark each edge in the tree as either \emph{solid} or \emph{dashed}, in such a way that for every node at most one of its child edges is solid. Then the solid edges form a system of \emph{solid paths}. We will maintain the invariant that between operations on the neighborhood structure, an edge is solid if and only if it is heavy. However, this invariant may temporarily be broken during operations.

We can represent each solid path using a balanced binary search tree
(ignoring the weights). Since the height of the original tree is
$\lmax$, the height of this binary search tree is $\BigO(\log\lmax)$,
and we can safely assume that concatenating or splitting solid paths
only touches a prefix of this rooted tree of size $\BigO(\log\lmax)$. These trees will help us in querying about various information, for example the highest vertex of a solid path (to traverse the tree efficiently), weights of subtrees (to efficiently determine heavy/light children and maintain the heavy paths structure) or the deepest vertex of a solid path with some prescribed property.\jhinline{E.g. the data needed to do \textproc{FindMarked} correctly.} %We remark that for efficiency reasons, we sometimes let heavy edges be temporarily not complying with how we defined them --- the subtle reasons will be made clear later on.

Each node can have many dashed children, and we store the set of these children using the \emph{almost biased binary tree} implemented by
the \emph{biased disjoint sets} structure of Section~\ref{sec:bds} (Theorem~\ref{thm:bds}). The root of this set for vertex $v$ will be called $v.\textrm{dchildren}$. After each external query to the neighborhood data structure, the set represented by $v.\textrm{dchildren}$ will correspond exactly to the set of light children of $v$, however in the intermediate states of computations, heavy children may at times be a part of this set too. That includes children that are ancestors of selected edges and it is easy to modify the structure to handle the at most two selected edges separately, giving them constant depth.
Otherwise, for a dashed child $c$ of a node $p$ the depth of $c$ in this structure
is $\BigO(\log\frac{w(T_p)}{w(T_c)}+\log\log w(T_p))$. By Theorem~\ref{thm:bds}, this depth is both the worst-case bound on the time required to follow the parent pointers from the leaf of an~almost biased binary tree corresponding to a dashed child (\Find{}), and the amortized cost for deleting a dashed child from the biased disjoint sets structure (\Delete{}). Note that if $c$ is a heavy child of $p$, but is not an ancestor of a selected edge, then $\log\frac{w(T_p)}{w(T_c)}$ is a constant and its depth is $\BigO(\log\log w(T_p))$.

The \emph{light depth} of a node $v$ is the number of light edges on
the path from $v$ to the root $r$, and is easily seen to be at most
$\floor[\big]{\log_2\frac{w(T_r)}{w(T_v)}}+2$.
Thus, the path from $v$ to $r$ intersects at most $\floor[\big]{\log_2\frac{w(T_r)}{w(T_v)}}+3$ heavy paths.

We will ensure that the similarly-defined \emph{dashed depth} is always $\BigO\paren[\big]{1+\log\frac{w(T_r)}{w(T_v)}}$, thus the  path from $v$ to $r$ intersects at most $\BigO\paren[\big]{1+\log\frac{w(T_r)}{w(T_v)}}$ solid paths, and the sum of the depths in the dashed-child structures for $v$ telescopes to
$\BigO((1+\log\frac{w(T_r)}{w(T_v)})\cdot\log\log w(T_r))$.

The total depth of a~node $v$ -- understood as the sum of the depths of the parts representing a subtree containing $v$ in all dashed-child and solid-path structures -- is therefore $\BigO((1+\log\frac{w(T_r)}{w(T_v)})\cdot(\log\log w(T_r)+\log\lmax))$.

%% \msinline{\footnotesize I added this paragraph below. Low prio: someone verify these operations do what I say.}
%% \wninline{And I changed it very much}
%% \wninline{High prio: I think Link/Cut should NOT modify the system of heavy paths and we need to perform the required modifications more manually}
%% \jhinline{In~\cite{HolmR20} we defined operations \textproc{link-exposed} and \textproc{cut-exposed} that do exactly what you suggest, with the additional requirement that the parent vertex of the new/deleted edge must be exposed before the operation (and stays exposed). This is the version of \textproc{Cut} I had in mind when writing this pseudocode, and explains the call to \textproc{Expose} in Algorithm~\ref{alg:nt-unzip}.  I also assumed that \textproc{Link} and \textproc{Cut} updated all auxillary information in the single solid path tree that gets affected, which means that \textproc{FullPathUpdate} should either not be necessary, or if we want it to be an explicit extra step it should be called \textproc{SolidPathUpdateAux} or similar.
%% }
The tree structure will be accessed and updated via several operations:
\begin{itemize}
  \item $\Call{Parent}{x}$ returns the parent of $x$ in the tree;
  %% \item $\Call{GetWeight}{x}$ returns $w(T_x)$ and information about
  %%   whether $x$ is $\Xsel$-marked.\jhinline{Do we actually need
  %%     $w(T_x)$, and in what time? Or do we really need
  %%     \begin{align*}
  %%       \hat{w}(T_x)=
  %%       \begin{cases}
  %%         w(T_x)-w(T_s)&\text{if }x\text{ has solid child }s\\
  %%         w(T_x)&\text{otherwise}
  %%       \end{cases}
  %%     \end{align*}
  %%     which is what we explicitly store. }

  \item $\Call{UnselectHeavy}{x}$ assumes the solid child (if any) is heavy and makes it dashed. Similarly $\Call{SelectHeavy}{x}$ assumes that $x$ has no solid child edge and makes heavy child of $x$ (if any) solid. Additionally, they pass all incurred updates to the biased disjoint sets data structure and update heavy path trees.  Note that these are similar to, and implemented in terms of, the \textproc{slice} and \textproc{splice} operations from~\cite{SleatorT81,HolmR20}.

  \item $\Call{Expose}{x}$ and $\Call{Conceal}{r}$ modifies the system
    of solid paths.  We follow the terminology from~\cite{HolmR20} and
    say that a tree where edges are solid if and only if they are
    heavy is called \emph{proper}.  A tree where there is a solid path
    with the root of the tree as its topmost vertex and $x$ as its
    bottommost vertex, and where every node not on this path has its
    solid child set to its heavy child (if any), is called
    \emph{$x$-exposed} or just \emph{exposed}.  $\Call{Expose}{x}$
    takes a proper tree containing node $x$ and makes it $x$-exposed
    and returns the root $r$.
    $\Call{Conceal}{r}$ takes an exposed tree with given root $r$ and
    makes it proper.

    It follows from~\cite{HolmR20} combined with our data structures
    for dashed children and solid paths that $\Call{Expose}{x}$ and
    $\Call{Conceal}{r}$ both take $\BigO( (1+\log\frac{W}{w(T_x)})\cdot
    (\log\log W+\log\lmax))$ time.

  \item $\Call{Link-Exposed}{x, p}$, where $x$ is the root of a proper
    tree and $p$ is a node in a (different) $p$-exposed tree, adds $x$
    as a~child of $p$, leaving the resulting tree
    $p$-exposed. Conversely $\Call{Cut-Exposed}{x}$, where $x$ has
    parent $p$ and the tree is $p$-exposed, unlinks $x$ from its
    parent $p$ in the tree, leaving the tree containing $p$ as
    $p$-exposed and the new tree with root $x$ proper. Since these
    operations only update the dashed children of a single node with
    no solid child, and on the auxillary information on the exposed
    root path they can be implemented to take
    $\BigO(1+\log\frac{\max\set{w(T_p),w(T_p')}}{w(T_x)}+\log\log W+\log\lmax)$ time.

  \item $\Call{Link}{x, p}$ adds $x$ as a~child of $p$ and $\Call{Cut}{x}$ unlinks $x$ from its parent in the tree.
    $\Call{Link}{x,p}$ can be trivially implemented as $r\gets\Call{Expose}{p}$, $\Call{Link-Exposed}{x,p}$ and $\Call{Conceal}{r}$ in $\BigO( (1+\log\frac{W}{w(T_x)} + \log\frac{W}{w(T_p)})\cdot (\log\log W + \log\lmax) )$ time.  Similarly $\Call{Cut}{x}$ can be implemented as $p\gets\Call{parent}{x}$, $r\gets\Call{Expose}{p}$, $\Call{Cut-Exposed}{x}$, $\Call{Conceal}{r}$, also in 
    $\BigO( (1+\log\frac{W}{w(T_x)}+\log\frac{W}{w(T'_p)})\cdot (\log\log W + \log\lmax) )$ time.

    Note that if $p$ is the root of its tree, or is $\Xsel$-marked,
    both operations take just $\BigO(\log\frac{W}{w(T_x)} + \log\log W
    + \log\lmax)$ time, and if furthermore $p$ has no other children
    it takes just $\BigO(\log\lmax)$ time.

%%     If we perform $\Call{Cut}{x}$ where $x$ was a heavy child, then we cut the corresponding heavy path tree in two parts, but we remark that we do \emph{not} perform any other changes to the heavy path structure as a result of any $\Call{Link}{}$ or $\Call{Cut}{}$ calls, possibly leading to having temporarily outdated heavy path structure. That will be manually fixed by calling the following operations on a carefully chosen subset of vertices;
%% %  (both of these operations may permanently update the system of heavy paths in the tree, in particular requiring updates to the relevant biased disjoint sets data structure);
%% \jhinline{I think \textproc{Link} and \textproc{Cut} \emph{should} update the solid paths, but we also need \textproc{link-exposed} and \textproc{cut-exposed} as defined in~\cite{HolmR20}, which don't but require the parent node of the affected edge to be exposed.  I'll have to go over the pseudocode to determine which version is correct in each case.}

\wninline{%%I guess my main point was that the previous version of pseudocodes treated all $\Call{Link}{}$ calls the same and there was a need to distinguish them into calls that do update and that don't update solid paths. How we end up naming them is a secondary issue, but I like the version with more expensive $\Call{Link}{}$ and cheap $\Call{Link-Exposed}{}$ as you suggest.

  Just to make sure, I'll remark that in my opinion the final complexities of $\Call{Link}{x}$ and $\Call{Cut}{x}$ will end up being $\widetilde{O}(\log n)$ --- no biased logs here. Or if one would really want, you could express it as $\log \frac{W}{w(T_{p(x)})}$ in the version of the structure that does not have $x$ (i.e. before $\Call{Link}{x}$ or after $\Call{Cut}{x}$), but that's not very useful I think.}
\jhinline{I have added what I think is the correct times above. They are essentially as you say.  I thought we had a problem in \Call{ZipInto}{}, but as you commented elsewhere we have the parents $\Xsel$-marked so that is actually cheap.}

%%   \item $\Call{FullPathUpdate}{x}$ for each ancestor $v$ of $x$ (in the bottom-up order, starting from $x$), call $\Call{UnselectHeavy}{v}$, update the weight of $v$ and call $\Call{SelectHeavy}{v}$;
%% %  \wninline{Maybe call it $\Call{TransientExpose}$ to reflect the similarity?}
%%   \jhinline{This operation should not be necessary. What it maintains is more expensive than we need, and what we need should already be maintained internally by \textproc{Link}, \textproc{Cut}, \textproc{Link-Exposed}, \textproc{Cut-Exposed}, \textproc{Expose}, and \textproc{Conceal} when they update the solid-path structures.}
%%   \wninline{I am pretty sure that in the worst case the $\Call{FullPathUpdate}{}$ is as expensive as $\Call{Link}{}$ and $\Call{Cut}{}$, which is $\widetilde{O}(\log n)$. But yes, depending on the final version of $\Call{Link}{}$ etc. it may turn out to not be necessary in the end.}

\end{itemize}

\paragraph{Annotated pseudocode} 
In the pseudocode below, we explain how we obtain the running times stated in Lemma~\ref{lem:nbhtree}, by annotating each line of the pseudocode with its running time. Most of these are relatively straightforward to see, but there are a few which require a short argument, and those, we now outline as claims.

\begin{claim}\label{cla:parent}
	When $\Call{Parent}{x}$ returns $y$, it takes time $\BigO(\log \frac{w(y)}{w(x)} + \log\log W)$. 
\end{claim}
\begin{proof}
  If $x$ is a heavy child, it has a direct pointer to its
  parent $y$, and $\Call{Parent}{x}$ takes $\BigO(1)$ time. Otherwise
  $x$ is a light child of $y$, and stored as an element in a biased
  disjoint sets tree whose root has a direct pointer to $y$. So in
  this case $y$ is the pointer stored in the node returned by
  $\Call{Find}{x}$, and the running time follows directly from the
  running time for \textproc{Find} in Theorem~\ref{thm:bds}.
\end{proof}

\begin{claim}\label{cla:nca}
	The query to the nearest common ancestor of $x$ and $y$, which we write as $x\bot y$, takes time $\BigO((1+f(x)+f(y))\cdot (\log\log W + \log\lmax) )$ 
	where $f(z)$ is $1$ if $z\in \Xsel$, and $f(z) = \log\frac{W}{w(z)}$ otherwise.
\end{claim}
\begin{proof}
We use $\Expose{y}$ to make $y$ the end of a heavy root path in $\BigO((1+f(y))\cdot(\log\log W+\log\lmax))$ amortized time. Then traverse up the $\BigO(1+f(x))$ light edges (via $\Call{Parent}{}$) and heavy paths (by traversing the binary search trees storing the paths) until hitting the path containing $y$. The total cost of this traversal is $\BigO((1+f(x))\cdot(\log\log W+\log\lmax))$.
\end{proof}

\begin{claim}\label{cla:lvla}
	Level ancestor of a selected node takes time $\BigO(\log \lmax)$. For any other node $x$, level ancestor takes time
	$\BigO((1+\log \frac{W}{w(x)})\cdot (\log \lmax + \log\log W))$.
\end{claim}
\begin{proof}
  Use $\Expose{x}$, and find the $i$th element in the resulting heavy path. 
  If $x$ is a selected node, $\Expose{x}$ takes $\BigO(\log \lmax)$ time. Otherwise, it may take $\BigO((1+\log \frac{W}{w(x)})\cdot (\log \lmax + \log\log W))$ time. Finding the $i$th node takes $\BigO(\log\lmax)$ time. 
\end{proof}

\begin{claim}\label{cla:assign-dchildren}
  Replacing the set of light children of a node $x$ with a new specified set of light children of the same weight takes $\BigO(1)$
  time. 
  If their weights differ, and $x$ is an ancestor of a selected node, it takes time $\BigO(\log\lmax)$.
  
  In the special case where $x$ is an isolated node, it again takes $\BigO(1)$ time.
  
  %  if the new weight is the same as the old, and 
%  time if $x$ is ancestor to one of the selected nodes   
%  , and otherwise
%  $\BigO((1+\log\frac{W}{w(x)})\cdot(\log\log W+\log\lmax))$ time
%  where $W$ is the maximum of the old and new total weights\jh{We never directly use the last case.}.
\end{claim}
\begin{proof}
	If there are no changes to the weights, this is local operation, changing only a constant number of pointers.
	
	If $x$ is an ancestor of a selected node, it has constant light depth, and it has constant depth in the at most one light structure to which it belongs. Thus, the cost is incurred by the at most two heavy paths, each yielding at most an additive $\BigO(\log \lmax)$ term. 

	Finally, if $x$ is an isolated node, there is nothing else to update, and it takes constant time.
\end{proof}

\begin{claim}\label{cla:unselect-select-heavy}
	The procedures $\Call{UnselectHeavy}{x}$ and $\Call{SelectHeavy}{x}$ work in the time complexity $\BigO(\log \log W + \log \lmax)$.
\end{claim}
\begin{proof}
	If $y$ is a heavy child of $x$ then either $y$ is an ancestor of a selected node or $2w(T_y) > w(T_x)$. In any case, as discussed earlier, the depth of it in the respective biased disjoint sets data structures is $\BigO(\log \log W)$, hence these procedures work in the time complexity $\BigO(\log \log W + \log \lmax)$.
\end{proof}

Note, finally, that the analysis annotated in the pseudocode counts the unitary operations in the data structure (i.e., it assumes that the leaves hold single values instead of bit vectors and counter vectors). To obtain the final analysis of the actual running time, one in fact has to take into account that we not only perform operations on one number, but on a vector of size $\lmax$. In other words, all the stated running times should be multiplied by a factor $T(\lmax)$, denoting the time it takes to perform operations on a vector of size $\lmax$. 

%\erinline{T(lmax) is the time it takes to perform operations on a vector of size lmax ... all times are multiplied by this? we do not include this in this section. Maybe mention it as a btw.}

%\jhinline{Lemma with running times for \textproc{Parent}, \textproc{LevelAncestor}, \textproc{NCA}, and for attaching/detaching a child?}
%\jhinline{Cite~\cite{HolmR20}[Lemma~16].  We can not use~\cite{HolmR20}[Lemma~20] directly because our weights are not $k$-positive.}
%\jhinline{Actually, they \emph{are} $k$-positive for $k=\lmax+2$ so we may be able to use them after all.  However, the proof in~\cite{HolmR20} is not really spelled out either and I am not sure I trust it in this context.}

\begin{algorithm}
  \begin{algorithmic}
    \caption{Insert} \label{alg:nt-insert}
    \Function{Insert}{$x$}
    \State $\Call{InsertAt}{x,\lmax}$
    \Comment{$\BigO(\log\frac{W}{w(x)} + \lmax \cdot (\log\log W + \cdot \log\lmax))$}
    \EndFunction
    \State
    \Function{InsertAt}{$x, \ell$}
    %% \State $x' \gets x$
    \For{$i=\ell$ \textbf{downto} $0$}
      \State Create a new node $p$ at level $i$
      \Comment{$\BigO(1)$}
      \State $\Call{Link}{x,p}$
      \Comment{$\BigO(\log\lmax)$}
      \State $x\gets p$
      \Comment{$\BigO(1)$}
    \EndFor
    \State $\Call{Link}{x,\mathrm{root}}$
    \Comment{$\BigO(\log\frac{W}{w(x)} + \log\log W + \log\lmax)$}
    %% \State $\Call{FullPathUpdate}{x'}$
    %% \Comment{$\BigO(l \cdot (\log \log W + \log \lmax))$}
    \EndFunction
  \end{algorithmic}
\end{algorithm}

\begin{algorithm}
  \begin{algorithmic}
    \caption{Delete} \label{alg:nt-delete}
    \Function{Delete}{$x$}
        \State $p \gets \Call{Parent}{x}$
        \Comment{$\BigO(\log\frac{w(p)}{w(x)} + \log\log W)$}
        \While {$p$ is not the root and $p$ has only one child}
            \State $x \gets p$
            \Comment{$\BigO(1)$}
            \State $p \gets \Call{Parent}{x}$
            \Comment{$\BigO(\log\frac{w(p)}{w(x)} + \log\log W)$}
        \EndWhile
        \State $\Call{Cut}{x}$
        %\Comment{$\BigO(\log\frac{w(p)}{w(x)} + \log\log W + \log\lmax)$}
        \Comment{$\BigO((1+\log\frac{W}{w(x)}+\log\frac{W}{w'(p)})\cdot(\log\log W+\log\lmax))$}
        \LComment{Now destroy the tree rooted at $x$.}
        %% \State $\Call{FullPathUpdate}{p}$
        %% \Comment{$\BigO(\lmax \cdot (\log \log W + \log \lmax))$}
%      \If{$x$ is not the root}
%        \State $p\gets\Call{Parent}{x}$
%        \Comment{$\BigO(\log\frac{w(p)}{w(x)} + \log\log W)$}
%        \If{$p$ has at most one child}
%         \State \Call{Delete}{$p$}
%         \Comment{$\BigO(\log\frac{W}{w(p)} + \Call{depth}{p}\cdot (\log\log W + \log\lmax))$}
%        \EndIf
%        \State \Call{Cut}{$x$}
%        \Comment{$\BigO(\log\frac{w(p)}{w(x)} + \log\log W + \log\lmax)$}
%      \EndIf
    \EndFunction
  \end{algorithmic}
\end{algorithm}

\begin{algorithm}
  \begin{algorithmic}
    \caption{Level} \label{alg:nt-level}
    \Function{Level}{$x, y$}
        \Comment{$\bot$ denotes nca below}
%    \If{$x\bot y$ not root}
%    \State return $x\bot y$ 
%    \EndIf
        \State $L \gets $ selected level
        \State \Return
    $\max\{x\bot y, \min\{L, \max \{x\bot s_1, x\bot s_2\},
    \max \{y\bot s_1, y\bot s_2\}\}\}$
%    \State return $\max \{x\bot y, \min \{x\bot s_1,y\bot s_2\}, \min \{y\bot s_1, x\bot s_2\}\}$
	\LComment{Note, $x\bot y$ dominates the running time: \hfill $\BigO((1+\log\frac{W}{w(x)}+\log\frac{W}{w(y)})\cdot (\log\log W + \log\lmax) )$}

%\jhinline{Here $x\bot y$ takes $\BigO((1+\log\frac{W}{w(x)}+\log\frac{W}{w(y)})\cdot (\log\log W + \log\lmax) )$ time, dominating each $z\bot s_i$ which take $\BigO((1+\log\frac{W}{w(z)})\cdot (\log\log W + \log\lmax) )$ time.}
    \EndFunction
  \end{algorithmic}
\end{algorithm}

\begin{algorithm}
  \begin{algorithmic}
    \caption{Zip} \label{alg:nt-zip}
    \Function{Zip} {$x, y, i$}
    \State $cx\gets$ ancestor of $x$ at level $i$
    \Comment{$\BigO((1+\log \frac{W}{w(x)})\cdot (\log \lmax + \log\log W))$}
    \State $cy\gets$ ancestor of $y$ at level $i$
    \Comment{$\BigO((1+\log \frac{W}{w(y)})\cdot (\log \lmax + \log\log W))$}
    \LComment{Since $x \nsim_i y$, we have $\mathcal{L}_i(x) \neq \mathcal{L}_i(y)$ and $cx \neq cy$.
    Also $x \sim_{i-1} y$ and therefore $\mathcal{L}_{i-1}(x) = \mathcal{L}_{i-1}(y)$.
    Hence either $\mathcal{L}'_{i-1}(x) = \mathcal{L}'_{i-1}(y)$ and $\Call{Parent}{cx} = \Call{Parent}{cy}$, or $\mathcal{L}_{i-1}(x)$ is a~disjoint union of $\mathcal{L}'_{i-1}(x)$ and $\mathcal{L}'_{i-1}(y)$, and $\Call{Parent}{cx} \neq \Call{Parent}{cy}$.}
    \If{$cx$ is not ancestor of a selected node}
      \State \Call{ZipInto}{$cx, cy$}
      \Comment{$\BigO((1 + \log \frac{W}{w(cx)} + \log \frac{W}{w(cy)}) \cdot (\log\log W + \log \lmax) )$}
    \ElsIf{$cy$ is not ancestor of a selected node}
      \State \Call{ZipInto}{$cy, cx$}
          \Comment{$\BigO((1 + \log \frac{W}{w(cx)} + \log \frac{W}{w(cy)}) \cdot (\log\log W + \log \lmax) )$}
    \Else{} \Comment{Both $cx$ and $cy$ are ancestors of the selected nodes}
      \State Increase selected level from $i-1$ to $i$
      \Comment{$\BigO(1)$}
    \EndIf
    \EndFunction
    \State
    \Function{ZipInto} {$cx, cy$}
    \State $px\gets\Call{parent}{cx}$
    \Comment{$\BigO(\log \frac{w(px)}{w(cx)} + \log\log W)$}
    \State $py\gets\Call{parent}{cy}$
    \Comment{$\BigO(\log \frac{w(py)}{w(cy)} + \log\log W)$}
    \If{$px\neq py$}
        \LComment{$p_x$ and $p_y$ are both $\Xsel$-marked, so the cut and link here are cheap.}
        \State $\Call{Cut}{cx}$
        \Comment{$\BigO(\log \frac{W}{w(cx)} + \log\log W + \log\lmax )$}
        \State $\Call{Link}{cx,py}$
        \Comment{$\BigO(\log \frac{W}{w(cx)} + \log\log W + \log\lmax )$}
        %% \LComment{In this case both $px$ and $py$ are $\Xsel$-marked, hence there is no need to call any $\Call{FullPathUpdate}{}$}
    \EndIf
    \State $\Call{UnselectHeavy}{cx}$
    \Comment{$\BigO(\log\log W + \log\lmax)$}
    \State $\Call{UnselectHeavy}{cy}$
    \Comment{$\BigO(\log\log W + \log\lmax)$}
    \State $\Call{UnselectHeavy}{py}$    
    \Comment{$\BigO(\log\log W + \log\lmax)$}
%    \State Unselect the heavy children of $cx, cy$ (if any) 
    \State Create a new node $c$ with weight $w(cx)+w(cy)$
    \Comment{$\BigO(1)$}
    \State $c.\textrm{dchildren}\gets\Call{RootUnion}{cx.\textrm{dchildren},cy.\textrm{dchildren}}$
    \Comment{$\BigO(\log\log W)$}
    \State $py.\textrm{dchildren}\gets \Call{Coalesce}{cx,cy,c}$
    \Comment{$\BigO(\log \frac{w(py)}{w(cx)} +\log \frac{w(py)}{w(cy)}+ \log\log W)$}
    \State $\Call{SelectHeavy}{c}$
    \Comment{$\BigO(\log\log W + \log\lmax)$}
    \State $\Call{SelectHeavy}{py}$
    \Comment{$\BigO(\log\log W + \log\lmax)$}
%    \State Select new heavy children of $c,p$ (if any)
%        \Comment{$\BigO(\log\log W + \log\lmax)$}
    \EndFunction
  \end{algorithmic}
\end{algorithm}

\begin{algorithm}
	\begin{algorithmic}
		\caption{LongZip} \label{alg:nt-longzip}
		\Function{LongZip}{$i_1, i_2$}
		\State $\SelectedLevel{}\gets i_2$
		\EndFunction
	\end{algorithmic}
\end{algorithm}

\begin{algorithm}
  \begin{algorithmic}
    \caption{Unzip} \label{alg:nt-unzip}
    \Function{Unzip}{$x, y, i$}
    \State $cx \gets $ ancestor of $x$ at level $i + 1$ 
    \Comment{$\BigO((1+\log \frac{W}{w(x)})\cdot (\log \lmax + \log\log W))$}
    \State $cy \gets $ ancestor of $y$ at level $i + 1$
    \Comment{$\BigO((1+\log \frac{W}{w(y)})\cdot (\log \lmax + \log\log W))$}
    \State $px\gets\Call{parent}{cx}$
    \Comment{$\BigO(\log \frac{w(px)}{w(cx)} + \log\log W)$}
    \State $py\gets\Call{parent}{cy}$
    \Comment{$\BigO(\log \frac{w(py)}{w(cy)} + \log\log W)$}
    \State $gx\gets\Call{parent}{px}$
    \Comment{$\BigO(\log \frac{w(gx)}{w(px)} + \log\log W)$}
    \State $gy\gets\Call{parent}{py}$
    \Comment{$\BigO(\log \frac{w(gy)}{w(py)} + \log\log W)$}
    \If{$cx$ is ancestor of some selected node}

        \If{$\SelectedLevel{} \geq i + 1$} 
        \Comment{The edge $cx,px$ is a double edge}
            \For{$j \in \{1,2\}$}
                \State $p_j\gets $ ancestor of $s_j$ at level $i$
                \Comment{$\BigO(\log \lmax)$}
            \EndFor
            \State Create a new node $p'$
            \State $p'.\mathrm{dchildren}\gets\Call{RootUnion}{p_1.\mathrm{dchildren},p_2.\mathrm{dchildren}}$
            \Comment{$\BigO(\log\log W)$}
            \State $p_1.\mathrm{dchildren}\gets \emptyset$
            \Comment{$\BigO(\log\lmax)$}
            \State $p_2.\mathrm{dchildren}\gets \emptyset$
            \Comment{$\BigO(\log\lmax)$}
            \State $\Call{Link}{p',gy}$
            \Comment{$\BigO(\log\frac{w(gy)}{w(p')} + \log\log W + \log\lmax)$ (since $gy$ is $\Xsel$-marked)}
            %% \LComment{In this case $gy$ is $\Xsel$-marked, hence there is no need to call $\Call{FullPathUpdate}{}$}
            
            %% \For{$j \in \{1,2\}$}
            %%     \State $c_j\gets $ ancestor of $S_j$ at level $i$
            %%     \Comment{$\BigO(\log \lmax)$}
            %%     \State $p_j\gets \Call{parent}{c_j}$
            %%     \Comment{$\BigO(1)$}
            %%     \State $g_j\gets \Call{parent}{p_j}$
            %%     \Comment{$\BigO(1)$}
            %%     \State $\Call{Cut}{p_j}$
            %%     \Comment{$\BigO(\log \lmax)$}
            %%     \State $\Call{Cut}{c_j}$
            %%     \Comment{$\BigO(\log \lmax)$}
            %%     \State $\Call{Create}{p_j'}$ as `S-marked' node with child $c_j$ and parent $g_j$
            %%     \Comment{$\BigO(\log\lmax)$}
            %% \EndFor
            %% \State $\Call{Merge}{p_1,p_2}$ into a single `unmarked' node with parent $gy$ 
            %% \Comment{$\BigO(\log\log W + \log \lmax)$}
            %% \jhinline{$\Call{Merge}{}$ here is really a \Call{RootUnion}{} on the child lists of $p_1,p_2$ plus adding the new node to the light children of $gy$ which is a \Call{RootUnion}{} on the child list of $gy$ and the singleton that is the new node.}
            \State \Return
        \EndIf
        \If{$\SelectedLevel{} =i$}
            \Comment{A double edge becomes two single edges}
            \State $py.\mathrm{dchildren}\gets\Call{RootUnion}{px.\mathrm{dchildren}, py.\mathrm{dchildren}}$
            \Comment{$\BigO(\log \log W + \log\lmax)$}
            \State $px.\mathrm{dchildren}\gets\emptyset$
            \Comment{$\BigO(\log\lmax)$}
            \State decrease \SelectedLevel{} by $1$
            \State \Return
        \EndIf
    \EndIf
    \LComment{Here $cx$ is not an~ancestor of any selected node.}

    \jhinline{New version...}
    \State $r\gets \Call{Expose}{px}$
    \Comment{$\BigO( (1+\log\frac{W}{w(cx)})\cdot(\log\log W + \log\lmax) )$}
    \State \Call{Cut-Exposed}{$cx$}
    \Comment{$\BigO( \log\frac{w(px)}{w(cx)} + \log\log W + \log\lmax)$}
    \State \Call{UnselectHeavy}{$gx$}
    \Comment{$\BigO(\log\log W + \log\lmax)$}
    \State \Call{SelectHeavy}{$px$}
    \Comment{$\BigO(\log\log W + \log\lmax)$}
    \LComment{Now the tree containing $gx$ is $gx$-exposed and rooted at $r$.}
    \State Create a new node $p'$ with weight $w(cx)$
    \Comment{$\BigO(1)$}
    \State $\Call{Link}{cx,p'}$ 
    \Comment{$\BigO(\log \lmax)$}
    \State \Call{Link-Exposed}{$p',gx$}
    \Comment{$\BigO( \log\frac{w(gx)}{w(cx)} + \log\log W + \log\lmax)$}
    \State \Call{Conceal}{$r$}
    \Comment{$\BigO( (1+\log\frac{W}{w(gx)})\cdot(\log\log W + \log\lmax) )$}
    \EndFunction
  \end{algorithmic}
\end{algorithm}

\begin{algorithm}
	\begin{algorithmic}
		\caption{LongUnzip} \label{alg:nt-longunzip}
		\Function{LongUnzip}{$i_2, i_1$}
		\State $\SelectedLevel{}\gets i_1$
		\EndFunction
	\end{algorithmic}
\end{algorithm}

\begin{algorithm}
  \begin{algorithmic}%TODO: maybe split this into deselect and select.
    \caption{Select} \label{alg:nt-select}
    \Function{Select}{$X_{\mathrm{sel}}'=\set{x,y}$}
    \LComment{Do the actual zips corresponding to \SelectedLevel{}}
    %\State \Comment{remember selected level, in order to get the start/end of for-loop.}
    \State $s \gets \Call{SelectedLevel}{}$ \Comment{Save the old selected level}
    \State $\{x_o,y_o\} \gets \Xsel$
    \State $\Xsel \gets \emptyset$
%    \State pretend selected level is 0 in order to zip from selected level all the way up to 0.
    \State $cx_s \gets $ ancestor of $x_o$ at level $s$
    \Comment{$\BigO((1+\log \frac{W}{w(x_o)})\cdot (\log \lmax + \log\log W))$}
    \State $cy_s \gets $ ancestor of $y_o$ at level $s$
    \Comment{$\BigO((1+\log \frac{W}{w(y_o)})\cdot (\log \lmax + \log\log W))$}
    \For{$i=s-1$ \textbf{downto} $0$}
      \State $cx_i\gets \Call{Parent}{cx_{i+1}}$
      \Comment{$\BigO(\log \frac{w(cx_i)}{w(cx_{i+1})} + \log\log W)$}
      \State $cy_i\gets \Call{Parent}{cy_{i+1}}$
      \Comment{$\BigO(\log \frac{w(cy_i)}{w(cy_{i+1})} + \log\log W)$}
    \EndFor
    \For{$i=0$ \textbf{to} $s$}
      \State $\Call{Cut}{cx_i}$
      \Comment{$\BigO(\log \frac{w(cx_{i-1})}{w(cx_{i})} + \log\log W + \log \lmax)$}
      \State $\Call{Cut}{cy_i}$
      \Comment{$\BigO(\log \frac{w(cy_{i-1})}{w(cy_{i})} + \log\log W + \log \lmax)$}
    \EndFor
    \For{$i=0$ \textbf{to} $s$}
      \State $\Call{UnselectHeavy}{cx_i}$
      \Comment{$\BigO(\log\log W + \log\lmax)$}
      \State $\Call{UnselectHeavy}{cy_i}$
      \Comment{$\BigO(\log\log W + \log\lmax)$}
      \State Create a new node $m_i$
      \Comment{$\BigO(1)$}
      \State $m_i.\mathrm{dchildren}\gets \Call{RootUnion}{cx_i.\mathrm{dchildren},cy_i.\mathrm{dchildren}}$
      \Comment{$\BigO(\log\log W)$}
      \State $\Call{SelectHeavy}{m_i}$
      \Comment{$\BigO(\log\log W + \log\lmax)$}
    \EndFor
%    \For{$i=0$ \textbf{to} $s$}
%      \State Create a new node $m_i$
%      \Comment{$\BigO(1)$}
%      \State Unselect the heavy children of $cx_i, cy_i$ (if any)
%      \Comment{$\BigO(\log\log W + \log\lmax)$}
%      \State $m_i.\mathrm{dchildren}\gets \Call{RootUnion}{cx_i.\mathrm{dchildren},cy_i.\mathrm{dchildren}}$
%      \Comment{$\BigO(\log\log W)$}
%      \State Select new heavy child of $m_i$ (if any)
%        \Comment{$\BigO(\log\log W + \log\lmax)$}
%    \EndFor
    \For{$i=s$ \textbf{downto} $1$}
      \State $\Call{Link}{m_i,m_{i-1}}$
      \Comment{$\BigO(\log\frac{w(m_{i-1})}{w(m_i)} + \log\log W + \log\lmax)$}
    \EndFor
    \State $\Call{Link}{m_0,\mathrm{root}}$
    \Comment{$\BigO( (1+\log\frac{W}{m_0}) \cdot (\log\log W + \log\lmax))$}
    %% \State $\Call{FullPathUpdate}{m_s}$
    %% \Comment{$\BigO(\lmax \cdot (\log \log W + \log \lmax))$}
    %% \For{$i=0$ to $s$}
    %%   \State $\Call{ZipInto}{cx_i,cy_i,i}$
    %% \EndFor

    \LComment{Completely unzip the shared path for the new selection.}
    % \State $\set{x,y}\gets X_{\mathrm{sel}}'$
    \State $s' \gets \Level{x,y}$
    \Comment{$\BigO( (1+\log\frac{W}{w(x)}+\log\frac{W}{w(y)})\cdot (\log\log W + \log \lmax) )$}
    \State $cx\gets$ ancestor of $x$ at level $s'+1$
    \Comment{$\BigO( (1+\log\frac{W}{w(x)})\cdot (\log\log W + \log \lmax) )$}
    \State $px \gets \Call{Parent}{cx}$ 
    \Comment{$\BigO(\log \frac{w(px)}{w(cx)} + \log\log W)$}
    %% \State $\Call{UnselectHeavy}{px}$
    %% \Comment{$\BigO(\log \log W + \log \lmax)$}
    \State $\Call{Cut}{cx}$
    \Comment{$\BigO((1+\log \frac{W}{w(cx)} + \log\frac{W}{w'(px)})\cdot (\log \log W + \log \lmax))$}
    \State $\Call{InsertAt}{cx,s' + 1}$
    %% \Comment{$\BigO(\log \frac{W}{w(cx)} + s'\cdot(\log\log W + \log \lmax))$}
    \Comment{$\BigO( \log\frac{W}{w(cx)} + \lmax \cdot (\log\log W +  \log \lmax ))$}
    %% \State $\Call{FullPathUpdate}{px}$
    %% \Comment{$\BigO( \log\frac{W}{w(px)} + \lmax \cdot (\log\log W +  \log \lmax ))$}

    %% \State $cx\gets$ ancestor to $x$ at level $s'$
    %% \State $cy\gets$ ancestor to $y$ at level $s'$
    %% \For{$i = s'$ to $0$}
    %%     \State $\Call{UnzipAt}{cx, cy, i}$
    %%     \State $cx\gets\Call{Parent}{cx}$
    %%     \State $cy\gets\Call{Parent}{cy}$
    %% \EndFor
    \State $\SelectedLevel{} \gets s'$
    \State $X_{\mathrm{sel}} \gets X_{\mathrm{sel}}'$
    \EndFunction
  \end{algorithmic}
\end{algorithm}

\begin{algorithm}
  \begin{algorithmic}
    \caption{SelectedLevel} \label{alg:nt-selectedlevel}
    \Function{SelectedLevel}{{}}
    \State \Return \SelectedLevel{}
    \EndFunction
  \end{algorithmic}
\end{algorithm}

\begin{algorithm}
  \begin{algorithmic}
    \caption{SetWeight} \label{alg:nt-setweight}

    \Function{SetWeight}{$x$, $w$}
        %% \wninline{The code that was here was 1) Expose, 2) update heavy path with the new weight, 3) Conceal. I think it doesn't work at all cause the weight update may affect heavy path structure and we really need that log n time here}
        %% \jhinline{I am not seeing the problem. The running time as stated is $\BigO( (1+\log\frac{W}{w(x)}+\log\frac{W'}{w(x)})\cdot(\log\log W+\log\log W'+\log\lmax) )\subseteq \BigO(\log n\cdot\log\log n)$ which should be good enough?}
        %% \wninline{OK, with more understanding of Expose and Conceal purpose, now I understand the original version, which was fine, it should be restored}
    %% \State $\Call{FullPathUpdate}{x}$
    %% \Comment{$\BigO(\lmax \cdot (\log \log W + \log \lmax))$}
   \State $r\gets \Call{Expose}{x}$
   \Comment{$\BigO( (1+\log\frac{W}{w(x)})\cdot(\log\log W+\log\lmax))$}
   \State Update the weight of $x$
   \Comment {$\BigO(1)$}
   \State Update the heavy path tree containing $x$ with the new weight
   \Comment{$\BigO(\log\lmax)$}
   \State \Call{Conceal}{r}
   \Comment{$\BigO( (1+\log\frac{W'}{w(x')})\cdot(\log\log W'+\log\lmax))$}
    \EndFunction
  \end{algorithmic}
\end{algorithm}

\clearpage
\subsection{Extensions}
\label{sec:nt-extensions}

To handle the counting and marking extensions we need to augment the
nodes in our neighborhood tree, our heavy path trees, and our biased
disjoint set trees with some additional information.

For the heavy path trees, their leaves are nodes of the neighborhood tree  and the counting (respectively, marking) information in an internal node is simply the sum (respectively, or) of the interval of leaves represented by the node.

For the biased disjoint sets counting (respectively, marking) is, again, simply a sum (respectively, or), of the values for the subset represented by the light set. 
Furthermore, we also need to point to the ``heaviest'' element in the light set, and for this heaviest light child, store which selected nodes (if any) are below this child. \wninline{The above is not particularly clear to me in the aspect of selected nodes.}

For a neighborhood tree nodes $x$, counting (respectively, marking) information is computed from the information in the root of its biased disjoint set.
% by clearing (i.e., setting to zero, respectively, false), the counters (respectively, marks) whose level is higher than that of $x$.
\wninline{I commented out the mention ``by clearing (i.e., setting to zero, respectively, false), the counters (respectively, marks) whose level is higher than that of $x$.'' because I think it was wrong}
%For a node $x$ of a neighborhood tree, the correct counting (respectively, marking) information is computed as the sum (respectively, or) of a 
%For each node $x$ of a neighborhood tree, we should have a way of querying the sum of all counters of leaves within the subtree of $x$. A natural way to aggregate this would be to explicitly maintain that information for each node of the neighborhood tree, where these counter information for a node $x$ would be the sum of the respective counters in sons of $x$. However, that approach would be too slow and we define 

The pseudocode below for $\Call{SumCounters}{}$ and $\Call{OrMarks}{}$ are, for the sake of simplicity, written to take amortized time, by using expose. The only reason for this not being worst-case are the temporary changes to the light structures caused by the expose. It is a cumbersome but straightforward exercise to see that these calls to expose can be avoided, resulting in worst-case running times of the same form.
\wninline{As argued in the previous subsection, I think that in the end Expose and Conceal should be removed}
\jhinline{I disagree. Expose and Conceal greatly simplifies the code if used correctly. BTW, the only reason this becomes amortized when using Expose is that the updates in the BDS structure are amortized. }
\wninline{Per my comment in previous subsection, I agree that Expose and Conceal should stay.}
%Standard techniques transform these to worst-case time, since these calls make no changes to topology, weights, of the data structure. 

\begin{algorithm}
  \begin{algorithmic}
    \caption{UpdateMarks} \label{alg:nt-updatemarks}
    \Function{UpdateMarks}{$x, \mathbf{b}^x$}
      \State $r\gets \Call{Expose}{x}$
      \Comment{$\BigO((1+\log\frac{W}{w(x)})\cdot(\log\log W + \log\lmax))$}
      \State $x.\totalmarksvec \gets \mathbf{b}^x$
      \Comment{$\BigO(1)$}
      \For{each ancestor $v$ of $x$ in its heavy path tree, bottom-up order}
        \Comment{$\BigO(\log\lmax)$ times}
        \State $v.\totalmarksvec\gets v.\mathrm{left}.\totalmarksvec \bitor v.\mathrm{right}.\totalmarksvec$
        \Comment{$\BigO(1)$}
      \EndFor
      \State $\Call{Conceal}{r}$
      \Comment{$\BigO((1+\log\frac{W}{w(x)})\cdot(\log\log W + \log\lmax))$}
    \EndFunction
  \end{algorithmic}
\end{algorithm}

\begin{algorithm}
	\begin{algorithmic}
		\caption{UpdateCounters} \label{alg:nt-updatecounters}
		\Function{UpdateCounters}{$x, \mathbf{c}^x$}
		\State $r\gets \Call{Expose}{x}$
		\Comment{$\BigO((1+\log\frac{W}{w(x)})\cdot(\log\log W+\log\lmax))$}
		\State $x.\totalcntvec \gets \mathbf{c}^x$
		\Comment{$\BigO(1)$}
		\For{each ancestor $v$ of $x$ in its heavy path tree, bottom-up order}
		\Comment{$\BigO(\log\lmax)$ times}
		\State $v.\totalcntvec\gets v.\mathrm{left}.\totalcntvec + v.\mathrm{right}.\totalcntvec$
		\Comment{$\BigO(1)$}
		\EndFor
		\State $\Call{Conceal}{r}$
		\Comment{$\BigO((1+\log\frac{W}{w(x)})\cdot(\log\log W+\log\lmax))$}
		\EndFunction
	\end{algorithmic}
\end{algorithm}

\begin{algorithm}
	\begin{algorithmic}
		\caption{FindMarked} \label{alg:nt-findmarked}
		\Function{FindMarked}{$x, i$}
		\wninline{This pseudocode should be in my opinion majorly rewritten as per my e-mail from 6th March. In particular, as written $v.\totalmarksvec$ is a single vector associated with the vertex $v$ and we cannot read off what we need from that.
			
		I think we may even opt out of the pseudocode for this one and write it with words as we don't really have definitions/subroutines to express that succinctly in a pseudocode.}
		\State $v\gets $ ancestor to $x$ at level $i$
		\Comment{$\BigO((1+\log\frac{W}{w(x)})\cdot(\log\log W + \log\lmax))$}
		\If{$v.\totalmarksvec$ has a $0$ at position $i$}
		\Return $\bot$
		\Comment{$\BigO(1)$}
		\EndIf
		\State Traverse the tree downwards from $v$ until the marked leaf $y$ is found and return it
		\State\Comment{$\BigO((1+\log\frac{W}{w(y)})\cdot(\log\log W + \log\lmax))$}
		\EndFunction
	\end{algorithmic}
\end{algorithm}

\begin{algorithm}
  \begin{algorithmic}
    \caption{OrMarks} \label{alg:nt-ormarks}
    \Function{OrMarks}{$x$}
      \State $r\gets \Call{Expose}{x}$
      \Comment{$\BigO((1+\log\frac{W}{w(x)})\cdot(\log\log W + \log\lmax))$}
      \State $\mathbf{marks}\gets r.\totalmarksvec$
      \Comment{$\BigO(1)$}
      \State $\Call{Conceal}{r}$
      \Comment{$\BigO((1+\log\frac{W}{w(x)})\cdot(\log\log W + \log\lmax))$}
      \If{$\Xsel=\emptyset$}
        \Return $\mathbf{marks}$
        \Comment{$\BigO(1)$}
      \EndIf
      \State $\set{s_1,s_2}\gets \Xsel$
      \Comment{$\BigO(1)$}
      \State $a_1\gets$ $x \perp s_1$
      \Comment{$\BigO((1+\log\frac{W}{w(x)})\cdot(\log\log W + \log\lmax))$}
      \State $a_2\gets$ $x \perp s_2$
      \Comment{$\BigO((1+\log\frac{W}{w(x)})\cdot(\log\log W + \log\lmax))$}
      \LComment{Note: $s_1 \perp s_2 = \mathrm{root}$, so at least one of $a_1$, $a_2$ is the root.}
      \If{$a_1$ is not the root}
        \LComment{Need to add data from the $s_2$ strand, starting at \SelectedLevel{}.}
        \State $c \gets$ ancestor to $s_2$ at level $0$
        \Comment{$\BigO(\log\lmax)$}
      \ElsIf{$a_2$ is not the root}
        \LComment{Need to add data from the $s_1$ strand, starting at \SelectedLevel{}.}
        \State $c \gets$ ancestor to $s_1$ at level $0$
        \Comment{$\BigO(\log\lmax)$}
      \Else
        \LComment{The selected nodes have no effect on $x$.}
        \State \Return $\mathbf{marks}$
        \Comment{$\BigO(1)$}
      \EndIf
      \State Remove the root from its heavy path, but without updating $r.\mathrm{dchildren}$
      \Comment{$\BigO(\log\lmax)$}
      \State $\mathbf{marks}'\gets c.\totalmarksvec$ with all marks at index $>\SelectedLevel{}$ zeroed out
      \Comment{$\BigO(1)$}
      \State Add the root back to its heavy path
      \Comment{$\BigO(\log\lmax)$}
      \State \Return $(\mathbf{marks} \bitor \mathbf{marks}')$
      \Comment{$\BigO(1)$}
    \EndFunction
  \end{algorithmic}
\end{algorithm}

\begin{algorithm}
  \begin{algorithmic}
    \caption{SumCounters} \label{alg:nt-sumcounters}
    \Function{SumCounters}{$x$}
      \State $r\gets \Call{Expose}{x}$
      \Comment{$\BigO((1+\log\frac{W}{w(x)})\cdot(\log\log W + \log\lmax))$}
      \State $\mathbf{cnt}\gets r.\totalcntvec$
      \Comment{$\BigO(1)$}
      \State $\Call{Conceal}{r}$
      \Comment{$\BigO((1+\log\frac{W}{w(x)})\cdot(\log\log W + \log\lmax))$}
      \If{$\Xsel=\emptyset$}
        \Return $\mathbf{cnt}$
        \Comment{$\BigO(1)$}
      \EndIf
      \State $\set{s_1,s_2}\gets \Xsel$
      \Comment{$\BigO(1)$}
      \State $a_1\gets$ $x \perp s_1$
      \Comment{$\BigO(\log\lmax)$}
      \State $a_2\gets$ $x \perp s_2$
      \Comment{$\BigO(\log\lmax)$}
      \If{$a_1$ is not the root}
        \LComment{Need to add data from the $s_2$ strand, starting at \SelectedLevel{}.}
        \State $c \gets$ ancestor to $s_2$ at level $0$
        \Comment{$\BigO(\log\lmax)$}
      \ElsIf{$a_2$ is not the root}
        \LComment{Need to add data from the $s_1$ strand, starting at \SelectedLevel{}.}
        \State $c \gets$ ancestor to $s_1$ at level $0$
        \Comment{$\BigO(\log\lmax)$}
      \Else
        \LComment{The selected nodes have no effect on $x$.}
        \State \Return $\mathbf{cnt}$
        \Comment{$\BigO(1)$}
      \EndIf
      \State Remove the root from its heavy path, but without updating $r.\mathrm{dchildren}$
      \Comment{$\BigO(\log\lmax)$}
      \State $\mathbf{cnt}'\gets c.\totalcntvec$ with all counters at index $>\SelectedLevel{}$ zeroed out
      \Comment{$\BigO(1)$}
      \State Add the root back to its heavy path
      \Comment{$\BigO(\log\lmax)$}
      \State \Return $\mathbf{cnt} + \mathbf{cnt}'$
      \Comment{$\BigO(1)$}
    \EndFunction
  \end{algorithmic}
\end{algorithm}

\clearpage

\section{Biased Disjoint Sets (BDS) structure}\label{sec:bds}
%\msinline{I made this section slightly longer, laying out some more details. Could someone please read this and make sure it still works (and is readable enough).}
%\wninline{It is fine}

In this section, we prove \Cref{thm:bds}, that is, we provide an efficient implementation of the Biased Disjoint Sets data structure, therefore completing the stack of reductions required for proving \Cref{thm:main}.

Our structure is inspired by Binomial Heaps and Fibonacci Heaps~\cite{Vuillemin78,FredmanT87}: we maintain, for each set $X$ currently in the collection, a partition of $X$ into at most $t \leq 2\log_2 w(X)$ subsets $X_1,\ldots,X_t$, represented by a perfectly biased binary tree for each $X_i$ (which we call \emph{lower trees}).
These lower trees are then combined into a~single tree using a simple balanced dynamic binary tree (e.g.\ AVL or red-black tree), which we call the \emph{upper tree}; we require the depth of an~$\ell$-leaf upper tree to be $\BigO(\log \ell)$, and we require that two upper trees, with $\ell_a$ and $\ell_b$ leaves each, can be joined into one balanced upper tree in worst-case time $\BigO(\log (\ell_a + \ell_b))$.
Each leaf of the upper tree is identified with the root of a~lower tree.
Since the upper tree is balanced, it is guaranteed to have height $\BigO(\log t) = \BigO(\log\log w(X))$.
We call the resulting tree structure an~\emph{encoding} of $X$.
We associate the potential value $t$ with this encoding.

Each node of a~lower tree is assigned an~integer \emph{rank}.
The rank of a~leaf representing a~single element $x$ is $\left\lfloor \log_2 w(x) \right\rfloor$.
In turn, a~non-leaf node of rank $i$ has two children of rank $i - 1$.
This way, each lower tree is perfectly biased, and in particular, the height of a~lower tree representing a~subset $X_i$ of $X$ is at most $\log_2 w(X_i)$.
Also, for $x \in X_i$, the depth of $x$ in the lower tree is (by definition) $\BigO\paren*{\log\frac{w(X_i)}{w(x)}}$.
Thus, the depth of any $x\in X$ in the encoding of $X$ is $\BigO\paren*{\log\frac{w(X)}{w(x)}+\log\log w(X)}$, as required by the definition of almost perfectly biased tree.

For an~encoding of a~set $X$, we maintain the bound on the potential value $t$ by partitioning according to rank.
Whenever an~encoding of $X$ holds more than $2 \log_2 w(X)$ lower trees, we perform a~\emph{simplification} of the encoding by combining lower trees: whenever we have two perfectly biased trees of rank $i$, we can combine them into one perfectly biased tree of rank $i+1$.
This process is applied whenever the encoding holds multiple lower trees of the same rank; therefore, when nothing more can be combined, we will hold at most $\log_2 w(X)$ different lower trees.
Finally, the new encoding of $X$ is formed by joining the resulting lower trees into an~encoding via a~balanced upper tree.
Observe that this simplification process can be performed in time $\BigO(t)$ and decreases the potential value of the encoding by at least $t - \log_2 w(X) \geq \frac12 t$.
Therefore, in the amortized setting, simplification of an~encoding can be performed ``for free'': the drop in potential caused by this merging pays for the work done (because all additional work is paid for by the drop in potential).

Consider now the operations.
We implement $\Find{x}$ by following the tree encoding of the set $X \ni x$ from the leaf corresponding to $x$ to the root of the encoding of $X$.
This trivially can be done in worst-case time $\BigO\paren*{\log\frac{w(X)}{w(x)}+\log\log w(X)}$.

For $\RootUnion{X, Y}$, the direct cost is just the cost of joining the two upper trees, which is $\BigO(\log\log(w(X)+w(Y)))$. However, this may lead to the resulting encoding having too many lower trees. When this happens, we perform the simplification of the encoding as described above.
This way, the amortized cost of \RootUnion{} is the same as its actual cost. 

The most complicated operation is $\Delete{x}$. For $x\in X_i$, we delete the lower tree representing $X_i$ from the upper tree. Then, we split this lower tree into $s=\BigO\paren*{\log\frac{w(X_i)}{w(x)}}$ smaller perfectly biased binary trees $X'_1,\ldots,X'_s$ by deleting all $s$ ancestors of $x$. Then, we make a new upper tree with $X'_1,\ldots,X'_s$ as leaves. Finally, we perform a \RootUnion{} on the two encodings: the original encoding with $X_i$ removed and the new encoding of $X'_1, \ldots, X'_s$. Note that both the direct cost and the amortized cost of the entire operation is $\BigO\paren*{s + \log\log w(X)}$.

It remains to observe that both $\Coalesce{x, y}$ and $\Union{x, y}$ can be implemented in terms of a~constant number of invocations to \MakeSet{}, \RootUnion{}, \Find{}, and \Delete{}.
It can then be easily verified that both the direct and the amortized cost of both operations is then $\BigO\paren*{\log\frac{w(X)}{w(x)} + \log\frac{w(Y)}{w(y)}  + \log\log(w(X)+w(Y))}$.
%\jhinline{I know the above is a very loose description, but at least it is now written down somewhere...}

\section{Approximate counting} \label{sec:approx-counting}
In this section we analyze the function $T(\lmax)$ representing the maximum cost of a~single operation on a~counter vector.
We will use the method of \emph{approximate counting} introduced in a~work of Thorup~\cite{Thorup00} and subsequently used in several results on dynamic graph connectivity \cite{Holm18a,HuangHKPT23} to show that for the purposes of the biconnectivity data structure, we can choose $T(\lmax) \in \BigO(\log \lmax)$ in the word RAM model.
Since $\lmax \in \BigO(\log n)$, we will conclude that operations on counter vectors can be performed in worst-case time $\BigO(\log \log n)$.
Note that the result of this section is completely analogous to its counterpart in the work on the dynamic maintenance of $2$-edge-connected components of a~graph~\cite{Holm18a}.

First, observe that in \Cref{sec:graph-structure}, one can replace the invariant ($\dagger$) with the following weaker statement, for some fixed $C \in (1, 2)$:

\begin{description}
	\item[($\dagger$')] For every $i \in \mathbb{N}$, biconnected components in graph $G_i$ have at most $\ceil{\frac{n}{C^i}}$ vertices.
\end{description}

This increases the maximum possible cover level $\lmax$ of a~non-tree edge slightly, from $\log_2 n$ to $\BigO(\log n)$.
On the other hand, this allows us to give a~looser definition of \emph{small} and \emph{large} biconnected components of the graph $G_i$ in the proof of \Cref{lem:uncover-correctness}: every biconnected component of $G_i$ is declared \emph{small} or \emph{large}, but it must be declared as \emph{large} if it has strictly more than $\ceil{\frac{n}{C^{i+1}}}$ vertices and \emph{small} if it has at most $\frac12 \ceil{\frac{n}{C^i}}$ vertices.
Then, in order to correctly declare a~component as small or large, it is enough to know the \emph{multiplicative approximation} $s'$ of the size $s$ of the component: a~positive real number $s' \geq 1$ such that $s' \leq s \leq Bs'$, where $B = 2/C > 1$.
We omit technical details here since completely analogous arguments appear in \cite{Thorup00,Holm18a}.

Thus, following \cite{Thorup00}, we implement counters as short \emph{floating-point numbers} with a~$b$-bit mantissa and an~$s$-bit exponent, i.e., approximate positive integers with values of the form $(1 + x \cdot 2^{-b}) \cdot 2^y$ for $x \in [0, 2^b)$, $y \in [0, 2^s)$.
Given two floating-point representations $r(\alpha)$, $r(\beta)$ of positive integers $\alpha$, $\beta$, we can add them together and round down the result to the largest representable number, producing a~representation $r(\alpha + \beta)$ of $\alpha + \beta$.
Then whenever $\alpha + \beta < 2^{2^s}$ and both $r(\alpha) \leq \alpha \leq r(\alpha) \cdot K$, $r(\beta) \leq \beta \leq r(\beta) \cdot K$ for some $K \geq 1$, we will have $r(\alpha + \beta) \leq \alpha + \beta \leq r(\alpha + \beta) \cdot K(1 + 2^{-b})$.\footnote{On the other hand, this representation prevents us from subtracting counters: notably, subtraction of floating-point representations of integers is prone to the risk of \emph{catastrophic cancellation}.}
This motivates the notion of the \emph{computation depth} of a~counter.
Each counter, on initialization with an~integer, has computation depth $1$, and the sum of two counters of computation depths $d_\alpha$ and $d_\beta$, respectively, has computation depth $\max\{d_\alpha, d_\beta\} + 1$.
Then the following simple fact is essentially argued in~\cite{Thorup00,Holm18a}:

\begin{lemma}[{\cite{Thorup00,Holm18a}}]
  \label{lem:approx-prec}
  Fix $B > 1$.
  Then there exists a~constant $A > 0$, depending only on $B$, such that a~counter of computation depth at most $A \cdot 2^b$ storing a~floating-point representation $r(\alpha)$ of an~integer $\alpha \in [1, 2^{2^s})$ satisfies $r(\alpha) \leq \alpha \leq r(\alpha) \cdot B$.
\end{lemma}

In the case of our biconnectivity data structure, we only use counters to find sizes of biconnected components (in terms of the number of vertices).
These never exceed $n$, so it is enough to choose $s \in \Theta(\log \log n)$.
Then, by \Cref{lem:approx-prec}, it is enough to pick some $b \in \Theta(\log d)$, where $d$ is the maximum computation depth of any counter in the biconnectivity data structure.
The following claim estimates the value of $d$.
%\wninline{I'd argue it makes more sense to write $s \in \Theta(\log \log n)$ and $b \in \Theta(\log d)$ rather than $\Oh$}

\begin{lemma}
  \label{lem:float-depth}
  The data structure preserves $d \in \log^{\BigO(1)} n$.
\end{lemma}
\begin{proof}
  We make a~very rough estimate of $d$ here, disregarding any optimizations stemming from the use of weights in the internal data structures; in the arguments below, we implicitly use the facts that these weights never exceed $\BigO(n)$, and that $\lmax \in \BigO(\log n)$.

  In the BDS structure in \Cref{sec:bds}, we maintain a~binary tree of height $\BigO(\log n)$. Each node of the tree maintains a~counter vector computed from its children within a~constant number of operations, hence each counter has computation depth $\BigO(\log n)$.
  
  By the same token, the neighborhood structure of \Cref{sec:nbd} maintains a~balanced binary search tree of height $\BigO(\log n \log \log n)$, where each node stores $\BigO(1)$ counter vectors and additionally an~instance of the BDS structure.
  The counters in a~node are again easily found to be computable in $\BigO(1)$ operations from the counters stored in the children in the binary tree and the counters kept in the root of the associated BDS structure; hence every counter in the neighborhood structure has computation depth $\log^{\BigO(1)} n$.
  
  Finally, the cover level data structure is a~top tree of height $\BigO(\log n)$, with each cluster storing a~collection of counter vectors and matrices.
  These counters are only recomputed on cluster merges, where they are the product of a~constant number of operations on counter vectors and matrices stored in the children of the cluster and a~single neighborhood structure.
  Thus once again, all counters in the top tree have computation depth $\log^{\BigO(1)} n$.
  
  Finally, the biconnectivity data structure consumes the counters provided by the cover level data structure as they are, and only compares them to thresholds of the form $\left\lceil \frac{n}{C^i} \right\rceil$.
  Hence the lemma follows.
\end{proof}

So by \Cref{lem:float-depth} we can choose $b \in \Theta(\log \log n)$, and therefore we implement counters as floating-point numbers with a~representation of length $b + s \in \Theta(\log \log n)$.
Then using the standard bit-tricks exploiting the word RAM model~\cite{Thorup00,Holm18a}, we implement the \emph{counter vectors} of length $\BigO(\log n)$ by packing an~array of floating-point representations of these counters into $\BigO(\log \log n)$ RAM words (each of size $\Theta(\log n)$).
Hence we can choose $T(\lmax) \in \BigO(\log \lmax)$: all supported operations on counter vectors (see the list in \Cref{ssec:prelims-counters}) can be implemented in time $\BigO(\log \log n)$.

%\erinline{we get as corollary (analogous to cite-cite) that:
%
%-> $50\%$ was just shorthand for "some constant fraction"
%
%-> we can use approximate counting instead
%
%-> we can pack vectors into words and sum in an efficient manner
%
%-> makes a log factor into a loglog factor
%
%-> \cite{Thorup00,Holm18a}
%
%}
%
%\msinline{for myself}

\bibliography{references}

\appendix

%\msinline{To Wojtek: push {\tt graph-impl.tex}}
\section{Implementation of the graph structure} \label{sec:graph-impl}
%\wninline{Make sure that latex is not putting these codes in crazy places, aka Appendix B starts after the last pseudocode from Appendix A}
\begin{algorithm}[H]
	\begin{algorithmic}
		
		\caption{PromoteEdge} \label{alg:promote-edge}
		\Function {PromoteEdge} {$x, z, i$}
		\State $N^i(x) \gets N^i(x) \setminus \{z\}$, $\UpdateMark{x, i}$
		\State $N^{i+1}(x) \gets N^{i+1}(x) \cup \{z\}$, $\UpdateMark{x, i+1}$
		\State $N^i(z) \gets N^i(z) \setminus \{x\}$, $\UpdateMark{z, i}$
		\State $N^{i+1}(z) \gets N^{i+1}(z) \cup \{x\}$, $\UpdateMark{z, i+1}$
		\State $\Cover{x, z, i+1}$
		\EndFunction
		
	\end{algorithmic}
\end{algorithm} 

\begin{algorithm}[H]
	\begin{algorithmic}
		
		\caption{FindNextEvent} \label{alg:find-next-event}
		\Function {FindNextEvent} {$a, b, i$}
		\State $(lr, p, u) \gets \FindFirstReach{a, b, i}$
		\If {$l = \perp$}
		\State \Return $((\perp, \perp), \perp, \perp)$
		\EndIf
		\If {$r = p$}
		\msinline{\footnotesize Above: Shouldn't we perhaps test for $p = l$? Below: What's $z$?}
		\wninline{No, there is no need to. It will be correctly handled by the last line. $z$ was a typo, should have been $r$.}
		\State $w \gets \FindStrongReach{a, b, lr, r, i}$
		\If {$w \neq \perp$}
		\State \Return $((w, \text{ any element of } N^i(w)), p, L(p))$
		\EndIf
		\EndIf
		\State \Return $((u, \text{ any element of } N^i(u)), p, R(p))$
		\msinline{\footnotesize above: Is it ok that we return the edge $ps^b(p)$ if $r=p$ and $w = \bot$?}
		\wninline{Yes, that's exactly the intention and the paragraph preceding that pseudocode is exactly the justification for that. But I guess that was not clear? Maybe it's better now that I actually described what FindNextEvent is supposed to return...}
		\EndFunction
		
	\end{algorithmic}
\end{algorithm} 

\begin{algorithm}[H]
	\begin{algorithmic}

		\caption{FindReplacement and Swap} \label{alg:swap}
		\Function{FindReplacement} {$u, v, i$}
		\ForAll {$(a, b) \in \{(u, v), (v, u)\})$}
		\While \True
		\State $((x, y), \cdot, \cdot, \cdot) \gets \FindNextEvent{a, b, i}$
%		\wninline{Adjust to changes to FindNextEvent. Make sure it works such that $x$ is reachable through blahblah}
		\If {$s^y(a) = b$}
		\State \Return $(x, y)$
		\Else
		\If {$\FindSize{x, y, i+1} \le \ceil{\frac{n}{2^{i+1}}}$}
		\State $\PromoteEdge{x, y}$
		\Else
		\State \Break
		\EndIf
		\EndIf
		\EndWhile
		\EndFor
		\EndFunction
		
		\medskip
		
		\Function {Swap} {$u, v$}
		\State $i \gets -1$
		\For {$j \gets \lmax \Downto 0$}
		\If {$\FindSize{u, v, j} > 2$}
		\State $i \gets j$
		\State \Break
		\EndIf
		\EndFor
		\State \Assert $i \neq -1$
		\State $(x, y) \gets \Call{FindReplacement}{u, v, i}$
		\State $\Cut{u, v}$
		\State $N^i(x) \gets N^i(x) \setminus \{y\}, \UpdateMark{x, i}$
		\State $N^i(y) \gets N^i(y) \setminus \{x\}, \UpdateMark{y, i}$
		\State $\Link{x, y}$
		\State $N^i(u) \gets N^i(u) \cup \{v\}, \UpdateMark{u, i}$
		\State $N^i(v) \gets N^i(v) \cup \{u\}, \UpdateMark{v, i}$
		\For {$j \gets 0 \ .. \ i$}
		\State $\Cover{u, v, i}$
		\EndFor
		\EndFunction
		
	\end{algorithmic}
\end{algorithm} 

\begin{algorithm}[H]
	\begin{algorithmic}
		
		\caption{Delete} \label{alg:delete}
		\Function {Delete} {$u, v$}
		\State $i \gets $ level of $uv$
		\If {$i = -1$}
		\State $\Call{Swap}{u, v}$
		\State $i \gets $ level of $uv$
		\EndIf
		\State $\Expose{u, v}$
		\For {$j \gets i \Downto 0$}
		\State $N^j(u) \gets N^j(u) \setminus \{v\}, \UpdateMark{u, j}$
		\State $N^j(v) \gets N^j(v) \setminus \{u\}, \UpdateMark{v, j}$
		\If {$j \neq 0$}
		\State $N^{j-1}(u) \gets N^{j-1}(u) \cup \{v\}, \UpdateMark{u, j-1}$
		\State $N^{j-1}(v) \gets N^{j-1}(v) \cup \{u\}, \UpdateMark{v, j-1}$
		\EndIf
		\State $\UncoverPath{u, v, j}$
		\EndFor
		\EndFunction
		
	\end{algorithmic}
\end{algorithm}

\begin{algorithm}[H]
	\begin{algorithmic}
		
		\caption{UncoverPath} \label{alg:uncover-path}
		\Function {UncoverPath} {$u, v, i$}
		\ForAll {$(a, b) \in \{(u, v), (v, u)\})$}
		\State $(l_{\rm prv}, r_{\rm prv}, p_{\rm prv}) \gets (\perp, \perp, \perp)$
		%\State $active\_edges = \{\}$
		%\State ${\rm stopped} \gets \False$
		\While \True
		\State $((x, z), p, (l, r)) \gets \Call{FindNextEvent}{a, b, i}$
		\LComment{We assume that $l$ is closer to $a$ than $r$.}
		\msinline{\footnotesize Isn't this break too early? I imagine we might still need to do a~pending local uncover at $p_{\rm prv}$}
		\wninline{That was a good point. I changed this code a lot. Please verify.}
		%				\If {$x = \perp$}
		%					\State \Break
		%				\EndIf
		\If {$p_{prv} = \perp$}
		\If {$x = \perp$}
		\State $\UniformUncover{a, b, i}$
		\State \Return
		\Else
		\State $\UniformUncover{a, p, i}$
		\EndIf
		\ElsIf {$(l_{\rm prv}, r_{\rm prv}, p_{\rm prv}) \neq (l, r, p)$}
		\If {$r_{\rm prv} = p_{\rm prv}$ \textbf{and} $\CoverLevel{s^a(p_{\rm prv}), s^b(p_{\rm prv})} = i$}
		\State $\LocalUncover{s^a(p_{\rm prv})p_{\rm prv}, p_{\rm prv} s^b(p_{\rm prv}), i}$
		\EndIf
		\If {$x = \perp$}
		\State $\UniformUncover{p_{prv}, b, i}$
		\State \Return
		\EndIf
		\If {$p \neq p_{\rm prv} \text{ \textbf{and} } p=l$ and  $\CoverLevel{s^a(p), s^b(p)} = i$}
		\State $\LocalUncover{s^a(p)p, ps^b(p), i)}$
		\EndIf
		\State $\UniformUncover{p_{\rm prv}, p, i}$
		%					\EndIf
		\EndIf
		\If {$\FindSize{x, z, i+1} \le \ceil{\frac{n}{2^{i+1}}}$}
		\State $\PromoteEdge{x, z, i}$
		\Else
		%\State ${\rm stopped} \gets \True$
		\State \Break
		\EndIf
		\State $(l_{\rm prv}, r_{\rm prv}, p_{\rm prv}) \gets (l, r, p)$
		\EndWhile
		%			\If {\Not ${\rm stopped}$}
		%				\State $\UniformUncover{p_{\rm prv}, b, i}$
		%				\State \Return
		%			\EndIf
		\EndFor
		\EndFunction
		
	\end{algorithmic}
\end{algorithm}

\section{Implementation of the tree structure}
 
\subsection{Bookkeeping in top trees}
\label{ssec:tree-impl-bookkeeping}

A~top trees data structure updates and propagates auxiliary information via the following callback functions:

\begin{itemize}
  \item $\OnCreate{C}$, called upon the creation of a~leaf cluster $C$.
    Recall that by our design choice, leaf clusters represent either single edges of the underlying tree $T$ (call these \emph{proper} leaf clusters), or dummy edges $vv'$ with one endpoint at a~vertex $v \in V(T)$ and the other at a~fictional leaf $v'$ (\emph{dummy} leaf clusters).
  \item $\OnMerge{C}$, called when a~cluster $C$ is formed as a~merge of several children of $C$.
    The merge always happens according to one of the cases in \Cref{fig:toptree-cases}; and, by looking at $C$, we can infer which case holds and whether $C$ is formed during a~transient (un)expose.
    After the merge, auxiliary information stored in $C$ needs to be recomputed based on data stored in the children of $C$; and whenever $C$ is formed outside of a~transient (un)expose, weights and the sets of selected items in neighborhood data structures need to be updated.
  \item $\Clean{C}$, called whenever lazy updates stored in $C$ need to be propagated into child clusters.
  This happens either just before $C$ is split, or whenever recursive search along a~vertical path of a~top tree is performed (e.g., in \Cref{ssec:implicit-refine}).
\end{itemize}

In this section, we show how to maintain the correct values of $\firstedge_{C,x}$ for $x \in \bnd C$, the enclosing clusters $I_x$ for $x \in V(T)$, the identification of cluster and interface edges (i.e., the sets $\interface_x$) and the correct edge and vertex weights.
Moreover, we construct neighborhood data structures $N_x$ for $x \in V(T)$, where we preserve the correct weights and sets of selected edges incident to~$x$.
Finally, for each cluster edge $\vec{uv}$, we maintain the smallest point cluster $\smallestpoint_{\vec{uv}}$ containing the edge with the boundary equal to $u$ (i.e., the smallest witness that $\vec{uv}$ is indeed a~cluster edge).

We refer to \Cref{alg:toptree-basic} for the implementation.
Here, $\size_C$ describes the number of vertices in the cluster $C$, excluding the boundary vertices or dummy leaves of $T$.

%\msinline{additional callback after all merges are finalized, so that we can process additional information based on present root clusters?}

\begin{algorithm}[H]
\begin{algorithmic}
  \caption{Basic information maintenance in top trees and neighborhood data structures} \label{alg:toptree-basic}
  \Function{TS.OnCreate}{$C$}
    \State $\size_C \gets 0$
    \If{$C$ is a~proper leaf cluster}
      \State $\{v, w\} \gets \bnd C$
      \State $\firstedge_{C,v} \gets \vec{vw}$,\, $\firstedge_{C,w} \gets \vec{wv}$
    \EndIf
  \EndFunction

  \bigskip

  \Function{TS.OnMerge}{$C$}
    \If{$C = A \cup B$ \textbf{and} $\bnd C = \bnd A = \bnd B \eqqcolon \{v\}$} \Comment{Case \CasePointToPointPoint}
      \State $\size_{C,v} \gets \size_{A,v} + \size_{B,v}$
    \ElsIf{$C = A \cup B$ \textbf{and} $\bnd C \eqqcolon \{v\}$, $\bnd A \eqqcolon \{v, x\}$, $\bnd B \eqqcolon \{x\}$} \Comment{Case \CasePointToPathPoint}
      \State $I_x \gets C$
      \State $\firstedge_{C, v} \gets \firstedge_{A, v}$
      \State $\interface_x \gets \{\firstedge_{A,x}\}$
      \State $\smallestpoint_{\firstedge_{A,v}} \gets C$
      \State $\size_C \gets \size_A + \size_B + 1$
      \If{$C$ not formed during transient (un)expose}
        \State $N_x.\Select{\emptyset}$
        \State $N_v.\SetWeight{\firstedge_{A,v}, \size_C}$
        \State $N_x.\SetWeight{\firstedge_{A,x}, \size_B + 1}$
      \EndIf
    \Else  \Comment{Case \CasePath}
      \State $C \eqqcolon A \cup B \cup P$, $\bnd C \eqqcolon \{v, w\}$, $\bnd A \eqqcolon \{v, x\}$, $\bnd B \eqqcolon \{x, w\}$, $\bnd P \eqqcolon \{x\}$
      \State $I_x \gets C$
      \State $\firstedge_{C,v} \gets \firstedge_{A,v}$
      \State $\firstedge_{C,w} \gets \firstedge_{B,w}$
      \State $\interface_x \gets \{\firstedge_{A,x}, \firstedge_{B,x}\}$
      \State $\size_C \gets \size_A + \size_B + \size_P + 1$
      \If{$C$ not formed during transient (un)expose}
        \State $N_x.\Select{\interface_x}$
        \State $N_x.\SetWeight{\firstedge_{A,x}, \size_P + 1}$
        \State $N_x.\SetWeight{\firstedge_{B,x}, \size_P + 1}$
      \EndIf
    \EndIf
  \EndFunction
\end{algorithmic}
\end{algorithm}

\subsection{Cover level data structure}
\label{ssec:tree-impl-cover}

We augment the previous data structure with the information on cover levels.
In the following algorithms, we assume that transient expose is invoked via $\TransientExpose{v, w}$, and it returns three transient root clusters $R_v, R_w, R_{vw}$, such that $\bnd R_v = \{v\}$, $\bnd R_w = \{w\}$, $\bnd R_{vw} = \{v, w\}$.
Then at the end of the query, the top tree must be transiently unexposed by calling $\TransientUnexpose$, thus restoring the original shape of the data structure.

\begin{algorithm}[H]
\begin{algorithmic}
  \caption{Maintaining auxiliary information and lazy updates} \label{alg:toptrees-cover-info}

  \Function{CL.OnCreate}{$C$}
    \State $\cover_C,\, \coverfrom_C,\, \covertop_C \gets \lmax$
    \State $\argcover_C \gets \bot$
  \EndFunction

  \bigskip

  \Function{CL.OnMerge}{$C$}
    \If{$|\bnd C| = 1$} \Comment{Cases \CasePointToPointPoint, \CasePointToPointPath}
      \State $\cover_C \gets \lmax$
      \State $\argcover_C \gets \bot$
    \Else \Comment{Case \CasePath}
      \State $C \eqqcolon A \cup B \cup P$, $\bnd C \eqqcolon \{v, w\}$, $\bnd A \eqqcolon \{v, x\}$, $\bnd B \eqqcolon \{x, w\}$, $\bnd P \eqqcolon \{x\}$
      \State $c \gets N_x.\Level{\firstedge_{A,x}, \firstedge_{B,x}}$
      \State $\cover_C \gets \min\{\cover_A, c, \cover_B\}$
      \If{$\cover_C = \cover_A$ and $\firstedge_{A,v} \neq \firstedge_{A,x}$}
        \State $\argcover_{C,v} \gets \argcover_{A,v}$
      \ElsIf{$\cover_C = c$}
        \State $\argcover_{C,v} \gets (\firstedge_{A,x}, \firstedge_{B,x})$
      \Else
        \State $\argcover_{C,v} \gets \argcover_{B,x}$
      \EndIf
      \LComment{$\argcover_{C,w}$ is computed symmetrically.}
    \EndIf
    \State $\coverfrom_C,\, \covertop_C \gets \cover_C$
  \EndFunction

  \bigskip

  \Function{CL.PropagateLazyUpdates}{$C$, $D$}
    \State \textbf{assert} $\cover_D \geq \coverfrom_C$
    \If{$\covertop_C \geq \cover_D$}
      \State $\cover_D \gets \cover_C$
      \State $\covertop_D \gets \max\{\covertop_C, \covertop_D\}$
    \EndIf
  \EndFunction

  \bigskip

  \Function{CL.Clean}{$C$}
    \If{$C = A \cup B \cup P$} \Comment{Case \CasePath}
      \State $\bnd C \eqqcolon \{v, w\}$, $\bnd A \eqqcolon \{v, x\}$, $\bnd B \eqqcolon \{x, w\}$, $\bnd P \eqqcolon \{x\}$
      \State \Call{CL.PropagateLazyUpdates}{$C$, $A$}
      \State \Call{CL.PropagateLazyUpdates}{$C$, $B$}
      \State $c \gets N_x.\SelectedLevel$
      \If{$c < \cover_C$}
        \State $N_x.\LongZip{c, \cover_C}$
      \ElsIf{$c \leq \covertop_C$ and $c > \cover_C$}
        \State $N_x.\LongUnzip{c, \cover_C}$
      \EndIf
    \EndIf
    \State $\coverfrom_C,\, \covertop_C \gets \cover_C$
  \EndFunction
\end{algorithmic}
\end{algorithm}

\begin{algorithm}[H]
\begin{algorithmic}
  \caption{Performing updates} \label{alg:toptree-cover-updates}
  
  \Function{CL.RecursiveCover}{$C$, $i$}
    \State \textbf{assert} $C = A \cup B \cup P$, $\bnd C = \{v, w\}$, $\bnd A = \{v, x\}$, $\bnd B = \{x, w\}$, $\bnd P = \{x\}$ \Comment{Case \CasePath}
    \State \textbf{assert} $\cover_C \geq i - 1$
    \If{$C$ is a~transient cluster}
      \State \Call{CL.RecursiveCover}{$A$, $i$}
      \State \Call{CL.RecursiveCover}{$B$, $i$}
      \If{$N_x.\Level{\firstedge_{A,x}$, $\firstedge_{B,x}} = i - 1$}
        \State $N_x.\Zip{\firstedge_{A,x}$, $\firstedge_{B,x}}$
      \EndIf
    \Else
      \If{$\cover_{C_j} = i - 1$}
        \State $\cover_{C_j} \gets i$
        \State $\covertop_{C_j} \gets \max\{\covertop_{C_j}, i\}$
      \EndIf
    \EndIf
  \EndFunction
  
  \bigskip

  \Function{CL.Cover}{$p$, $q$, $i$}
    \State $R_p, R_q, R_{pq} \gets \TransientExpose{p, q}$
    \Call{CL.RecursiveCover}{$R_{pq}$, $i$}
    \State $\TransientUnexpose$
  \EndFunction

  \bigskip
  
  \Function{CL.RecursiveUniformUncover}{$C$, $i$}
    \State \textbf{assert} $C = A \cup B \cup P$, $\bnd C = \{v, w\}$, $\bnd A = \{v, x\}$, $\bnd B = \{x, w\}$, $\bnd P = \{x\}$ \Comment{Case \CasePath}
    \State \textbf{assert} $\cover_C \geq i$
    \If{$C$ is a~transient cluster}
      \State \Call{CL.RecursiveUniformUncover}{$A$, $i$}
      \State \Call{CL.RecursiveUniformUncover}{$B$, $i$}
      \If{$N_x.\SelectedLevel{\xspace} = i$}
        \State $N_x.\LongUnzip{i, i - 1}$
      \EndIf
    \Else
      \If{$\cover_{C_j} = i$}
        \State $\cover_{C_j} \gets i - 1$
      \EndIf
    \EndIf
  \EndFunction
  
  \bigskip

  \Function{CL.UniformUncover}{$p$, $q$, $i$}
    \State $R_p, R_q, R_{pq} \gets \TransientExpose{p, q}$
    \Call{CL.RecursiveUniformUncover}{$R_{pq}$, $i$}
    \State $\TransientUnexpose$
  \EndFunction

  \bigskip

  \Function{CL.LocalUncover}{$e_1$, $e_2$, $i$}
    \State $e_1, e_2 \eqqcolon px, xq$
    \State $\TransientExpose{x}$
    \State $N_x.\Unzip{e_1, e_2, i}$
    \State $\TransientUnexpose$
  \EndFunction
\end{algorithmic}
\end{algorithm}

%\wninline{Adjust LocalUncover to the new simpler way (just change (p,q) to (x) in the argument of TransientExpose I believe). Also, I think that the Unzip can actually be LongUnzip :p}
%\msinline{TransientExpose: ok, changed to $x$. Unzip is the only correct way here, though: $e_1, e_2$ might not be the selected items in $N_x$ since transient exposes won't change any sets of selected items}
%\wninline{I still think that LongUnzip is fine, cause LocalUncover also satisfies the same property as UniformUncover that the pair of edges is selected. But if that does not hurt us then whatever.}

\begin{algorithm}[H]
\begin{algorithmic}
  \caption{Answering cover level queries} \label{alg:toptree-cover-queries}

  \Function{CL.CoverLevel}{$p$, $q$}
    \State $R_p, R_q, R_{pq} \gets \TransientExpose{v, w}$
    \State $\mathrm{answer} \gets \cover_{R_{pq}}$
    \State $\TransientUnexpose$
    \State \Return $\mathrm{answer}$
  \EndFunction

  \bigskip
  
  \Function{CL.MinCoveredPair}{$p$, $q$}
    \State $R_p, R_q, R_{pq} \gets \TransientExpose{v, w}$
    \State $\mathrm{answer} \gets \argcover_{R_{pq}}$
    \State $\TransientUnexpose$
    \State \Return $\mathrm{answer}$
  \EndFunction
\end{algorithmic}
\end{algorithm}

\subsection{Counting reachable vertices}
\label{ssec:tree-impl-count}

We begin by implementing the helper function $N_v.\SumCounters{A}$ for a~given set $A$ of items stored in $N_v$, as shown in \Cref{lem:excl-sum-counters}; see \Cref{alg:toptree-count-sumcounters}.
Afterwards, we show how to maintain counters for the top tree clusters (\Cref{alg:toptrees-count-info}).
Finally, we resolve the $\FindSize{}$ query in \Cref{alg:toptree-count-query}.

In the algorithms, we extensively manipulate counter vectors and matrices; we use the notation laid down in \Cref{ssec:prelims-counters}.
We additionally denote the all-one counter by ${\bf 1}$.

\begin{algorithm}[H]
\begin{algorithmic}
  \caption{Implementation of $\SumCounters{}$ from \Cref{lem:excl-sum-counters}} \label{alg:toptree-count-sumcounters}
  \Function{$N_v$.SumCounters}{$A_+$, $A_-$}
    \If{$|A_+| = 0$}
      \State \Return ${\bf 0}$
    \Else
      \State $e \gets$ arbitrary element of $A_+$
      \State ${\bf c}_{e} \gets N_v.\SumCounters{e}$ \Comment{Provided by \Cref{lem:neighborhood-ds}}
      \State ${\bf c}_{\neq e} \gets N_v.\SumCounters{A_+ \setminus \{e\}, A_-}$
      \State $p \gets -1$
      \ForAll{$e' \in (A_+ \cup A_-) \setminus \{e\}$}
        \State $p \gets \max\{p, N_v.\Level{e, e'}\}$
      \EndFor
      \State \Return ${\bf c}_{\neq e} + \vectorsplice{\bf 0}{p + 1}{{\bf c}_e}$
    \EndIf
  \EndFunction
  
  \bigskip  
  
  \Function{$N_v$.SumCounters}{$A_+$}
    \State \Return \Call{$N_v$.SumCounters}{$A_+$, $\emptyset$}
  \EndFunction
\end{algorithmic}
\end{algorithm}

\begin{algorithm}[H]
\begin{algorithmic}
  \caption{Maintaining counters in the top tree clusters} \label{alg:toptrees-count-info}
  \Function{CR.OnCreate}{$C$}
    \State $\totalcntvec_C \gets {\bf 0}$
    \State $\diagcntvec^\star_{C,x} \gets {\bf 0}$ for each $x \in \bnd C$
  \EndFunction

  \bigskip

  \Function{CR.GetDiagCnt}{$C$, $v$} \Comment{Computes $\diagcntvec_{C,v}$}
    \State ${\bf A} \gets \matrixsplice{\diagcntvec^\star_{C,v}}{\covertop_C + 1}{\bf 0}$ \Comment{$A_{j,i} = \diagcnt^\star_{C,v,j,i} \cdot [j \leq \covertop_C]$}
    \State ${\bf B} \gets \matrixsplice{\bf 0}{\covertop_C + 1}{\diagcntvec^\star_{C,v}}$ \Comment{$B_{j,i} = \diagcnt^\star_{C,v,j,i} \cdot [j > \covertop_C]$}
    \State ${\bf w} \gets \matrixsum({\bf A})$ \Comment{$w_i = \sum_j A_{j,i}$}
    \State ${\bf M} \gets \addvector({\bf B}, {\bf w}, \cover_C)$ \Comment{$M_{j,i} = B_{j,i} + w_i \cdot [j = \cover_C]$}
    \State \Return ${\bf M}$
  \EndFunction

  \bigskip

  \Function{CR.GetTotalCnt}{$C$, $v$} \Comment{Computes $\totalcntvec_{C,v}$}
    \State ${\bf M} \gets \Call{CR.GetDiagCnt}{C, v}$
    \State ${\bf w} \gets \uppersum({\bf M})$ \Comment{$w_i = \sum_{j \geq i} M_{j, i}$}
    \State \Return ${\bf w}$
  \EndFunction

  \bigskip

  \Function{CR.OnMerge}{$C$}
    \If{$C = A \cup B$, $\bnd C \eqqcolon \{v\}$, $\bnd A \eqqcolon \{v, x\}$, $\bnd B \eqqcolon \{x\}$} \Comment{Case \CasePointToPathPoint}
      \State ${\bf c}^x \gets N_x.\SumCounters{\{\firstedge_{A, x}\}} + {\bf 1}$
      \State $e \gets \firstedge_{A, v}$
      \State $\totalcntvec_{A, v} \gets \Call{CR.GetTotalCnt}{A, v}$
      \State $\clustercntvec_e \gets \totalcntvec_{A, v} + \vectorsplice{{\bf c}^x}{\cover_A + 1}{\bf 0}$
      \State $N_v.\UpdateCounters{e, \clustercntvec_e}$
    \ElsIf{$C = A \cup B \cup P$, $\bnd C \eqqcolon \{v, w\}$, $\bnd A \eqqcolon \{v, x\}$, $\bnd B \eqqcolon \{x, w\}$, $\bnd P \eqqcolon \{x\}$} \Comment{Case \CasePath}
      \State ${\bf c}^{AB} \gets N_x.\SumCounters{\{\firstedge_{A, x},\, \firstedge_{B, x}\}} + {\bf 1}$
      \State ${\bf c}^{A} \gets N_x.\SumCounters{\{\firstedge_{A, x}\}} + {\bf 1}$
      \State ${\bf c}^{B} \gets N_x.\SumCounters{\{\firstedge_{B, x}\}} + {\bf 1}$
      \State ${\bf c}^{A \setminus B} \gets N_x.\SumCounters{\{\firstedge_{A, x}\}, \{\firstedge_{B, x}\}}$
      \State ${\bf c}^{B \setminus A} \gets N_x.\SumCounters{\{\firstedge_{B, x}\}, \{\firstedge_{A, x}\}}$
      \State $\ell_{AB} \gets N_x.\Level{\firstedge_{A, x}, \firstedge_{B, x}}$
      \State $\totalcntvec_C \gets \totalcntvec_A + \totalcntvec_B + {\bf c}^{AB}$
      \ForAll{$(\hat{v}, \hat{A}, \hat{B}) \in \{(v, A, B), (w, B, A)\}$}
        \State $\diagcntvec_{\hat{A}, \hat{v}} \gets \Call{CR.GetDiagCnt}{\hat{A}, \hat{v}}$
        \State $\diagcntvec_{\hat{B}, x} \gets \Call{CR.GetDiagCnt}{\hat{B}, x}$
        \State $r \gets \cover_{\hat{A}}$
        \State ${\bf M} \gets \matrixsplice{\diagcntvec_{\hat{B}, x}}{\min\{r, \ell_{AB}\}}{\diagcntvec_{\hat{A}, \hat{v}}}$
        \State ${\bf u} \gets \matrixsum(\matrixsplice{\bf 0}{\min\{r, \ell_{AB}\}}{\diagcntvec_{\hat{B}, x}})$
        \State ${\bf M} \gets \addvector({\bf M}, {\bf u}, \min\{r, \ell_{AB}\})$
        \If{$\ell_{AB} < r$}
          \State ${\bf M} \gets \addvector({\bf M}, {\bf c}^{\hat{A}}, r)$
          \State ${\bf M} \gets \addvector({\bf M}, {\bf c}^{\hat{B} \setminus \hat{A}}, \ell_{AB})$
        \Else
          \State ${\bf M} \gets \addvector({\bf M}, {\bf c}^{AB}, r)$
        \EndIf
        \State $\diagcntvec^\star_{\hat{A}, \hat{v}} \gets {\bf M}$
      \EndFor
    \EndIf
  \EndFunction
\end{algorithmic}
\end{algorithm}
%\wninline{I think that adding ones to $c^A$ etc. is not consistent with how they were defined in the text. Hence, I'd suggest removing them from the definitions of the variables and inline $+1$ where needed in the following parts.}
%\msinline{Uhhh I know. Leaving this as low prio.}

\begin{algorithm}[H]
\begin{algorithmic}
  \caption{$\FindSize{}$ query} \label{alg:toptree-count-query}
  \Function{CR.FindSize}{$p$, $q$, $i$}
    \State $R_p, R_q, R_{pq} \gets \TransientExpose{p, q}$
    \State $\mathrm{answer}_p \gets 1 + \left(N_p.\SumCounters{\{\firstedge_{R_{pq}, p}\}}\right)_i$
    \State $\mathrm{answer}_q \gets 1 + \left(N_q.\SumCounters{\{\firstedge_{R_{pq}, q}\}}\right)_i$
    \State $\mathrm{answer}_{pq} \gets \totalcnt_{R_{pq}, i}$
    \State $\TransientUnexpose$
    \State \Return $\mathrm{answer}_p + \mathrm{answer}_q + \mathrm{answer}_{pq}$
  \EndFunction
\end{algorithmic}
\end{algorithm}

\subsection{Finding reachable vertices}
\label{ssec:tree-impl-find}

Recall that we keep, for each vertex $v$ of the tree, a~bit vector $\bmarksvec_v$ where $\bmarksvec_{v,i}$ denotes whether $v$ is $i$-marked.
Initially, all bit vectors are identically zero.
Moreover, we assume access to a~copy of the neighborhood data structure $N_v$ named $N^+_v$.
We assume that all updates to $N_v$ performed in \Cref{ssec:tree-impl-bookkeeping,ssec:tree-impl-cover} are also applied to $N^+_v$.

We first augment $N_v$ with methods $\FindMarked{A, i}$ and $\OrMarks{A, i}$ and $N^+_v$ with $\FindStrongMarked{A, i}$ as described in \Cref{lem:excl-or-marks,lem:excl-find-marked}; see \Cref{alg:toptree-find-aux}.
Afterwards, we maintain bit vectors in the top tree clusters analogously to the counters in \Cref{ssec:tree-impl-count} (\Cref{alg:toptree-find-info}).
Next, we resolve the updates and queries in \Cref{alg:toptree-find-updates,alg:toptree-find-queries}, delaying the implementation of the refinement functions until \Cref{alg:toptree-find-refine}.

\begin{algorithm}[H]
\begin{algorithmic}
  \caption{$\OrMarks{}$, $\FindMarked{}$ and $\FindStrongMarked{}$ from \Cref{lem:excl-or-marks,lem:excl-find-marked}} \label{alg:toptree-find-aux}
  \Function{$N_v$.OrMarks}{$A_+$, $A_-$}
    \State ${\bf a} \gets {\bf 0}$
    \ForAll{$e \in A_+$}
      \State ${\bf c}_e \gets N_v.\OrMarks{e}$ \Comment{Provided by \Cref{lem:neighborhood-ds}}
      \State $p \gets -1$
      \ForAll{$f \in A_-$}
        \State $p \gets N_v.\Level{e, f}$
      \EndFor
      \State ${\bf a} \gets {\bf a} \bitor [{\bf 0} : p + 1 : {\bf c}_e]$ 
    \EndFor
    \State \Return ${\bf a}$
  \EndFunction
  
  \bigskip
  
  \Function{$N_v$.OrMarks}{$A_+$}
    \State \Return \Call{$N_v$.OrMarks}{$A_+$, $\emptyset$}
  \EndFunction
  
  \bigskip
  
  \Function{$N_v$.FindMarked}{$A$, $i$}
    \ForAll{$e \in A$}
      \State $\mathrm{cand} \gets N_v.\FindMarked{e, i}$ \Comment{Provided by \Cref{lem:neighborhood-ds}}
      \If{$\mathrm{cand} \neq \bot$}
        \State \Return $\mathrm{cand}$
      \EndIf
    \EndFor
    \State \Return $\bot$
  \EndFunction

  \bigskip
  
  \Function{$N^+_v$.FindStrongMarked}{$A$, $i$}
    \ForAll{$e \in A$}
      \State $\mathrm{cand} \gets N^+_v.\FindMarked{e, i + 1}$ \Comment{Provided by \Cref{lem:neighborhood-ds}}
      \If{$\mathrm{cand} \neq \bot$}
        \State \Return $\mathrm{cand}$
      \EndIf
    \EndFor
    \State \Return $\bot$
  \EndFunction
\end{algorithmic}
\end{algorithm}

\begin{algorithm}[H]
\begin{algorithmic}
  \caption{Maintaining bit vectors in the top tree clusters} \label{alg:toptree-find-info}
  \Function{FR.OnCreate}{$C$}
    \State $\totalmarksvec_C \gets {\bf 0}$
    \State $\diagmarksvec^\star_{C,x} \gets {\bf 0}$ for each $x \in \bnd C$
  \EndFunction

  \bigskip

  \Function{FR.GetDiagMarks}{$C$, $v$} \Comment{Computes $\diagmarksvec_{C,v}$}
    \State ${\bf A} \gets \matrixsplice{\diagmarksvec^\star_{C,v}}{\covertop_C + 1}{\bf 0}$ \Comment{$A_{j,i} = \diagcnt^\star_{C,v,j,i} \wedge [j \leq \covertop_C]$}
    \State ${\bf B} \gets \matrixsplice{\bf 0}{\covertop_C + 1}{\diagmarksvec^\star_{C,v}}$ \Comment{$B_{j,i} = \diagcnt^\star_{C,v,j,i} \wedge [j > \covertop_C]$}
    \State ${\bf w} \gets \matrixsum({\bf A})$ \Comment{$w_i = \bigvee_j A_{j,i}$}
    \State ${\bf M} \gets \addvector({\bf B}, {\bf w}, \cover_C)$ \Comment{$M_{j,i} = B_{j,i} \,\operatorname{or}\, (w_i \wedge [j = \cover_C])$}
    \State \Return ${\bf M}$
  \EndFunction

  \bigskip

  \Function{FR.GetTotalMarks}{$C$, $v$} \Comment{Computes $\totalmarksvec_{C,v}$}
    \State ${\bf M} \gets \Call{FR.GetDiagMarks}{C, v}$
    \State ${\bf w} \gets \uppersum({\bf M})$ \Comment{$w_i = \bigvee_{j \geq i} M_{j, i}$}
    \State \Return ${\bf w}$
  \EndFunction

  \bigskip

  \Function{FR.OnMerge}{$C$}
    \If{$C = A \cup B$, $\bnd C \eqqcolon \{v\}$, $\bnd A \eqqcolon \{v, x\}$, $\bnd B \eqqcolon \{x\}$} \Comment{Case \CasePointToPathPoint}
      \State ${\bf c}^x \gets N_x.\OrMarks{\{\firstedge_{A, x}\}} \bitor \bmarksvec_x$
      \State $e \gets \firstedge_{A, v}$
      \State $\totalmarksvec_{A, v} \gets \Call{FR.GetTotalMarks}{A, v}$
      \State $\ismarkedvec_e \gets \totalmarksvec_{A, v} \bitor \vectorsplice{{\bf c}^x}{\cover_A + 1}{\bf 0}$
      \State $N_v.\UpdateMarks{e, \ismarkedvec_e}$
    \ElsIf{$C = A \cup B \cup P$, $\bnd C \eqqcolon \{v, w\}$, $\bnd A \eqqcolon \{v, x\}$, $\bnd B \eqqcolon \{x, w\}$, $\bnd P \eqqcolon \{x\}$} \Comment{Case \CasePath}
      \State ${\bf c}^{AB} \gets N_x.\OrMarks{\{\firstedge_{A, x},\, \firstedge_{B, x}\}} \bitor \bmarksvec_x$
      \State ${\bf c}^{A} \gets N_x.\OrMarks{\{\firstedge_{A, x}\}} \bitor \bmarksvec_x$
      \State ${\bf c}^{B} \gets N_x.\OrMarks{\{\firstedge_{B, x}\}} \bitor \bmarksvec_x$
      \State ${\bf c}^{A \setminus B} \gets N_x.\OrMarks{\{\firstedge_{A, x}\}, \{\firstedge_{B,x}\}}$
      \State ${\bf c}^{B \setminus A} \gets N_x.\OrMarks{\{\firstedge_{B, x}\}, \{\firstedge_{A,x}\}}$
      \State $\ell_{AB} \gets N_x.\Level{\{\firstedge_{A, x}, \firstedge_{B, x}\}}$
      \State $\totalmarksvec_C \gets \totalmarksvec_A \bitor \totalmarksvec_B \bitor {\bf c}^{AB}$
      \ForAll{$(\hat{v}, \hat{A}, \hat{B}) \in \{(v, A, B), (w, B, A)\}$}
        \State $\diagmarksvec_{\hat{A}, \hat{v}} \gets \Call{FR.GetDiagMarks}{\hat{A}, \hat{v}}$
        \State $\diagmarksvec_{\hat{B}, x} \gets \Call{FR.GetDiagMarks}{\hat{B}, x}$
        \State $r \gets \cover_{\hat{A}}$
        \State ${\bf M} \gets \matrixsplice{\diagmarksvec_{\hat{B}, x}}{\min\{r, \ell_{AB}\}}{\diagmarksvec_{\hat{A}, \hat{v}}}$
        \State ${\bf u} \gets \matrixsum(\matrixsplice{\bf 0}{\min\{r, \ell_{AB}\}}{\diagmarksvec_{\hat{B}, x}})$
        \State ${\bf M} \gets \addvector({\bf M}, {\bf u}, \min\{r, \ell_{AB}\})$
        \If{$\ell_{AB} < r$}
          \State ${\bf M} \gets \addvector({\bf M}, {\bf c}^{\hat{A}}, r)$
          \State ${\bf M} \gets \addvector({\bf M}, {\bf c}^{\hat{B} \setminus \hat{A}}, \ell_{AB})$
        \Else
          \State ${\bf M} \gets \addvector({\bf M}, {\bf c}^{AB}, r)$
        \EndIf
        \State $\diagmarksvec^\star_{\hat{A}, \hat{v}} \gets {\bf M}$
      \EndFor
    \EndIf
  \EndFunction
\end{algorithmic}
\end{algorithm}

%\wninline{I think it was never stated that adding, sum and uppersum use the or operation as an analogue of additions for counters. Is it needed?}
%\msinline{I'd assume the reader is smart enough to figure it out themself}
%\wninline{(low prio) I mean, it's not wild to think that 1+1=0 in boolean algebra, I'd prefer having it explicit}
%\msinline{Added some sentence in prelims that additions are replaced with bitwise ors}

\begin{algorithm}[H]
\begin{algorithmic}
  \caption{Marking updates} \label{alg:toptree-find-updates}
  \Function{FR.Mark}{$u$, $i$}
    \State $\TransientExpose{u}$
    \State $\bmarks_{u,i} \gets 1$
    \State $\TransientUnexpose$
  \EndFunction

  \bigskip

  \Function{FR.Unmark}{$u$, $i$}
    \State $\TransientExpose{u}$
    \State $\bmarks_{u,i} \gets 0$
    \State $\TransientUnexpose$
  \EndFunction
\end{algorithmic}
\end{algorithm}

\begin{algorithm}[H]
\begin{algorithmic}
  \caption{Marking queries} \label{alg:toptree-find-queries}
  
  \Function{FR.FindInPath}{$C$, $v$, $w$, $i$}
    \State \textbf{assert} $C = A \cup B \cup P$, $\bnd C = \{v, w\}$, $\bnd A \eqqcolon \{v, x\}$, $\bnd B \eqqcolon \{x, w\}$, $\bnd P \eqqcolon \{x\}$ \Comment{Case \CasePath}
    
    \Clean{$C$}  \Comment{Push down any pending lazy updates in $C$}
    \If{$\totalmarks_{A, i}$}
      \State \Return \Call{FR.FindInPath}{$A$, $v$, $x$, $i$}
    \ElsIf{$\bmarks_{x, i}$}
      \State \Return $(\firstedge_{A, x}, x, i)$
    \ElsIf{$f_A \coloneqq N_x.\FindMarked{\{\firstedge_{A, x}\}, i} \neq \bot$}
      \State \Return $(\firstedge_{A, x}, x, f_A)$
    \ElsIf{$f_B \coloneqq N_x.\FindMarked{\{\firstedge_{B, x}\}, i} \neq \bot$}
      \State \Return $(\firstedge_{B, x}, x, f_B)$
    \Else
      \State \textbf{assert} $\totalmarks_{B, i}$
      \State \Return \Call{FR.FindInPath}{$B$, $x$, $w$, $i$}
    \EndIf
  \EndFunction
  
  \bigskip
  
  \Function{FR.FindFirstReach}{$p$, $q$, $i$}
    \State $R_p, R_q, R_{pq} \gets \TransientExpose{p, q}$
    \State $\mathrm{answer} \gets (\bot, \bot, \bot)$
    \If{$f \coloneqq N_p.\FindMarked{\{\firstedge_{R_{pq}, p}\}, i} \neq \bot$}
      \State $\mathrm{answer} \gets (\firstedge_{R_{pq}, p}, p, f)$
    \ElsIf{$\totalmarks_{R_{pq}, i}$}
      \State $\mathrm{answer} \gets \Call{FR.FindInPath}{R_{pq}, p, q, i}$
    \ElsIf{$f \coloneqq N_q.\FindMarked{\{\firstedge_{R_{pq}, q}\}, i} \neq \bot$}
      \State $\mathrm{answer} \gets (\firstedge_{R_{pq}, q}, q, f)$
    \EndIf
%    \State $\mathrm{answer} \gets$ \Call{FR.FindFirstReachContainer}{$R_p$, $R_q$, $R_{pq}$ $i$} \Comment{Phase $1$}
    \LComment{Now $\mathrm{answer}$ is either negative $(\bot, \bot, \bot)$ or positive (of the form $(ab, c, Y)$, with $Y$ a~container for $y$).}
    %\State $(ab, c, Y) \coloneqq \mathrm{answer}$
    \If {$(ab, c, Y) \coloneqq \mathrm{answer} \neq (\bot, \bot, \bot)$}
      \State $y \gets$ \Call{FR.RefineContainer}{$Y$, $i$} \Comment{Phase $2$ (\Cref{alg:toptree-find-refine})}
      \State $\mathrm{answer} \gets (ab, c, y)$
    \EndIf
    \State $\TransientUnexpose$
    \State \Return $\mathrm{answer}$
  \EndFunction
  
  \bigskip  
  
  \Function{FR.FindStrongReach}{$p$, $q$, $e$, $b$, $i$}
%    \If{$\bmarks_{b, i}$}
%    \State \Return $b$
%    \EndIf
    \State $\TransientExpose{p, q}$ 
    \State $\mathrm{answer} \gets N^+_b.\FindStrongMarked{\{e\}, i}$
    
    \If{$\mathrm{answer}$ is a~container $Y$}
      \State $\mathrm{answer} \gets$ \Call{FR.RefineContainer}{$Y$, $i$} \Comment{Phase $2$ (\Cref{alg:toptree-find-refine})}
    \EndIf
    \State $\TransientUnexpose$
    \State \Return $\mathrm{answer}$
  \EndFunction
\end{algorithmic}
\end{algorithm}
%\wninline{I changed a bit the code for FindFirstReach, check if it's fine.}
%\wninline{I simplified the code for FindStrongReach (and made it closer to its description). Check that it's still fine.}
\begin{algorithm}[H]
\begin{algorithmic}
  \caption{Cluster refinement} \label{alg:toptree-find-refine}
  \Function{FR.RefineClusterEdgeContainer}{$\vec{vw}$, $i$}
    \State \textbf{assert} $\ismarked_{\vec{vw}, i}$
    \State $C \gets \smallestpoint_{\vec{vw}}$
    \State \textbf{assert} $C = A \cup B$, $\bnd C = \{v\}$, $\bnd A \eqqcolon \{v, x\}$, $\bnd B \eqqcolon \{x\}$ \Comment{Case \CasePointToPathPoint}
    \If{$\totalmarks_{A, v, i}$}
      \State \Return \Call{FR.RefinePathClusterContainer}{$A$, $v$, $x$, $i$}
    \EndIf
    \State \textbf{assert} $\cover_A \geq i$
    \If{$\bmarks_{x,i}$}
      \State \Return $x$
    \EndIf
    \State $f \gets N_x.\FindMarked{\{\firstedge_{A, x}\}, i}$
    \State \textbf{assert} $f \neq \bot$
    \State \Return \Call{FR.RefineClusterEdgeContainer}{$f$, $i$}
  \EndFunction

  \bigskip

  \Function{FR.RefinePathClusterContainer}{$C$, $v$, $w$, $i$}
    \State \textbf{assert} $\totalmarks_{C, v, i}$
    \State \textbf{assert} $C = A \cup B \cup P$, $\bnd C = \{v, w\}$, $\bnd A \eqqcolon \{v, x\}$, $\bnd B \eqqcolon \{x, w\}$, $\bnd P \eqqcolon \{x\}$ \Comment{Case \CasePath}
    \State \Clean{$C$}
    \If{$\totalmarks_{A, v, i}$}
      \State \Return \Call{FR.RefinePathClusterContainer}{$A$, $v$, $x$, $i$}
    \EndIf
    \State \textbf{assert} $\cover_A \geq i$
    \If{$\bmarks_{x, i}$}
      \State \Return $x$
    \EndIf
    \If{$N_x.\Level{\firstedge_{A, x}, \firstedge_{B, x}} \geq i$ and $\totalmarks_{B, x, i}$}
      \State \Return \Call{FR.RefinePathClusterContainer}{$B$, $x$, $w$, $i$}
    \EndIf
    \State $f \gets N_x.\FindMarked{\{\firstedge_{A,x}\}, i}$
    \State \textbf{assert} $f \neq \bot$
    \State \Return \Call{FR.RefineClusterEdgeContainer}{$f$, $i$}
  \EndFunction

  \bigskip

  \Function{FR.RefineContainer}{$Y$, $i$}
    \If{$Y$ is the explicit container $y$}
      \State \Return $y$
    \ElsIf{$Y$ is the cluster edge container $\vec{vw}$}
      \State \Return \Call{FR.RefineClusterEdgeContainer}{$\vec{vw}$, $i$}
    \Else
      \State \textbf{assert} $Y$ is the path cluster container $(C, v, w)$
      \State \Return \Call{FR.RefinePathClusterContainer}{$C$, $v$, $w$, $i$}
    \EndIf
  \EndFunction
\end{algorithmic}
\end{algorithm}

%\wninline{Just checking... Can $C$ be a leaf in for example RefineClusterEdgeContainer?}
%\msinline{Not really, thanks to the dummy leaf edges added to each vtx of the tree. (The edge $\vec{vw}$ is a~real edge of $T$, so the cluster contains also a dummy leaf edge attached to $w$.)}
%\wninline{(low prio) Right, I remember that this trick was mentioned at some point, but to be honest I was not really aware of its significance. Is it explained anywhere?}

\end{document}